\documentclass[11pt, a4]{article}     % onecolumn (ditto)
\usepackage{a4wide}
\usepackage{fullpage}
\usepackage{hyperref}
\usepackage{authblk}

\usepackage{graphicx}
\usepackage{amsmath}
\usepackage{amssymb}
\usepackage{tikz}
\usetikzlibrary{math,positioning,shapes,decorations.pathreplacing, calligraphy, backgrounds}
\usepackage{tikz-network}
\usepackage{multirow}
\newcommand{\tnet}{\mbox{$\mathbb{T}$}}
\newcommand{\els}{\mathbb{E}}
\newcommand{\reach}[2]{\mbox{$\mathcal{R}_{#1}(#2)$}}
\newcommand{\optproblem}[2]{\medskip\fbox{\parbox{0.92\textwidth}{\textsc{#1.} #2}}\medskip}
\newcommand{\mrtt}{\textsc{mrtt}}

\newcommand{\otomrtt}{\textsc{o2o-rtt}}

\newcommand{\kotomrtt}{\textsc{$k$-o2o-rtt}}
\newcommand{\ssmrtt}{\textsc{ss-mrtt}}

\newcommand{\tsat}{\textsc{3-sat}}
\makeatletter
\newcommand*{\centerfloat}{%
  \parindent \z@
  \leftskip \z@ \@plus 1fil \@minus \textwidth
  \rightskip\leftskip
  \parfillskip \z@skip}
\makeatother
\newtheorem{theorem}{Theorem}
\newcommand{\qed}{\hfill\raisebox{3.5pt}{ \fbox{}}\medskip}
\newenvironment{proof}{\medskip\noindent\emph{Proof. }\ignorespaces}{\medskip\par\noindent\ignorespacesafterend}

\newtheorem{lemma}{Lemma}
\newtheorem{corollary}{Corollary}
\newtheorem{fact}{Fact}

\newtheorem{claim}{Claim}

\newcommand{\ptime}{\mbox{P}}
\newcommand{\nptime}{\mbox{NP}}
\newcommand{\npotime}{\mbox{NPO}}
\newcommand{\true}{\mbox{\textsc{True}}}
\newcommand{\false}{\mbox{\textsc{False}}}
\newcommand{\tree}[1]{\mathcal{T}_{#1}}
\newcommand{\symtrip}[1]{\mbox{\reflectbox{#1}}}
\newcommand\fixup{\kern-\fontcharic\scriptfont2`\"}

\usepackage{xcolor}

\begin{document}

\title{Maximising Reachability in a Temporal Graph Obtained\\
by Assigning Starting Times to a Collection of Walks
}

%\titlerunning{On The Complexity of Maximising Temporal Reachability via Trip Temporalisation}        % if too long for running head

\author[1]{Filippo Brunelli}
\author[2]{Pierluigi Crescenzi}
\author[1]{Laurent Viennot}


\affil[1]{\small Université de Paris, Inria, CNRS, IRIF, F-75013 Paris, France%suported by ANR project Multimod ANR-17-CE22-0016.
}
\affil[2]{\small Gran Sasso Science Institute, I-67100 L'Aquila, Italy\break
\texttt{filippo.brunelli@inria.fr, pierluigi.crescenzi@gssi.it, laurent.viennot@inria.fr}}


\maketitle

\begin{abstract}
In a temporal graph, each edge appears and can be traversed at specific points in time. In such a graph, temporal reachability of one node from another is naturally captured by the existence of a temporal path where edges appear in chronological order. Inspired by the optimisation of bus/metro/tramway schedules in a public transport network, we consider the problem of turning a collection of walks (called trips) in a directed graph into a temporal graph by assigning a starting time to each trip in order to maximise the reachability among pairs of nodes. Each trip represents the trajectory of a vehicle and its edges must be scheduled one right after another. Setting a starting time to the trip thus forces the appearance time of all its edges. We call such a starting time assignment a trip temporalisation.

We obtain several results about the complexity of maximising reachability via trip temporalisation. Among them, we show that maximising reachability via trip temporalisation is hard to approximate within a factor $\sqrt{n}/12$ in an $n$-vertex digraph, even if we assume that for each pair of nodes, there exists a trip temporalisation connecting them. On the positive side, we show that there must exist a trip temporalisation connecting a constant fraction of all pairs if we additionally assume symmetry, that is, when the collection of trips to be scheduled is such that, for each trip, there is a symmetric trip visiting the same nodes in reverse order.

\end{abstract}

\medskip

\paragraph{Funding:} This work was suported by the French National Research Agency (ANR) through project Multimod with reference number ANR-17-CE22-0016.

% ------------------- \input{body/introduction}
%
\section{Introduction}
\label{sec:intro}

A temporal graph is a graph whose edges are present and can be traversed at specific points in time. They are of particular interest for studying dynamic processes in graphs that evolve with time, would it be information or epidemic propagation, or navigation in a transport network. In that context, the natural notion of connectivity arises from temporal paths, that are paths where edges appear in chronological order. An interesting measure of connectivity is then the temporal reachability that can be defined as the number of pairs of nodes connected by a temporal path.

Inspired by the problem of scheduling buses/metros/tramways in a public transport network, we consider the problem of assigning appearance times to the edges of a digraph (a directed graph) so that the resulting temporal graph maximises temporal reachability. In doing so, we assume that the trajectory of each bus/metro/tramway is fixed and is given by a walk in the digraph so that the edges of that walk must be scheduled one right after another. The input of our problem thus consists in a digraph and a collection of walks in that digraph. These given walks are called trips in reference to the application to transit networks, and the output is an assignment of starting times to these walks that we call a trip temporalisation. The goal is to maximise the reachability of the resulting temporal graph.

The problem of turning an undirected graph into a temporal graph by scheduling independently each edge has already been considered in the gossip setting where a graph is said to be label-connected if each edge can be assigned an appearance time (as a label) so that all pairs of nodes are temporally connected. Most prominently, it has been shown that it is NP-complete to decide whether an undirected graph is label-connected~\cite{Goebel1991}. Surprisingly, approximating maximum temporal reachability seem to have received little attention. More recently, a related minimisation problem has been studied in the context of epidemiology~\cite{Enright2021} where appearance times must be assigned to the edges of a graph (or a digraph) to minimise temporal reachability. The following type of temporal dependency between edges is considered: the edges are partitioned into subsets such that all edges in a subset must be scheduled in parallel (that is, they must be assigned the same appearance time). The authors prove the NP-completeness of the problem and leave as an open question the existence of a constant-factor approximation algorithm. Our trip temporalisation problem takes into account a different type of temporal dependency where edges are grouped into walks that must be sequentially scheduled. To the best of our knowledge, this idea is new although it seems natural in contexts such as transit networks. 

\subsection{Temporal graph model}

A \textit{temporal graph} is a directed multigraph where each edge is labeled with an appearance time and a travel time. It can be represented as a list of \textit{temporal edges} $(u,v,t,\lambda)$, where $u$ and $v$ are two nodes, $t$ is the \textit{appearance time} of the temporal edge, and $\lambda$ is its \textit{travel time}. That is, we can traverse the edge starting from $u$ at time $t$ (and no later) and reaching $v$ at time $t+\lambda$, which is the \textit{arrival time} of the temporal edge.
In the context of public transport networks, such a quadruple is also called a connection~\cite{Dibbelt2018}, and models an elementary movement of a vehicle departing from a stop at time $t$ and reaching the next one at time $t+\lambda$.
A node $v$ is \textit{temporally reachable} from a node $u$ when there exists a \textit{temporal path} from $u$ to $v$, that is a path with the additional temporal constraint that the arrival time of each edge is no later than the appearance time of the next edge so that all edges can be successively traversed one after another.
We then also say that the temporal path \textit{connects} $u$ to $v$, and that the pair $(u,v)$ is \textit{temporally connected}. % when a temporal path connects $u$ to $v$.
We define the \textit{temporal reachability} of a temporal graph as the number of pairs of nodes that are temporally connected.

Figure~\ref{fig:inducedtemporalgraph} shows an example of a temporal graph. For the sake of clarity, travel times are denoted with a ``$+$'' in contrast to appearance times. In that temporal graph, $v_1$ can reach $v_8$ at time 14 by following the temporal path $(v_1,v_2,1,+1),(v_2,v_3,6,+1),(v_3,v_6,7,+1),(v_6,v_7,8,+1),(v_7,v_8,12,+2)$.
However, $(v_5,v_4)$ is not temporally connected as $v_5$ cannot reach $v_7$ before time $12$ while the only temporal edge crossing the cut $\{v_5,v_6,v_7,v_8\}$, $\{v_1,v_2,v_3,v_4\}$ is $(v_7,v_2,9,+1)$ whose appearance time is $9<12$.
Indeed, node $v_1$ can reach all nodes but node $v_5$, while node $v_7$ can reach only nodes $v_2$, $v_7$, and $v_8$. The temporal reachability of this temporal graph is $30$.

\begin{figure}[t]
\centering{\SetVertexStyle[FillColor=white]
\SetEdgeStyle[Color=black]
\begin{tikzpicture}[x=2cm,y=2cm]
  \tikzmath{\xd=3;\yd=3;\x1=0;\y1=0;\x2=\x1+\xd;\y3=\y1+\yd;\x5=\x2+2*\xd;\x6=\x2+\xd;}
  \Vertex[x=\x2,y=\y3,label=$v_{3}$,size=0.5]{x3}
  \Vertex[x=\x1,y=\y3,label=$v_{4}$,size=0.5]{x4}
  \Vertex[x=\x6,y=\y1,label=$v_{6}$,size=0.5]{x6}
  \Vertex[x=\x6,y=\y3,label=$v_{7}$,size=0.5]{x7}
  \Vertex[x=\x5,y=\y3,label=$v_{8}$,size=0.5]{x8}
  \Vertex[x=\x1,y=\y1,label=$v_{1}$,size=0.5]{x1}
  \Vertex[x=\x2,y=\y1,label=$v_{2}$,size=0.5]{x2}
  \Vertex[x=\x5,y=\y1,label=$v_{5}$,size=0.5]{x5}
  \Edge[Direct,label={$1,+1$},position=below](x1)(x2)
  \Edge[Direct,bend=-30,label={$2,+2$},position=right](x2)(x3)
  \Edge[Direct,bend=30,label={$6,+1$},position=left](x2)(x3)
  \Edge[Direct,label={$4,+2$},position=above](x3)(x4)
  \Edge[Direct,label={$10,+1$},position=below](x5)(x6)
  \Edge[Direct,bend=-30,label={$11,+1$},position=right](x6)(x7)
  \Edge[Direct,bend=30,label={$8,+1$},position=left](x6)(x7)
  \Edge[Direct,label={$12,+2$},position=above](x7)(x8)
  \Edge[Direct,label={$7,+1$},distance=.27,position=above](x3)(x6)
  \Edge[Direct,label={$9,+1$},distance=.4,position=above](x7)(x2)
\end{tikzpicture}}
\caption{A temporal graph where each edge is labeled by its appearance time and its travel time respectively. For the sake of clarity, each travel time is indicated with a ``$+$'' as a reminder that the sum of appearance time and travel time of a temporal edge is equal to its arrival time.}
\label{fig:inducedtemporalgraph}
\end{figure}

The temporal constraint required for temporal paths makes the computation of the nodes reachable from a given node a little bit more complicated than in the case of graphs, but still doable in time $\tilde{O}(m)$, where $m$ is the number of temporal edges and the notation $\tilde{O}$ hides poly-logarithmic factors (see~\cite{Dibbelt2018,Wu2016}). The temporal reachability can thus be computed in $\tilde{O}(nm)$ time where $n$ denotes the number of nodes. Note that we do not put any constraint on how long it is possible to wait at a node in-between two temporal edges, as these further constraints can dramatically change the complexity of such temporal path computation~\cite{Casteigts2020Waiting}.

\subsection{The maximum reachability trip temporalisation problem}

Given a weighted directed multigraph $D$, an \textit{edge temporalisation} assigns to each edge $(u,v,\lambda)$ of $D$ an appearance time $t$, making it a temporal edge $(u,v,t,\lambda)$. For example, the weighted directed multigraph $D$ could represent the map of a public transit network where the weight of an edge represents the time needed by a vehicle to travel along that edge. Multiple types of vehicle traversing the same edge can be captured by multiple edges with appropriate weights.
However, in that context, the edges are not ``independent'', in the sense that an appearance time cannot be assigned to an edge independently of the appearance time assigned to other edges. Indeed, the edges are grouped into walks where each walk represents the journey of a vehicle.
We are thus given a collection $\tnet$ of walks in $D$ that we call \textit{trips} to distinguish them from other arbitrary walks in $D$.
When several vehicles travel along the same walk, we assume that a distinct trip is associated to each one of them. We also suppose that the waiting time at a stop is negligible and that scheduling a vehicle amounts to assigning an appearance time to the first edge of the corresponding trip, since all the other appearance times are a consequence of it: indeed, each edge of the trip appears right after the arrival of the previous one. More precisely, given a trip $T=e_1,\ldots,e_k$, where, for $i\in[k]$,\footnote{In the following, for any $h\in\mathbb{N}$, we denote by $[h]$ the set $\{1,2,\ldots,h\}$.} $e_i=(u_i,v_i,\lambda_i)$, assigning a starting time $t$ to $T$ results in the set of temporal edges $f_i=(u_i,v_i,t+\sum_{j=1}^{i-1}\lambda_j,\lambda_i)$, for $i\in[k]$.\footnote{As it is standard, we assume that the summation with no summands evaluates to zero.}
The operation of assigning a starting time to each walk in \tnet{} is called a \textit{trip temporalisation}.

\begin{figure}[t]
\centering{\SetVertexStyle[FillColor=white]
\SetEdgeStyle[Color=black]
\begin{tikzpicture}[x=2cm,y=2cm]
  \tikzmath{\xd=1.5;\yd=1.5;\x1=0;\y1=0;\x2=\x1+\xd;\y3=\y1+\yd;\x5=\x2+2*\xd;\x6=\x2+\xd;}
  \Vertex[x=\x1,y=\y1,label=$v_{1}$,size=0.5]{x1}
  \Vertex[x=\x2,y=\y1,label=$v_{2}$,size=0.5]{x2}
  \Vertex[x=\x2,y=\y3,label=$v_{3}$,size=0.5]{x3}
  \Vertex[x=\x1,y=\y3,label=$v_{4}$,size=0.5]{x4}
  \Vertex[x=\x5,y=\y1,label=$v_{5}$,size=0.5]{x5}
  \Vertex[x=\x6,y=\y1,label=$v_{6}$,size=0.5]{x6}
  \Vertex[x=\x6,y=\y3,label=$v_{7}$,size=0.5]{x7}
  \Vertex[x=\x5,y=\y3,label=$v_{8}$,size=0.5]{x8}
  \Edge[Direct,label={~$+1$~},position=below](x1)(x2)
  \Edge[Direct,bend=-20,label=$+2$,position=right](x2)(x3)
  \Edge[Direct,bend=20,label=$+1$,position=left](x2)(x3)
  \Edge[Direct,label={~$+2$~},position=above](x3)(x4)
  \Edge[Direct,label={~$+1$~},position=below](x5)(x6)
  \Edge[Direct,label=$+1$,position=right](x6)(x7)
  \Edge[Direct,label={~$+2$~},position=above](x7)(x8)
  \Edge[Direct,label={$+1$},distance=0.25,position=above](x3)(x6)
  \Edge[Direct,label={$+1$},distance=0.25,position=above](x7)(x2)
\end{tikzpicture}
\qquad
\begin{tikzpicture}[x=2cm,y=2cm]
  \tikzmath{\xd=1.5;\yd=1.5;\x1=0;\y1=0;\x2=\x1+\xd;\y3=\y1+\yd;\x5=\x2+2*\xd;\x6=\x2+\xd;}
  \Vertex[x=\x2,y=\y3,label=$v_{3}$,size=0.5]{x3}
  \Vertex[x=\x1,y=\y3,label=$v_{4}$,size=0.5]{x4}
  \Vertex[x=\x6,y=\y1,label=$v_{6}$,size=0.5]{x6}
  \Vertex[x=\x6,y=\y3,label=$v_{7}$,size=0.5]{x7}
  \Vertex[x=\x5,y=\y3,label=$v_{8}$,size=0.5]{x8}
  \SetVertexStyle[FillColor=white,LineColor=blue]
  \Vertex[x=\x1,y=\y1,label=$v_{1}$,size=0.5]{x1}
  \SetVertexStyle[FillColor=white,LineColor=green]
  \Vertex[x=\x2,y=\y1,label=$v_{2}$,size=0.5]{x2}
  \SetVertexStyle[FillColor=white,LineColor=red]
  \Vertex[x=\x5,y=\y1,label=$v_{5}$,size=0.5]{x5}
  \Edge[Direct,label={~$+1$~},position=below,color=blue,fontcolor=black](x1)(x2)
  \Edge[Direct,bend=-20,color=blue,label=$+2$,position=right,fontcolor=black](x2)(x3)
  \Edge[Direct,bend=20,color=green,style={dashed},label=$+1$,position=left,fontcolor=black](x2)(x3)
  \Edge[Direct,color=blue,label={~$+2$~},position=above,fontcolor=black,fontcolor=black](x3)(x4)
  \Edge[Direct,color=red,style={dotted},label={~$+1$~},position=below,fontcolor=black](x5)(x6)
  \Edge[Direct,bend=-20,color=red,style={dotted},label=$+1$,position=right,fontcolor=black](x6)(x7)
  \Edge[Direct,bend=20,color=green,style={dashed},label=$+1$,position=left,fontcolor=black](x6)(x7)
  \Edge[Direct,color=red,style={dotted},label={~$+2$~},position=above,fontcolor=black](x7)(x8)
  \Edge[Direct,color=green,style={dashed},label={$+1$},distance=0.25,position=above,fontcolor=black](x3)(x6)
  \Edge[Direct,color=green,style={dashed},label={$+1$},distance=0.25,position=above,fontcolor=black](x7)(x2)
\end{tikzpicture}}
\caption{An example of a weighted directed multigraph $D$ (left) where weights represent travel times of edges, and a collection $\tnet$ of walks on $D$ (right) called ``trips''. Each starting node of a trip has a colored border. The edge $(v_6,v_7,1)$ of $D$ is ``used'' by both the green dashed trip $T_2$ and the red dotted trip $T_3$. The duration of the blue solid trip $T_1$ is the sum of the travel times of its edges which amounts to $+5$, while the duration of the other two trips is $+4$. No trip temporalisation exists such that both $v_8$ is reachable from $v_1$ and $v_4$ is reachable from $v_5$ in the induced temporal graph.}
\label{fig:example}
\end{figure}

The temporal graph \textit{induced} by a trip temporalisation of $\tnet$ is the temporal graph whose node set is the same as the node set of $D$, and whose set of temporal edges is the disjoint union of all the temporal edges resulting from the assignment of the starting times to the walks in $\tnet$.
The \textit{reachability} of a trip temporalisation is the temporal reachability of the temporal graph it induces, that is, the number of pairs of nodes which are temporally connected.
For example, let us consider the weighted directed multigraph $D$ shown in the left part of Figure~\ref{fig:example}, and the following collection $\tnet$ of walks on $D$ (depicted in the right part of the figure): $T_1=(v_1,v_2,+1),(v_2,v_3,+2),(v_3,v_4,+2)$ (blue solid trip), $T_2=(v_2,v_3,+1),(v_3,v_6,+1),(v_6,v_7,+1)$, $(v_7,v_2,+1)$ (green dashed trip), and $T_3=(v_5,v_6,+1),(v_6,v_7,+1),(v_7,v_8,+2)$ (red dotted trip). Note how the edge $(v_6,v_7,+1)$ is ``used'' by two different trips (that is, $T_2$ and $T_3$): this might correspond to two different vehicles travelling through this edge. Note also that there is no trip temporalisation such that both pairs of nodes $(v_1,v_8)$ and $(v_5,v_4)$ are temporally connected in the induced temporal graph. Indeed, if $v_8$ is reachable from $v_1$, then the starting time assigned to $T_1$ has to be smaller than the starting time assigned to $T_3$ as reaching $v_7$ from $v_2$ requires at least $+3$ units of time (using edges of $T_2$), while if $v_4$ is reachable from $v_5$, then the starting time assigned to $T_3$ has to be smaller than the starting time assigned to $T_1$ as reaching $v_3$ from $v_6$ also requires at least $+3$ units of time: these two inequalities cannot be both satisfied. Let us consider the trip temporalisation which assigns to $T_1$ the starting time $1$, to $T_2$ the starting time $6$, and to $T_3$ the starting time $10$ (this trip temporalisation intuitively corresponds to scheduling the three trips one after the other). The temporal graph induced by this trip temporalisation is indeed shown in Figure~\ref{fig:inducedtemporalgraph}
and its reachability is equal to $30$ as already mentioned. On the other hand, it is possible to verify that the trip temporalisation which assigns to $T_1$ the starting time $9$, to $T_2$ the starting time $5$, and to $T_3$ the starting time $1$ induces a temporal graph whose reachability is $32$ (see also Table~\ref{tbl:example} at page~\pageref{tbl:example}). 

The network optimisation problem, called \textsc{Maximum Reachability Trip Temporalisation} (in short, \mrtt{}), that we will analyse in this paper is, hence, the following one: \textit{given a weighted directed multigraph $D$ and a collection \tnet{} of walks on $D$, find a trip temporalisation of \tnet{} which maximises the reachability of the induced temporal graph.}

\subsection{Our results}

Our results are summarised in Table~\ref{tbl:results}, where two other combinatorial problems are also considered. The first decision problem, denoted by \otomrtt{}, is the one-to-one version of the \mrtt{} problem, in which the question is whether a trip temporalisation exists making one given node $t$ temporally reachable from another given node $s$. The second maximisation problem, denoted by \ssmrtt{}, is the single-source version of the \mrtt{} problem, in which the question is to find a trip temporalisation maximising the number of nodes temporally reachable from a given source node $s$.

Note that, although we assume negligible waiting times in the sense that edges of a trip must be scheduled one right after the other, our results can be generalized in a setting where, for each pair of consecutive edges of a trip, a fixed waiting time is imposed. The reason is similar to the fact that we can restrict ourselves to a setting where all travel times are 1 as explained in Section~\ref{sec:preliminaries}.

Quite surprisingly, our first result (see Theorem~\ref{thm:otomrtthard}) shows that \textit{the \otomrtt{} problem is NP-complete}. Using a classical gap technique, we then obtain that if $\ptime\neq\nptime$, then \textit{the \mrtt{} and the \ssmrtt{} problems cannot be approximated within a factor $n^{1-\varepsilon}$ for any $\varepsilon>0$, where $n$ is the number of nodes} (see Theorems~\ref{thm:mrtthard} and~\ref{thm:ssmrtthard}). We also show that the parameterised version of the \otomrtt{} problem with respect to the number $k$ of trips used in the resulting temporal graph, in order to go from $s$ to $t$, can be solved in time $2^{O(k)}m\log\left|\tnet\right|$ where $m=\sum_{T\in \tnet}|T|$ is the sum of trip lengths (see Theorem~\ref{thm:parameterised}).

The above non-approximability results are the main reason for focusing our attention on an  interesting restriction of the \mrtt{} problem, that is, the one in which the collection of trips \tnet{} satisfies the very natural property of being ``temporally'' strongly connected in the following sense. A collection of trips \tnet{} is \textit{strongly temporalisable} if, for each pair of nodes $u$ and $v$, there exists a trip temporalisation of \tnet{} that allows $u$ to (temporally) reach $v$. Note that this requirement is a rather weak one, since we are not asking for a unique trip temporalisation, but for a trip temporalisation for each pair of nodes (indeed, this is a requirement which is satisfied in many applications of temporal graphs). We first show that the strong temporalisability property is not sufficient to get high reachability. To this aim, we prove that \textit{there exists an infinite family of trip collections, all strongly temporalisable, such that any trip temporalisation connects  at most an $O(1/\sqrt{n})$ fraction of all pairs} (see Theorem~\ref{thm:pairscheduleconnected}). By using this construction, we then show that if $\ptime\neq\nptime$, then \textit{the \mrtt{} and the \ssmrtt{} problems cannot be approximated within a factor less than $\sqrt{n}/12$, when restricted to strongly temporalisable trip networks} (see Theorems~\ref{th:inapprox-mrtt} and~\ref{th:inapprox-ssmrtt}).

However, the situation changes if we add another quite natural property of a trip collection, that is, symmetry. A trip collection \tnet{} is \textit{symmetric} if, for each trip $T\in\tnet$, \tnet{} includes also a reverse trip, that is, a trip starting from the last node of $T$, arriving in the first node of $T$, and passing through all the nodes in $T$ in reverse order.\footnote{For what concerns our results, the travel time of a temporal edge and the travel time of its ``corresponding reverse'' temporal edge need not be necessarily equal.} For example, referring to public transport systems, symmetry is almost always respected, since for any bus/metro/tramway trip, there is usually also the same bus/metro/tramway trip in the opposite direction. It is quite easy to show that, if the collection of trips \tnet{} is \textit{symmetric}, then \tnet{} is strongly temporalisable if and only if the weighted directed multigraph $D$ is strongly connected (see Corollary~\ref{cor:connected}). 

We show that \textit{the \mrtt{} problem is NP-hard even if restricted to symmetric and strongly temporalisable collection of trips} (see Theorem~\ref{th:sym-np}). However, our final result shows that, \textit{given a symmetric and strongly temporalisable collection of trips, it is possible to find in polynomial time a trip temporalisation achieving a reachability proportional to the total number of pairs} (see Theorem~\ref{th:symmetric}). This implies the existence of a constant-factor approximation algorithm in the symmetric and strongly temporalisable setting.
It also gives ground to the classical design of public transit networks using symmetric lines.

All our hardness results are proved starting from the \tsat{} problem, which is \nptime-complete~\cite{Garey1979}. Moreover, it is easy to show that the decision and maximisation problems we consider are in \nptime{} or in \npotime{} (that is, the class of \nptime{} optimisation problems~\cite{AusielloMCGPK99}), respectively. Indeed, given an instance of the problem and a trip temporalisation, checking that a node $t$ is reachable from a node $s$ in the induced temporal graph can be done in polynomial time as already mentioned.

\begin{table}[t]
\centerfloat

\begin{tabular}{||p{2cm}||p{12cm}||}
\hline
\multicolumn{2}{||c||}{\textbf{Maximisation problems}}\\
\hline
\multicolumn{1}{||c||}{\textbf{Problem}} & \multicolumn{1}{c||}{\textbf{Complexity}}\\
\hline\hline
% \mret & NP-hard (Theorem~\ref{thm:hardnessedgetemporalisation})\\
% \hline
\mrtt & %\multirow{2}{*}{
       Not approximable within a factor $n^{1-\varepsilon}$ for any $\varepsilon>0$ (Theorem~\ref{thm:mrtthard} \\
\cline{1-1}
% \ssmret & Linear-time solvable (trivial)\\
% \hline
\ssmrtt & and Theorem~\ref{thm:ssmrtthard})\\
\hline\hline
\end{tabular}

\vspace*{0.25cm}

\begin{tabular}{||p{2cm}||p{12cm}||}
\hline
\multicolumn{2}{||c||}{\textbf{Decision problems}}\\
\hline
\multicolumn{1}{||c||}{\textbf{Problem}} & \multicolumn{1}{c||}{\textbf{Complexity}}\\
\hline\hline
% \otomret & Linear-time solvable (trivial)\\
% \hline
\otomrtt & NP-complete (Theorem~\ref{thm:otomrtthard})\\
\hline
\kotomrtt & Solvable in time $2^{O(k)}m\log\left|\tnet{}\right|$ (Theorem~\ref{thm:parameterised})\\
\hline\hline
\end{tabular}

\vspace*{0.25cm}

\begin{tabular}{||p{2cm}||p{4.5cm}|p{7.15cm}||}
\cline{2-3}
\multicolumn{1}{c||}{} & \multicolumn{2}{c||}{\textbf{Property}}\\
\hline
\textbf{Problem} & \textit{Strongly temporalisable} & \textit{Strongly temporalisable and symmetric}\\
\hline\hline
\otomrtt & \multicolumn{2}{c||}{Linear-time solvable (trivial)}\\
\hline
\mrtt & \multirow{2}{*}{\begin{minipage}{4.25cm}Not approximable within a factor less than $\sqrt{n}/12$ (Theorems~\ref{th:inapprox-mrtt} and~\ref{th:inapprox-ssmrtt})\end{minipage}} & \begin{minipage}{6.75cm}NP-hard (Theorem~\ref{th:sym-np}) and $r$-approxi\-mable for some $r>0$ (Theorem~\ref{th:symmetric})\end{minipage}\\
\cline{1-1}\cline{3-3}
\begin{minipage}{2cm}\ssmrtt\end{minipage} &  & \begin{minipage}{6.75cm}Linear-time solvable (consequence of Fact 4)\end{minipage}\\[10pt]
\hline\hline
\end{tabular}

\caption{Our results assuming $\ptime\neq\nptime$ ($n$ denotes the number of nodes, $m=\sum_{T\in \tnet}|T|$ denotes the sum of trip lengths and $k$ denotes the number of trips that can be used in a temporal path). The approximability result is obtained by proving that we can get a high temporal reachability (that is, a temporal reachability proportional to the total number of pairs of nodes). A further result is Theorem~\ref{thm:pairscheduleconnected}, which intuitively states that the strong temporalisability property is not sufficient to get high reachability. The table leaves as the main open problem the question whether the \mrtt{} and the \ssmrtt{} problems are approximable within a sub-linear factor, when restricted to strongly temporalisable collections of trips.}
\label{tbl:results}
\end{table}

\subsection{Related work}
\label{sec:relatedwork}

Temporal graphs (also known as edge-scheduled networks~\cite{Berman1996}, dynamic graphs~\cite{Harary1997}, temporal networks~\cite{Kempe2002Connectivity}, evolving graphs~\cite{BhadraF03}, time-stamped graphs~\cite{Cheng2003}, time-varying graphs~\cite{Casteigts2012}, link streams~\cite{Latapy2018}, or point-availability time-dependent networks~\cite{Brunelli2021}) have received increasing attention over the last two decades~\cite{Holme2012,Holme2013,Masuda2016,Michail2016} as they have applications in a wide variety of contexts, ranging from phone  calls to  contact tracing,  from cattle  exchanges to messaging, from communication traffic to public transport systems. They also have been repeatedly used in order to study classical notions from the field of graph theory (such as degree, path, connectivity, clique, and so on) in a more realistic framework, in which the topology of the graph changes over time.

Optimisation of timetables in a transit network has been studied as an operation research problem (see \cite{CacchianiT2012} for a survey) at a fine-grained level of modeling, taking into account sharing of route segments or tracks, and mixing various objectives such as operation costs or overall user waiting time when the traffic demand is known. We have a higher-level approach that aims at grasping the connectivity of the network. 

Problems similar to the one considered in this paper have already been analysed~\cite{Coro2019,Enright2019,Kempe2002Connectivity,MertziosMCS2013,Mertzios2019,Mertzios2021,Molter2021}. For instance, in~\cite{Enright2019} the authors consider the problem of deleting edges from a given temporal graph in order to reduce its reachability, motivated by the context of epidemiology. Later on, the problem of assigning appearance times of to the edges of a graph in order to minimise the reachability of the resulting temporal graph is studied in~\cite{Enright2021} as already mentioned.
Another closely related work~\cite{DeligkasP2020} studies the problem of minimising the average reachability (as well as other similar objectives which could be interesting goals in the context of transport networks also) in a temporal graph by delaying some edges. Various NP-hardness results as well as a polynomial-time algorithm are given, depending on the type of delay operations that are permitted. The authors leave as open the complexity of maximising reachability which is similar to our goal.

As far as we know, the trip temporalisation problem has never been studied before in such generality. The topic of temporalising edges to increase reachability is connected to gossip and broadcasting protocols (see~\cite{HedetniemiHL1988} for a survey). However, apart from the already mentioned~\cite{Goebel1991} where the undirected setting and the independence of edges makes the problem different from here, the objective is usually different, that is, minimising the time for a message to reach all nodes.
From an application point of view, our trip temporalisation is closely related to the last train timetable synchronization problem for which several algorithms (based on mixed-integer programming and genetic approaches) are proposed in~\cite{Chen2019,Zhou2019}.

\section{Preliminary definitions and results}
\label{sec:preliminaries}

A \textit{weighted directed multigraph} (or just \textit{weighted multidigraph}) $D=(V,E)$ consists of a set $V$ of \textit{nodes} and a set $E\subseteq V\times V\times\mathbb{R}^{+}$ of \textit{(weighted) edges} (for each edge $(u,v,w)$, we say that $u$ is the \textit{tail}, $v$ is the \textit{head}, and $w$ is the weight of the edge). A \textit{walk} $T$ in a weighted multidigraph $D=(V,E)$ from a node $u$ to a node $v$ is a sequence $e_1,\ldots,e_{k}$ of edges in $E$ such that, for each $i$ with $i\in[k-1]$, the head of $e_i$ is equal to the tail of $e_{i+1}$, $u$ is the tail of $e_1$, and $v$ is the head of $e_k$. The \textit{duration} $\delta(T)$ is defined as the sum of the weights of all the edges of $T$. A node $v$ is said to be \textit{reachable} from a node $u$ in $D$, if there exists a walk in $D$ from $u$ to $v$ (in the following, we will assume that a node is reachable from itself).

A \textit{trip network} is a weighted multidigraph $D=(V,E)$ (also called the \textit{underlying multidigraph} of the trip network) along with a collection $\tnet=\{T_1,\ldots,T_{\left|\tnet\right|}\}$ of walks in $D$.\footnote{Here, the term ``collection'' is used to mean a multiset, that is, a set in which order is ignored but multiplicity is significant. The cardinality $|A|$ of a multiset $A$ denotes the sum of the multiplicities of the distinct elements in $A$: in our case, $|\tnet|$ denotes the number of (not necessarily distinct) walks in $\tnet$.} The walks in \tnet{} are also called \textit{trips} to distinguish them from other arbitrary walks in $D$. In the following, without loss of generality, we will assume that any node in $V$ and any edge in $E$ appears in at least one trip in $\tnet$. Note that the disjoint union of the trips in \tnet{} defines a weighted multidigraph $M$ that we call the \textit{induced multidigraph} of $(D,\tnet)$ (we assume that the edges of $M$ have an additional label specifying which trip they belong to, and which is represented in the figures by means of different line colors and styles).

A \textit{temporal graph} $G=(V,\els)$ consists of a set $V$ of \textit{nodes} and a set $\els\subseteq V\times V\times\mathbb{R}\times\mathbb{R}^{+}$ of \textit{temporal edges}. Given a temporal edge $e=(u,v,t,\lambda)$, we say that $u$ is the \textit{tail}, $v$ is the \textit{head}, $t$ is the \textit{appearance} time, $\lambda >0$ is the \textit{travel} time, and $t+\lambda$ is the \textit{arrival} time of $e$. A \textit{temporal path} from a node $u$ to a node $v$ in a temporal graph $G$ is a sequence of temporal edges $e_1,\ldots,e_k$ such that the tail of $e_1$ is $u$, the head of $e_k$ is $v$, and, for each $i\in[k-1]$, the tail of $e_{i+1}$ is equal to the head of $e_i$ and the appearance time of $e_{i+1}$ is at least equal to the arrival time of $e_i$ (as travel times are strictly positive, the path is strict in the sense that the appearance time of $e_{i+1}$ is greater than the appearance time of $e_i$, for each $i\in[k-1]$). A node $v$ is \textit{temporally reachable} from a node $u$ if there exists a temporal path from $u$ to $v$ (in the following, we will assume that a node is temporally reachable from itself). The set of nodes temporally reachable from a node $u$ in $G$ is denoted as $\reach{G}{u}$. The \textit{temporal reachability} of a node $u$ in $G$ is defined as $\left|\reach{G}{u}\right|$, while the \textit{temporal reachability} of $G$ is defined as $\sum_{u\in V}\left|\reach{G}{u}\right|$.

Given a trip network $(D,\tnet)$, a \textit{temporalisation} $\tau$ of the trip network assigns a real number $\tau(T)$ to each trip $T$ in $\tnet$, indicating the \textit{starting time} of $T$. Such a temporalisation induces the temporal graph $G[D,\tnet,\tau]=(V,\els)$ defined as follows. For each $T=e_1,\ldots,e_{k}$ in $\tnet$, with $e_i=(u_i,v_i,w_i)$, $\els$ contains the temporal edges $(u_i,v_i,\tau(T)+\sum_{j=1}^{i-1}w_{j},w_{i})$: we say that these temporal edges are \textit{induced} by $T$ (with respect to the temporalisation $\tau$). A node $v$ is said to be $\tau$-reachable from a node $u$ if $v\in\reach{G[D,\tnet,\tau]}{u}$. The \textit{$\tau$-reachability} of a node $u$ is the temporal reachability of $u$ in $G[D,\tnet,\tau]$, and the \textit{$\tau$-reachability} of the trip network is the temporal reachability of $G[D,\tnet,\tau]$. Our main optimisation problem is the following one.

\optproblem{Maximum Reachability Trip Temporalisation (\mrtt)}{Given a trip network $(D,\tnet)$, find a temporalisation $\tau$ of the trip network which maximises its $\tau$-reachability.}

We will also study the restriction of the \mrtt{} problem to the case in which, for each pair of nodes, there exists a temporalisation allowing us to reach one from the other. More precisely, given a trip network $(D,\tnet)$ and two nodes $s$ and $t$, $(D,\tnet)$ is said to be \textit{$(s,t)$-temporalisable} if there exists a temporalisation $\tau$ of $(D,\tnet)$ such that $t$ is $\tau$-reachable from $s$, and it is said to be \textit{strongly temporalisable} if, for any two nodes $s$ and $t$, $(D,\tnet)$ is $(s,t)$-temporalisable. Moreover, we will also consider symmetric trip networks in the following sense. We say that a trip network $(D,\tnet)$ is \textit{symmetric} if all trips in \tnet{} can be grouped into disjoint pairs $(T,\symtrip{$T$})$ such that \symtrip{$T$} is the reverse of $T$ ($T$ and \symtrip{$T$} are two distinct trips in \tnet{}): \symtrip{$T$} visits the same nodes as $T$, but in reverse order. 

We will often refer to a particular kind of temporalisation. Given a trip network $(D,\tnet)$, a \textit{schedule} of the trip network is an ordering of the trips in $\tnet$. Note that a schedule $S$ immediately induces a temporalisation $\tau_{S}$ of the trip network defined as follows. If $S=T_1,\ldots,T_{\left|\tnet\right|}$, then $\tau_S(T_1)=0$ and $\tau_S(T_{i+1})=\sum_{j=1}^{i}\delta(T_{j})$, for $i\in[|\tnet|-1]$. A node $v$ is said to be $S$-reachable from a node $u$ if it is $\tau_S$-reachable. The \textit{$S$-reachability} of the trip network is defined as its $\tau_S$-reachability (see Table~\ref{tbl:example} where, for any possible schedule $S$, we indicate the $S$-reachability of the trip network shown in Figure~\ref{fig:example}).

\begin{table}[t]
\centerfloat
\begin{small}
\begin{tabular}{||c||c|c|c|c|c|c||}
\cline{2-7}
\multicolumn{1}{c||}{} & $T_1,T_2,T_3$ & $T_1,T_3,T_2$ & $T_2,T_1,T_3$ & $T_2,T_3,T_1$ & $T_3,T_1,T_2$ & $T_3,T_2,T_1$\\
\hline\hline
$v_1$ & $V\setminus\{v_5\}$ & $V\setminus\{v_5,v_8\}$ & $\{v_1,v_2,v_3,v_4\}$ & $\{v_1,v_2,v_3,v_4\}$ & $V\setminus\{v_5,v_8\}$ & $\{v_1,v_2,v_3,v_4\}$\\
\hline
$v_2$ & $V\setminus\{v_1,v_5\}$ & $V\setminus\{v_1,v_5,v_8\}$ & $V\setminus\{v_1,v_5\}$ & $V\setminus\{v_1,v_5\}$ & $V\setminus\{v_1,v_5,v_8\}$ & $V\setminus\{v_1,v_5,v_8\}$\\
\hline
$v_3$ & $V\setminus\{v_1,v_5\}$ & $V\setminus\{v_1,v_5,v_8\}$ & $V\setminus\{v_1,v_5\}$ & $V\setminus\{v_1,v_5\}$ & $V\setminus\{v_1,v_5,v_8\}$ & $V\setminus\{v_1,v_5,v_8\}$\\
\hline
$v_5$ & $\{v_5,v_6,v_7,v_8\}$ & $V\setminus\{v_1,v_3,v_4\}$ & $\{v_5,v_6,v_7,v_8\}$ & $\{v_5,v_6,v_7,v_8\}$ & $V\setminus\{v_1,v_3,v_4\}$ & $V\setminus\{v_1\}$\\
\hline
$v_6$ & $\{v_2,v_6,v_7,v_8\}$ & $\{v_2,v_6,v_7,v_8\}$ & $\{v_2,v_6,v_7,v_8\}$ & $V\setminus\{v_1,v_5\}$ & $\{v_2,v_6,v_7,v_8\}$ & $V\setminus\{v_1,v_5\}$\\
\hline
$v_7$ & $\{v_2,v_7,v_8\}$ & $\{v_2,v_7,v_8\}$ & $V\setminus\{v_1,v_5,v_6\}$ & $V\setminus\{v_1,v_5,v_6\}$ & $\{v_2,v_7,v_8\}$ & $V\setminus\{v_1,v_6,v_5\}$\\
\hline\hline
\multicolumn{1}{c||}{} & $30$ & $28$ & $29$ & $31$ & $28$ & $32$\\
\cline{2-7}
\end{tabular}
\end{small}
\caption{The possible schedules of the trip network of Figure~\ref{fig:example}. For each schedule $S$ and for each source node $v$, the corresponding cell shows the set of nodes $S$-reachable from $v$ (the last row shows the value of the $S$-reachability). Note that in the underling multidigraph of the trip network the number of pairs of nodes $u$ and $v$ such that $v$ is reachable from $u$ is equal to $38$.}
\label{tbl:example}
\end{table}

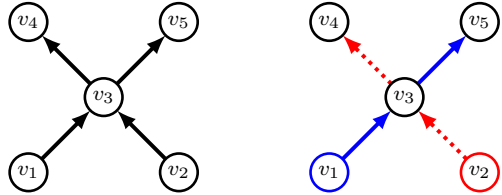
\begin{figure}[b]
\centering{\SetVertexStyle[FillColor=white]
\SetEdgeStyle[Color=black]
\begin{tikzpicture}[x=2cm,y=2cm]
  \tikzmath{\xd=1;\yd=1;\x1=0;\x2=\x1+\xd;\x3=\x1+2*\xd;\y1=0;\y2=\y1+\yd;\y3=\y1+2*\yd;}
  \Vertex[x=\x1,y=\y1,label=$v_{1}$,size=0.5]{x1}
  \Vertex[x=\x3,y=\y1,label=$v_{2}$,size=0.5]{x2}
  \Vertex[x=\x2,y=\y2,label=$v_{3}$,size=0.5]{x3}
  \Vertex[x=\x1,y=\y3,label=$v_{4}$,size=0.5]{x4}
  \Vertex[x=\x3,y=\y3,label=$v_{5}$,size=0.5]{x5}
  \Edge[Direct](x1)(x3)
  \Edge[Direct](x3)(x5)
  \Edge[Direct](x2)(x3)
  \Edge[Direct](x3)(x4)
\begin{scope}[xshift=4cm]
  \Vertex[x=\x2,y=\y2,label=$v_{3}$,size=0.5]{x3}
  \Vertex[x=\x1,y=\y3,label=$v_{4}$,size=0.5]{x4}
  \Vertex[x=\x3,y=\y3,label=$v_{5}$,size=0.5]{x5}
  \SetVertexStyle[FillColor=white,LineColor=blue]
  \Vertex[x=\x1,y=\y1,label=$v_{1}$,size=0.5]{x1}
  \SetVertexStyle[FillColor=white,LineColor=red]
  \Vertex[x=\x3,y=\y1,label=$v_{2}$,size=0.5]{x2}
  \Edge[Direct,color=blue](x1)(x3)
  \Edge[Direct,color=blue](x3)(x5)
  \Edge[Direct,color=red,style={dotted}](x2)(x3)
  \Edge[Direct,color=red,style={dotted}](x3)(x4)
\end{scope}
\end{tikzpicture}}
\caption{An example of a trip network $(D,\tnet)$, where the underlying digraph $D$ is depicted on the left (all edges have weight 1, so that $D$ is a simple digraph) and $\tnet$ (depicted in the induced multidigraph on the right) contains the trips $T_1=\langle v_1,v_3,v_5\rangle$ (blue solid trip) and $T_2=\langle v_2,v_3,v_4\rangle$ (red dotted trip), such that the maximum $\tau$-reachability obtainable through a temporalisation is higher than the maximum $S$-reachability obtainable through a schedule. Indeed, the two possible schedules both achieve a reachability equal to $12$, while a temporalisation that assign the same starting time to $T_1$ and $T_2$ achieves a reachability equal to $13$ (which is also the number of pairs of nodes $u$ and $v$ such that $v$ is reachable from $u$ in the underlying digraph).}
\label{fig:taubetterthans}
\end{figure}

\begin{fact}\label{fact:schedule2}
Let $(D,\tnet)$ be a trip network and $S$ be a schedule of $(D,\tnet)$. Let $C$ be a weighted multidigraph obtained starting from $D$ by arbitrarily modifying only the weights of the edges of $D$. The $S$-reachability of $(D,\tnet)$ is equal to the $S$-reachability of $(C,\tnet)$.
\end{fact}

\begin{proof}
The fact simply follows from the fact that a temporal path in $G[D,\tnet,\tau_S]$ is also a temporal path in $G[C,\tnet,\tau_S]$, since edges of different trips cannot be interleaved inside a temporal path obtained through a schedule, where all edges of a trip $T$ are assigned smaller appearance times than all edges of the trips scheduled after $T$.\qed
\end{proof}

For the sake of simplicity and without loss of generality, in the following we will present our results by referring to trip networks in which the weight of all edges are equal to $1$: indeed, as a consequence of Fact~\ref{fact:schedule2}, \textit{all} our results will apply to general trip networks as well (either because they are hardness results or because the lower bounds on the temporal reachability are obtained by referring to schedules). Under this assumption, a trip $T_i=(u_0,u_1,1),\ldots,(u_{k-1},u_k,1)$ will also be indicated as $T_i=\langle u_0,\ldots,u_k\rangle$. Note, however, that, in general, the maximum $\tau$-reachability obtainable through a temporalisation can be higher than the maximum $S$-reachability obtainable through a schedule (see, for example, Figure~\ref{fig:taubetterthans}), and that the presence of weights can, in general, increase the maximum $\tau$-reachability of a trip network (see, for example, Figure~\ref{fig:weightbetterthannoweight}).

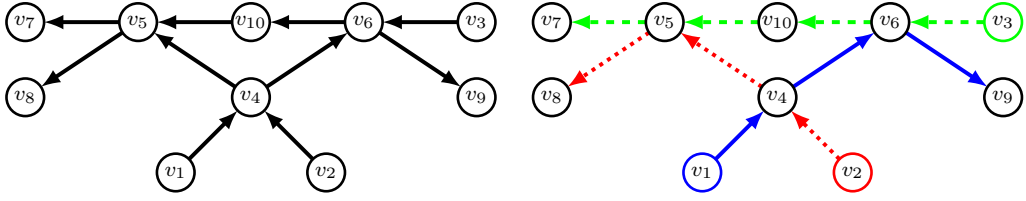
\begin{figure}[t]
\centerfloat
\SetVertexStyle[FillColor=white]
\SetEdgeStyle[Color=black]
\begin{tikzpicture}[x=2cm,y=2cm]
  \tikzmath{\xd=1;\yd=1;\x1=0;\x2=\x1+2*\xd;\x3=\x1+\xd;\x6=\x1-2*\xd;\x4=\x1-0.5*\xd;\x5=\x2+0.5*\xd;\x8=\x2+2*\xd;\y1=0;\y3=\y1+\yd;\y4=\y1+2*\yd;}
  \Vertex[x=\x1,y=\y1,label=$v_{1}$,size=0.5]{x1}
  \Vertex[x=\x2,y=\y1,label=$v_{2}$,size=0.5]{x2}
  \Vertex[x=\x3,y=\y3,label=$v_{4}$,size=0.5]{x3}
  \Vertex[x=\x4,y=\y4,label=$v_{5}$,size=0.5]{x4}
  \Vertex[x=\x3,y=\y4,label=$v_{10}$,size=0.5]{x10}
  \Vertex[x=\x5,y=\y4,label=$v_{6}$,size=0.5]{x5}
  \Vertex[x=\x6,y=\y4,label=$v_{7}$,size=0.5]{x6}
  \Vertex[x=\x6,y=\y3,label=$v_{8}$,size=0.5]{x7}
  \Vertex[x=\x8,y=\y4,label=$v_{3}$,size=0.5]{x8}
  \Vertex[x=\x8,y=\y3,label=$v_{9}$,size=0.5]{x9}
  \Edge[Direct](x1)(x3)
  \Edge[Direct](x3)(x5)
  \Edge[Direct](x2)(x3)
  \Edge[Direct](x3)(x4)
  \Edge[Direct](x5)(x10)
  \Edge[Direct](x10)(x4)
  \Edge[Direct](x4)(x6)
  \Edge[Direct](x4)(x7)
  \Edge[Direct](x8)(x5)
  \Edge[Direct](x5)(x9)
\begin{scope}[xshift=7cm]
  \Vertex[x=\x3,y=\y3,label=$v_{4}$,size=0.5]{x3}
  \Vertex[x=\x4,y=\y4,label=$v_{5}$,size=0.5]{x4}
  \Vertex[x=\x3,y=\y4,label=$v_{10}$,size=0.5]{x10}
  \Vertex[x=\x5,y=\y4,label=$v_{6}$,size=0.5]{x5}
  \Vertex[x=\x6,y=\y4,label=$v_{7}$,size=0.5]{x6}
  \Vertex[x=\x6,y=\y3,label=$v_{8}$,size=0.5]{x7}
  \Vertex[x=\x8,y=\y3,label=$v_{9}$,size=0.5]{x9}
  \SetVertexStyle[FillColor=white,LineColor=blue]
  \Vertex[x=\x1,y=\y1,label=$v_{1}$,size=0.5]{x1}
  \SetVertexStyle[FillColor=white,LineColor=red]
  \Vertex[x=\x2,y=\y1,label=$v_{2}$,size=0.5]{x2}
  \SetVertexStyle[FillColor=white,LineColor=green]
  \Vertex[x=\x8,y=\y4,label=$v_{3}$,size=0.5]{x8}
  \Edge[Direct,color=blue](x1)(x3)
  \Edge[Direct,color=blue](x3)(x5)
  \Edge[Direct,color=red,style={dotted}](x2)(x3)
  \Edge[Direct,color=red,style={dotted}](x3)(x4)
  \Edge[Direct,color=green,style={dashed}](x5)(x10)
  \Edge[Direct,color=green,style={dashed}](x10)(x4)
  \Edge[Direct,color=green,style={dashed}](x4)(x6)
  \Edge[Direct,color=red,style={dotted}](x4)(x7)
  \Edge[Direct,color=green,style={dashed}](x8)(x5)
  \Edge[Direct,color=blue](x5)(x9)
\end{scope}
\end{tikzpicture}
\caption{An example of a trip network $(D,\tnet)$, where the underlying digraph $D$ is depicted on the left and $\tnet$ contains the three trips (depicted on the right) $T_1$ (blue solid trip), $T_2$ (green dashed trip), and $T_3$ (red dotted trip), such that the presence of weights can increase the maximum $\tau$-reachability obtainable through a temporalisation. Indeed, if all weights are equal to $1$, no temporalisation $\tau$ can make the four nodes $v_{7}$, $v_{8}$, $v_{9}$, and $v_{10}$ all $\tau$-reachable from the three nodes $v_{1}$, $v_{2}$, and $v_{3}$: hence, for any temporalisation $\tau$, the $\tau$-reachability is less than the number $R$ of pairs of nodes $u$ and $v$ such that $v$ is reachable from $u$ in $D$. On the contrary, if the edge from $v_{4}$ to $v_{5}$ has weight $3$, then there exists a temporalisation whose reachability is equal to $R$ (such a temporalisation assigns $1$ to the trips $T_{1}$ and $T_{2}$, and $2$ to $T_{3}$).}
\label{fig:weightbetterthannoweight}
\end{figure}

\section{The maximum reachability trip temporalisation problem}
\label{sec:mrtt}

We first consider the following one-to-one version of the \mrtt{} problem, called \textsc{One-To-One Reachability Trip Temporalisation} (in short, \otomrtt): given a trip network $(D,\tnet)$ and two nodes $s$ and $t$, is $(D,\tnet)$ $(s,t)$-temporalisable? Quite surprisingly, even this restricted version of the \mrtt{} problem seems to be difficult to solve in polynomial time.

\begin{theorem}\label{thm:otomrtthard}
The \otomrtt{} problem is \nptime-complete.
\end{theorem}

\begin{proof}
We reduce in polynomial time \tsat{} to \otomrtt. Let us consider a \tsat{} formula $\Phi$, with $n$ variables $x_1, \dots, x_n$ and $m$ clauses $c_1, \dots, c_m$. We first define the directed graph $D=(V,E)$ as the union of the following gadgets.

\medskip
\noindent\textbf{Intermediate and final nodes}. $V$ contains two nodes $v_{n+1}$ and $w_{m+1}$.

\medskip
\noindent\textbf{Variable gadgets} (see Figure~\ref{fig:theorem1}(a)). For each variable $x_i$ with $i\in[n]$, let $p_i$ be the number of clauses that contain the literal $x_i$, and $n_i$ the number of clauses that contain the literal $\neg x_i$ (without loss of generality, we may assume that both $p_i$ and $n_i$ are positive numbers). Then, $V$ contains the following $p_i+n_i+1$ nodes: $v_i,f_i^{1},\ldots,f_i^{p_i},t_i^{1},\ldots,t_i^{n_i}$. Moreover, $E$ contains the following $p_i+n_i+2$ directed edges: $(v_i,f_i^{1})$, $(f_i^{h},f_i^{h+1})$ for $h\in[p_i-1]$, $(f_i^{p_i},v_{i+1})$, $(v_i,t_i^{1})$, $(t_i^{h},t_i^{h+1})$ for $h\in[n_i-1]$, and $(t_i^{n_i},v_{i+1})$.

\medskip
\noindent\textbf{Clause gadgets} (see Figure~\ref{fig:theorem1}(b)). For each clause $c_j$ with $j\in[m]$, $V$ contains the following four nodes: $w_j,l_j^{1},l_j^{2},l_j^{3}$. Moreover, $E$ contains the following six edges: $(w_j,l_j^{h})$ and $(l_j^{h},w_{j+1})$, for $h\in[3]$.

\medskip
\noindent\textbf{Variable-clause edge}. $E$ contains the edge $(v_{n+1},w_1)$.

\medskip
\noindent\textbf{Clause-variable edges} (see Figure~\ref{fig:theorem1}(c)). For each clause $c_j$ with $j\in[m]$, for each variable $x_i$ with $i\in[n]$, for $h\in[3]$, and for $k\in[n_i]$, $E$ contains the edge $(l_j^{h},t_i^{k})$ if the $h$-th literal of $c_j$ is $\neg x_i$ and $c_j$ is the $k$-th clause in which the literal $\neg x_i$ occurs. Analogously, for each clause $c_j$ with $j\in[m]$, for each variable $x_i$ with $i\in[n]$, for $h\in[3]$, and for $k\in[p_i]$, $E$ contains the edge $(l_j^{h},f_i^{k})$ if the $h$-th literal of $c_j$ is $x_i$ and $c_j$ is the $k$-th clause in which the literal $x_i$ occurs.

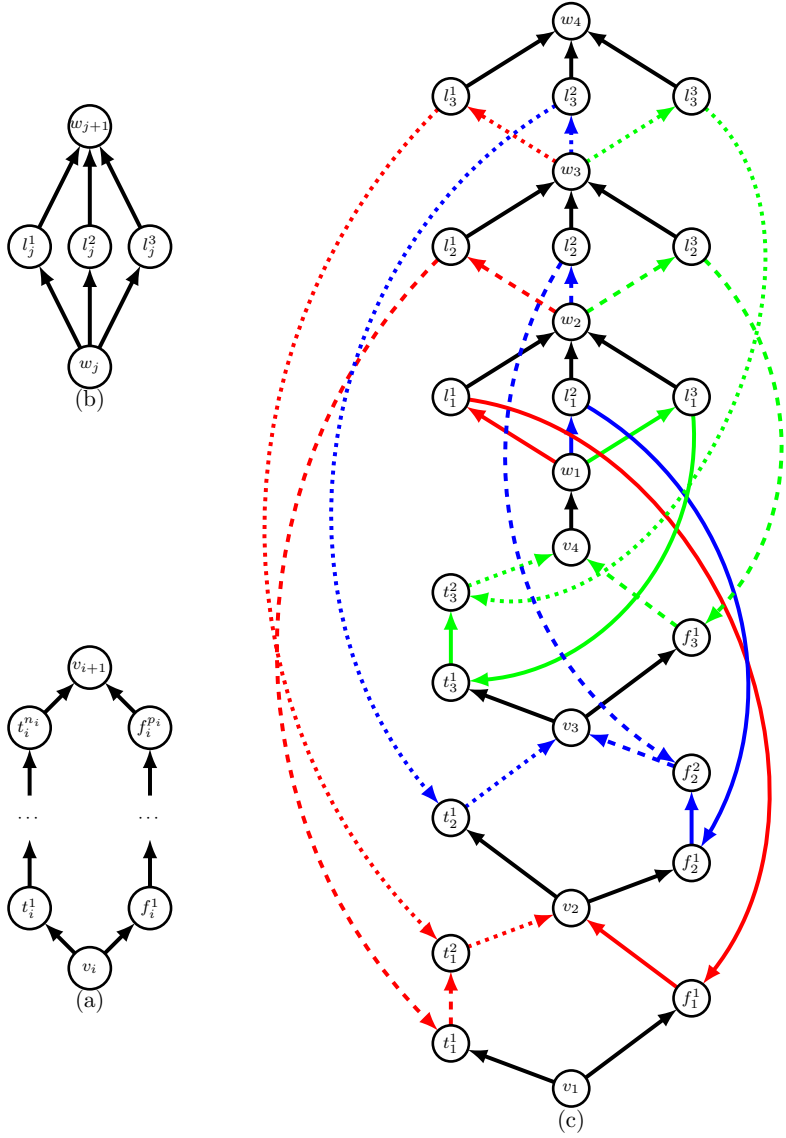
\begin{figure}[th]
    \centering{\SetVertexStyle[FillColor=white]
\SetEdgeStyle[Color=black]
\begin{tikzpicture}[scale=0.8, every node/.style={scale=0.8}]
\node at (0,1.45) {(a)};
\node at (0,11.45) {(b)};
\node at (8,-0.55) {(c)};
\begin{scope}[yshift=2cm]
  \Vertex[y=0,x=0,label=$v_{i}$,size=0.7]{vi}
  \Vertex[y=1,x=1,label=$f_i^1$,size=0.7]{fif}
  \Vertex[y=2.5,x=1,label=$\cdots$,style={color=white},size=0.7]{fim}
  \Vertex[y=4,x=1,label=$f_i^{p_i}$,size=0.7]{fil}
  \Vertex[y=1,x=-1,label=$t_i^1$,size=0.7]{tif}
  \Vertex[y=2.5,x=-1,label=$\cdots$,style={color=white},size=0.7]{tim}
  \Vertex[y=4,x=-1,label=$t_i^{n_i}$,size=0.7]{til}
  \Vertex[y=5,x=0,label=$v_{i+1}$,size=0.7]{vip1}
  \Edge[Direct](vi)(tif)
  \Edge[Direct](tif)(tim)
  \Edge[Direct](tim)(til)
  \Edge[Direct](vi)(fif)
  \Edge[Direct](fif)(fim)
  \Edge[Direct](fim)(fil)
  \Edge[Direct](til)(vip1)
  \Edge[Direct](fil)(vip1)
\end{scope}
\begin{scope}[yshift=12cm]
  \Vertex[y=0,x=0,label=$w_{j}$,size=0.7]{wj}
  \Vertex[y=2,x=1,label=$l_j^3$,size=0.7]{lj3}
  \Vertex[y=2,x=0,label=$l_j^2$,size=0.7]{lj2}
  \Vertex[y=2,x=-1,label=$l_j^1$,size=0.7]{lj1}
  \Vertex[y=4,x=0,label=$w_{j+1}$,size=0.7]{wjp1}
  \Edge[Direct](wj)(lj1)
  \Edge[Direct](wj)(lj2)
  \Edge[Direct](wj)(lj3)
  \Edge[Direct](lj1)(wjp1)
  \Edge[Direct](lj2)(wjp1)
  \Edge[Direct](lj3)(wjp1)
\end{scope}
\begin{scope}[xshift=8cm]
  \Vertex[y=0,x=0,label=$v_{1}$]{v1}
  \Vertex[y=1.5,x=2,label=$f_1^1$]{f11}
  \Vertex[y=0.75,x=-2,label=$t_1^1$]{t11}
  \Vertex[y=2.25,x=-2,label=$t_1^2$]{t12}
  \Vertex[y=3,x=0,label=$v_2$]{v2}
  \Vertex[y=3.75,x=2,label=$f_2^1$]{f21}
  \Vertex[y=5.25,x=2,label=$f_2^2$]{f22}
  \Vertex[y=4.5,x=-2,label=$t_2^1$]{t21}
  \Vertex[y=6,x=0,label=$v_3$]{v3}
  \Vertex[y=7.5,x=2,label=$f_3^1$]{f31}
  \Vertex[y=6.75,x=-2,label=$t_3^1$]{t31}
  \Vertex[y=8.25,x=-2,label=$t_3^2$]{t32}
  \Vertex[y=9,x=0,label=$v_4$]{v4}
  \Vertex[y=10.25,x=0,label=$w_1$]{w1}
  \Vertex[y=11.5,x=2,label=$l_1^3$]{l13}
  \Vertex[y=11.5,x=0,label=$l_1^2$]{l12}
  \Vertex[y=11.5,x=-2,label=$l_1^1$]{l11}
  \Vertex[y=12.75,x=0,label=$w_2$]{w2}
  \Vertex[y=14,x=2,label=$l_2^3$]{l23}
  \Vertex[y=14,x=0,label=$l_2^2$]{l22}
  \Vertex[y=14,x=-2,label=$l_2^1$]{l21}
  \Vertex[y=15.25,x=0,label=$w_3$]{w3}
  \Vertex[y=16.5,x=2,label=$l_3^3$]{l33}
  \Vertex[y=16.5,x=0,label=$l_3^2$]{l32}
  \Vertex[y=16.5,x=-2,label=$l_3^1$]{l31}
  \Vertex[y=17.75,x=0,label=$w_4$]{w4}
  \Edge[Direct](v1)(f11)
  \Edge[Direct](v1)(t11)
  \Edge[Direct,color=red,style={dashed}](t11)(t12)
  \Edge[Direct,color=red](f11)(v2)
  \Edge[Direct,color=red,style={dotted}](t12)(v2)
  \Edge[Direct](v2)(f21)
  \Edge[Direct](v2)(t21)
  \Edge[Direct,color=blue](f21)(f22)
  \Edge[Direct,color=blue,style={dashed}](f22)(v3)
  \Edge[Direct,color=blue,style={dotted}](t21)(v3)
  \Edge[Direct](v3)(f31)
  \Edge[Direct](v3)(t31)
  \Edge[Direct,color=green](t31)(t32)
  \Edge[Direct,color=green,style={dashed}](f31)(v4)
  \Edge[Direct,color=green,style={dotted}](t32)(v4)
  \Edge[Direct](v4)(w1)
  \Edge[Direct,color=red](w1)(l11)
  \Edge[Direct,color=blue](w1)(l12)
  \Edge[Direct,color=green](w1)(l13)
  \Edge[Direct](l11)(w2)
  \Edge[Direct](l12)(w2)
  \Edge[Direct](l13)(w2)
  \Edge[Direct,color=red,style={dashed}](w2)(l21)
  \Edge[Direct,color=blue,style={dashed}](w2)(l22)
  \Edge[Direct,color=green,style={dashed}](w2)(l23)
  \Edge[Direct](l21)(w3)
  \Edge[Direct](l22)(w3)
  \Edge[Direct](l23)(w3)
  \Edge[Direct,color=red,style={dotted}](w3)(l31)
  \Edge[Direct,color=blue,style={dotted}](w3)(l32)
  \Edge[Direct,color=green,style={dotted}](w3)(l33)
  \Edge[Direct](l31)(w4)
  \Edge[Direct](l32)(w4)
  \Edge[Direct](l33)(w4)
  \Edge[Direct,bend=60,color=red](l11)(f11)
  \Edge[Direct,bend=45,color=blue](l12)(f21)
  \Edge[Direct,bend=45,color=green](l13)(t31)
  \Edge[Direct,bend=-45,color=red,style={dashed}](l21)(t11)
  \Edge[Direct,bend=-45,color=blue,style={dashed}](l22)(f22)
  \Edge[Direct,bend=45,color=green,style={dashed}](l23)(f31)
  \Edge[Direct,bend=-45,color=red,style={dotted}](l31)(t12)
  \Edge[Direct,bend=-50,color=blue,style={dotted}](l32)(t21)
  \Edge[Direct,bend=75,color=green,style={dotted}](l33)(t32)
\end{scope}
\end{tikzpicture}}
    \caption{The reduction from \tsat{} to \otomrtt{}. The variable gadget (a) corresponding to the variable $x_i$ ($p_i$ is the number of clauses that contains the literal $x_i$, while $n_i$ is the number of clauses that contains the literal $\neg x_i$), the clause gadget (b) corresponding to the clause $c_j$, and the trip network (c) corresponding to the \tsat{} formula $(x_1\vee x_2\vee\neg x_3)\wedge(\neg x_1\vee x_2\vee x_3)\wedge(\neg x_1\vee\neg x_2\vee\neg x_3)$ (colors correspond to variables and line styles correspond to clauses).}
    \vspace{-5mm}
    \label{fig:theorem1}
\end{figure}

\medskip
We now define the trip collection \tnet{} on $D$. For each clause $c_j$ with $j\in[m]$ and for $h\in[3]$, \tnet{} contains the trip $\langle w_j,l_j^{h},t_i^{k},o_i^{k}\rangle$, if $(l_j^{h},t_i^{k})\in E$ and $o_i^{k}$ is defined as the unique out-neighbour of $t_i^{k}$ (that is, $o_i^k=t_i^{k+1}$ if $k<n_i$, and $o_i^k=v_{i+1}$ if $k=n_i$), and the trip $\langle w_j,l_j^{h},f_i^{k},o_i^{k}\rangle$, if $(l_j^{h},f_i^{k})\in E$ and $o_i^{k}$ is defined as the unique out-neighbour of $f_i^{k}$ (that is, $o_i^k=f_i^{k+1}$ if $k<p_i$, and $o_i^k=v_{i+1}$ if $k=p_i$). Each of the other $2n+3m+1$ edges, that are not yet included in a trip, forms a one-edge trip. Figure~\ref{fig:theorem1}(c) shows an example of the reduction in the case of the Boolean formula $(x_1\vee x_2\vee\neg x_3)\wedge(\neg x_1\vee x_2\vee x_3)\wedge(\neg x_1\vee\neg x_2\vee\neg x_3)$.

Let $\tau$ be a temporalisation of the trip network $(D,\tnet)$ and let $G=G[D,\tnet{},\tau]$ be the temporal graph induced by $\tau$. Note that, each edge in $D$ belongs to exactly one trip in $\tnet$, which means that, for each edge $e \in E$, there is exactly one temporal edge in $G$ with the same head and tail of $e$. Note also that, due to the topology of $D$, if $w_{m+1}\in\reach{G}{v_1}$, then the first part of the temporal path $P$ from $v_1$ to $w_{m+1}$ consists in moving from $v_1$ to $v_{n+1}$ by passing, for each $i\in[n]$, through the node $v_i$ and either through the nodes $t_i^1,\ldots,t_i^{n_i}$ or through the nodes $f_i^1,\ldots,f_i^{p_i}$. The second part of the temporal path $P$ consists in moving from $v_{n+1}$ to $w_{1}$ and, then, from $w_{1}$ to $w_{m+1}$ by passing, for each $j\in[m]$, through the node $w_j$ and exactly one $l_j$-node. Indeed, we can assume that $P$ does not go back from a $l_j$-node to a variable node to which it is connected, since otherwise $P$ should have to pass again through the edge $(v_{n+1},w_1)$, contradicting the fact that, as observed above, there is only one temporal edge in $G$ corresponding to this edge. Moreover, if $P$ uses a temporal edge with tail $w_j$ and head $l_j^{h}$, for some $h\in[3]$, and if $(l_j^{h},t_i^{k})\in E$ (respectively, $(l_j^{h},f_i^{k})\in E$), for some $i\in[n]$ and $k\in[n_i]$ (respectively, $k\in[p_i]$), then $P$ must have passed, in its first part, through the $f_i$-nodes (respectively, $t_i$-nodes) corresponding to the variable $x_i$. Otherwise, as $P$ is a temporal path, the edge outgoing $t_i^k$ (respectively, $f_i^k$) with head $o_i^k$ would have an appearance time smaller than that of $(w_j,l_j^h)$, contradicting the fact that $\tau$ is a temporalisation of the trip $\langle w_j,l_j^{h},t_i^{k},o_i^{k}\rangle$ (respectively, $\langle w_j,l_j^{h},f_i^{k},o_i^{k}\rangle$).

Let us now prove that $(D,\tnet{})$ is $(v_1,w_{m+1})$-temporalisable if and only if there exists an assignment $\alpha$ to the variables that satisfies the Boolean formula $\Phi$. Let us first suppose that $(D,\tnet{})$ is $(v_1,w_{m+1})$-temporalisable, that is, there exists a temporalisation $\tau$ of $(D,\tnet{})$ such that $w_{m+1}\in\reach{G}{v_1}$, where $G=G[D,\tnet{},\tau]$ is the temporal graph induced by $\tau$. Let $P$ be a temporal path from $v_1$ to $w_{m+1}$ in $G$. For each variable $x_i$ with $i\in[n]$, we set $\alpha(x_i)=\true$ if and only if $P$ passes through the $t_i$-nodes corresponding to $x_i$. We now prove that any clause $c_j$, with $j\in[m]$, is satisfied by $\alpha$. Let $l_j^{h}$, for some $h\in[3]$, be the node which $P$ goes to, when moving from $w_j$. If $(l_j^{h},f_i^{k})\in E$ (respectively, $(l_j^{h},t_i^{k})\in E$), for some $i\in[n]$ and $k\in[p_i]$ (respectively, $k\in[n_i]$), then the $h$-th literal of $c_j$ is $x_i$ (respectively, $\neg x_i$), and, because of the previous observations, $P$ passes through the $t_i$-nodes (respectively, $f_i$-nodes) corresponding to $x_i$: this implies that $\alpha(x_i)=\true$ (respectively, $\alpha(x_i)=\false$) and, hence, that the clause $c_j$ is satisfied.

Let us now suppose that there exists an assignment $\alpha$ to the variables that satisfies the formula $\Phi$, and let us consider the following walk $P$ in $D$, which starts from $v_{1}$ and arrives in $w_{m+1}$. The first part of $P$ arrives at $v_{n+1}$ and consists in moving from $v_i$ to $v_{i+1}$, for $i\in[n]$, by passing through the $t_i$-nodes (respectively, $f_i$-nodes) corresponding to $x_i$, if $\alpha(x_i)=\true$ (respectively, $\alpha(x_i)=\false$). We know that, for each clause $c_j$ with $j\in[m]$, at least one literal of $c_j$ is satisfied: suppose that the first such literal is the $h_j$-th one, for some $h_j\in[3]$. Note that the choice of $h_j$ implies that the first part of $P$ does not use any edge of the trip $T$ containing the edge $(c_j,l_j^{h_j})$: indeed, if the corresponding literal is $x_i$ (respectively, $\neg x_i$), $T$ goes through the $f_i$-nodes (respectively, $t_i$-nodes), while $P$ goes through the $t_i$-nodes (respectively, $f_i$-nodes) as $\alpha(x_i)=\true$ (respectively, $\alpha(x_i)=\false$). The second part of $P$ starts from $v_{n+1}$, moves to $w_1$, arrives at $w_{m+1}$, and consists in moving from $w_j$ to $w_{j+1}$, for $j\in[m]$, by passing through the node $l_j^{h_j}$. Because of the definition of \tnet{} and the choice of $h_1,\ldots,h_m$, each edge of $P$ belongs to a different trip in \tnet. We can then consider a schedule $S$ of $(D,\tnet{})$ in which the trips corresponding to the edges in $P$ are scheduled in the same order as they appear in $P$ itself. The walk $P$ thus induces a temporal path in $G[D,\tnet{},\tau_S]$ from $v_{1}$ to $w_{m+1}$. 
Thus, $(D,\tnet{})$ is $(s,t)$-temporalisable and the theorem has been proved.\qed
\end{proof}

Note that, as a consequence of its proof, the above theorem holds even with the following restrictions: the in-degree and the  out-degree of the nodes in $D$ are bounded by 3, and $\tnet$ contains only \textit{simple} trips (that is, trips that do not pass through the same node more than once), which are pairwise edge-disjoint. 

We are now ready to prove the inapproximability result of the \mrtt\ problem. Note that, given an instance $(D,\tnet)$ with $n$ nodes, any solution of the problem has a value greater than or equal to $n$ (since we have assumed that any node is temporally reachable from itself). Since the number of pairs of nodes is $n^2$, the \mrtt\ problem is trivially $n$-approximable. The following theorem states that this approximation ratio cannot be improved by a polynomial factor, unless $\ptime=\nptime$.

\begin{theorem}\label{thm:mrtthard}
Unless $\mathrm{P}=\mathrm{NP}$, the \mrtt{} problem is not $r(\cdot)$-approximable, for any $\epsilon\in(0,1)$ and for any non-decreasing function $r$ in $O(n^{1-\epsilon})$, where $n$ is the number of nodes.
\end{theorem}

\begin{proof}
The proof makes use of the well-known gap technique (see Section~3.1.4 of~\cite{AusielloMCGPK99}). Suppose by contradiction that there exists a $r(\cdot)$-approximation algorithm $\cal{A}$ for the \mrtt{} problem, for some $\epsilon\in(0,1)$ and for some function $r(n)$ of the number $n$ of nodes that satisfies $r(n)\le cn^{1-\epsilon}$ for some constant $c$. We will now show that it is possible to exploit such an algorithm in order to solve in polynomial time the \otomrtt{} problem, which would imply that $\ptime=\nptime$ (because of Theorem \ref{thm:otomrtthard}). Let us consider an instance $\langle (D=(V,E),\tnet),s,t\rangle$ of the \otomrtt{} problem, where $V=\{s=v_1,\ldots,v_n=t\}$. Without loss of generality, we assume that $n>c+1$. We define an instance $(D'=(V',E'),\tnet')$ of \mrtt{} as follows.

\begin{itemize}
    \item $V' = V \cup \{ v_{n+i} : i\in[2K]\}$ with $K = \left\lceil (c n)^{1/\epsilon} (n+2)^{\frac{2 - \epsilon}{\epsilon}}\right\rceil$.
    \item $E' = E  \cup \{(v_{n+i},s), (t,v_{n+K+i}) : i\in[K]\}\}$.
    \item \tnet$'$ is the union of \tnet{} with all the one-edge trips corresponding to the edges in $E' \setminus E$.
\end{itemize}

Consider an optimal temporalisation $\tau^*$ of $(D',\tnet{}')$: the maximum reachability is thus $\mathtt{opt} = \sum_{u\in V}|\reach{G[D',\tnet',\tau^*]}{u}|$. Moreover, let $x$ be the value of the reachability achieved by the temporalisation computed by the approximation algorithm $\cal{A}$ with input $(D',\tnet{}')$: hence, $\frac{\mathtt{opt}}{r(n')}\leq x\leq \mathtt{opt}$ where $n'=n+2K$. 

Let us upper-bound $\mathtt{opt}$ in the case in which $(D,\tnet)$ is not $(s,t)$-temporalisable. To this aim, we upper-bound the number of nodes $\tau^*$-reachable from each node in $V'$ (in the following, $G^* = G[D',\tnet{}',\tau^*]$ is the temporal graph induced by $\tau^*$).

\begin{itemize}
    \item $|\reach{G^*}{v_1}|\leq n-1$ (since $t$ is not $\tau^*$-reachable from $s=v_1$).
    
    \item For each $i \in [n-1]$, $|\reach{G^*}{v_{i+1}}|\leq n+K$ (since all nodes $v_{n+i}$ with $i \in [K]$ have in-degree equal to zero in $D'$).

    \item For each $i \in [K]$, $|\reach{G^*}{v_{n+i}}|\leq n$ (since $t$ is not $\tau^*$-reachable from $s$).
    
    \item For each $i \in [K]$, $|\reach{G^*}{v_{n+K+i}}|=1$ (since all these nodes have out-degree equal to zero in $D'$).
\end{itemize}

Thus, if $(D,\tnet)$ is not $(s,t)$-temporalisable, we have that $x \leq \mathtt{opt} \leq (n-1) + (n-1)(n+K)+ Kn + K = n^2+2nK-1$. On the other hand, if $(D,\tnet)$ is  $(s,t)$-temporalisable, then it is easy to produce a {temporalisation $\tau'$} of $(D',\tnet')$ such that $v_{n+K+i}\in\reach{G[D',\tnet',\tau']}{v_{n+j}}$, for any $i,j\in[K]$. Hence, in this case we have that $x \geq \frac{\mathtt{opt}}{r(n')} \geq \frac{K^2}{r(n')}$.  

If we prove that $\frac{K^2}{r(n')} > n^2+2nK-1$, then we have that $x > n^2+2nK-1$ if and only if $(D,\tnet)$ is $(s,t)$-temporalisable. This would imply that the \otomrtt{} problem is solvable in polynomial time, and the theorem is proved. Let us then show that $\frac{K^2}{r(n')} > n^2+2nK-1$. Since

\begin{equation*}
    \frac{K^2}{r(n')} \ge \frac{K^2}{c (n + 2K)^{1 - \epsilon}} >\frac{K^2}{c (n + 2)^{1 - \epsilon}K^{1 - \epsilon}} = \frac{K^{1 + \epsilon}}{c (n + 2)^{1 - \epsilon}},
\end{equation*}
it is sufficient to prove that $K^{1 + \epsilon} \geq c(n+2)^{1-\epsilon}n(n+2K)$. Since, by definition, $K \geq (cn)^{1/\epsilon} (n+2)^{\frac{2 - \epsilon}{\epsilon}}$, that is, $K^\epsilon \geq c n(n+2)^{2 - \epsilon}$, we have that $K^{1 + \epsilon} \geq cnK(n+2)^{2 - \epsilon} > cn(n + 2K)(n + 2)^{1 - \epsilon}$, and the proof is completed.\qed
\end{proof}

We can state a third hardness result concerning an optimisation problem which is, somehow, in between the \otomrtt{} problem and the \mrtt{} problem, that is, the following \textsc{Single Source Maximum Reachability Trip Temporalisation} (in short, \ssmrtt) problem: given a trip network $(D,\tnet)$ and a node $s$, find a temporalisation $\tau$ of $(D,\tnet)$ which maximises the $\tau$-reachability of $s$. Note that, given an instance of this problem, any solution has a value greater than or equal to $1$. Since the maximum number of nodes temporally reachable from $s$ is $n$, we have that the \ssmrtt{} problem is trivially $n$-approximable. The following theorem states that this approximation ratio cannot be improved by a polynomial factor, unless $\ptime=\nptime$.

\begin{theorem}\label{thm:ssmrtthard}
Unless $\mathrm{P}=\mathrm{NP}$, the \ssmrtt{} problem is not $r(\cdot)$-approximable, for any $\epsilon\in(0,1)$ and for any non-decreasing function $r$ in $O(n^{1-\epsilon})$, where $n$ is the number of nodes.
\end{theorem}

\begin{proof}
Similar to the proof of Theorem~\ref{thm:mrtthard}, we will make use of the gap technique. Suppose by contradiction that there exists an $r(\cdot)$-approximation algorithm $\cal{A}$ for the \ssmrtt{} problem, for some $\epsilon\in(0,1)$ and for some function $r(n)$ of the number $n$ of nodes that satisfies $r(n)\le cn^{1-\epsilon}$ for some constant $c$. We will now show that it is possible to exploit such an algorithm in order to solve in polynomial time the \otomrtt{} problem, which would imply that $\ptime=\nptime$ (because of Theorem \ref{thm:otomrtthard}). Let us consider an instance $\langle(D=(V,E),\tnet),s,t\rangle$ of the \otomrtt{} problem, where $V=\{s=v_1,\ldots,v_n=t\}$. Without loss of generality, we assume that $n>c+1$. We define an instance $\langle(D'=(V',E'),\tnet'),s\rangle$ of \ssmrtt{} as follows.

\begin{itemize}
    \item $V' = V \cup \{ v_{n+i}:i\in[K]\}$ with $K = \lceil c n^{2/\epsilon}\rceil$.
    \item $E' = E \cup \{(t,v_{n+i}) : i\in[K]\}$.
    \item \tnet$'$ is the union of \tnet{} with all the one-edge trips corresponding to the edges in $E' \setminus E$.
\end{itemize}

Consider an optimal temporalisation $\tau^*$ of $\langle(D',\tnet{}'),s\rangle$: the maximum reachability of $s$ is thus $\mathtt{opt} =|\reach{G[D',\tnet',\tau^*]}{s}|$. Moreover, let $x$ be the value of the temporalisation computed by the approximation algorithm $\cal{A}$: hence, $\frac{\mathtt{opt}}{r(n')}\leq x\leq \mathtt{opt}$ where $n'=n+K$.

If $(D,\tnet)$ is $(s,t)$-temporalisable, then $\mathtt{opt}\geq K+2$. Indeed, the very same temporalisation can be extended to $(D',\tnet')$, by assigning to all the new one-edge trips a starting time greater than the arrival time in $t$, so that all the $K$ new out-neighbors of $t$ are reachable. Note that $\mathtt{opt}\geq K + 2$ implies that $x\geq \frac{\mathtt{opt}}{r(n')}\geq \frac{K+2}{r(n')}\geq\frac{cn^{2/\epsilon}}{r(n')}$. Since $n'=n+K$, we have that $n'<n+cn^{2/\epsilon}+1<n^{1+2/\epsilon}$ (since $n>c+1$ and $\epsilon\in(0,1)$). Since $r$ is non-decreasing, we get $x\geq\frac{cn^{2/\epsilon}}{r(n')} > \frac{cn^{2/\epsilon}}{r(n^{1+2/\epsilon})} \geq \frac{cn^{2/\epsilon}}{c(n^{1+2/\epsilon})^{1-\epsilon}} = n^{1+\epsilon}>n$. Hence, if $(D,\tnet)$ is $(s,t)$-temporalisable, then $x>n$. On the other hand, if $(D,\tnet)$ is not $(s,t)$-temporalisable, then $\mathtt{opt}<n$, since none of the new out-neighbors of $t$ are $\tau^*$-reachable from $s$ (without passing through $t$). Hence, $x\leq\mathtt{opt}<n$.

We can conclude that $(D,\tnet)$ is $(s,t)$-temporalisable if and only if the value $x$ of the solution computed by the approximation algorithm $\cal{A}$ is greater than $n$. This implies that the \otomrtt{} problem is solvable in polynomial time, and the theorem is proved.\qed
\end{proof}

\subsection{Bounding the number of used trips}
\label{sec:fpt}

In this section, we study the \otomrtt{} problem parameterised by the number of trips needed to go from the source to the destination.

Given a trip network $(D,\tnet)$, a temporalisation $\tau$ and $k \in \mathbb{N}$, a node $v$ is said to be \emph{$(k,\tau)$-reachable} from $u$ if there exists a temporal path $P$ in $G=G[D,\tnet,\tau]$ from $u$ to $v$ which is composed of edges induced by at most $k$ different trips in $\tnet$.
Given such a path $P$, let $\tnet_P$ be a set of at most $k$ trips which induce the edges contained in $P$. 
Note that, without loss of generality, we can suppose that the edges induced by the same trip are contiguous in $P$.
Indeed, if $P=e_1, \dots, e_p$, let $T_1\in\tnet_P$ be one of the trips that induces $e_{1}$ and let $e_h$ be the last edge in $P$ which is induced by $T_{1}$. The temporalisation of $T_{1}$ induces a temporal path $e_1'=e_1,e_2',\dots,e_{l-1}',e_l'=e_h$ in $G$. We can then consider the temporal path $P'= e_1',\dots,e_l',e_{h+1},\dots,e_p$, which has the property that all the edges which are induced by $T_{1}$ are contiguous. We can now apply the same argument by considering the trip $T_{2}\in\tnet_P$ as one of the trips that induces $e_{h+1}$, and go on like this until all the edges induced by each trip considered are contiguous in the final temporal path from $u$ to $v$.

We now consider the \otomrtt{} problem parameterised by the number $k$ of trips used in the resulting temporal graph, in order to go from $s$ to $t$. More precisely, given a trip network $(D,\tnet)$, a source node $s$, and a target node $t$, the parameterised problem $k$-\otomrtt{} consists in deciding whether there exists a temporalisation $\tau$ such that $t$ is $(k,\tau)$-reachable from $s$. By using the color coding technique developed in~\cite{AlonYZ16}, we can obtain the following result.

\begin{theorem}\label{thm:parameterised}
The $k$-\otomrtt{} problem can be solved in $2^{O(k)}m\log |\tnet|$ time where $m=\sum_{T\in \tnet}|T|$ is the sum of trip lengths.
\end{theorem}

\begin{proof}
     Let us consider a trip network $(D=(V,E),\tnet)$, a source node $s$ and a target node $t$, and let $M$ be the induced multidigraph of $(D,\tnet{})$ (note that $m$ is also the number of edges in $M$). We suppose that for each trip $T\in\tnet$, we are given the list of its edges in their respective order in $T$. We also let $V(T)$ denote the set of nodes appearing in $T$ and write $u\preceq_T v$ when $u,v\in V(T)$ and the first occurrence of $u$ precedes the last occurrence of $v$ in $T$. In particular, we have $u\preceq_Tu$ for $u\in V(T)$ only when $u$ occurs several times in $T$. Let $\chi:\tnet{}\rightarrow [k]$ be any color assignment to the trips in $\tnet{}$. For $i\in[k]$, an \emph{$(i,\chi)$-walk} $P$ in $M$ is a walk which is the concatenation of exactly $i$ subtrips of distinct trips in $\tnet{}$ with pairwise distinct colors (in the following, $\chi(P)$ will denote the set of colors ``used'' by such a walk $P$). 
     %Although we could restrict our attention to paths in $M$ rather than walks, it is more convenient to do so as our computation will allow loops. 
     Note that the existence of an $(i,\chi)$-walk in $M$ from $s$ to $t$ with $i\in[k]$ implies the existence of a schedule $S$ such that $t$ is $(k,\tau_{S})$-reachable from $s$: indeed, we can first schedule the $i$ trips of the walk in the order they appear in it, and then the remaining trips of $\tnet{}$ in any order. In order to test the existence of an $(i,\chi)$-walk in $M$ from $s$ to $t$ for $i\in[k]$, we can use the dynamic programming technique~\cite{Bellman1954}. For any node $v\in V$, let $A_{\chi}[v,i]$ denote the collection of sets $C$ of $i$ colors for which there exists an $(i,\chi)$-walk $P$ in $M$ from $s$ to $v$ such that $\chi(P)=C$. Clearly, there exists an $(i,\chi)$-walk in $M$ from $s$ to $t$ if and only if $A_{\chi}[t,i]\neq\emptyset$. We have that, for any node $v\in V$,
     \[A_{\chi}[v,1] = \{\{\chi(T)\}:T \in\tnet \wedge s \preceq_T v\}.\]
    \newcommand{\negspace}{\!\!\!\!\!\!}
    Moreover, for any $i\in[k-1]$,
    \begin{align*}
    \begin{split}
        A_{\chi}[v,i+1] = &
        \negspace\bigcup_{T\in\mathbb{T}: v\in V(T)}\negspace
    \left\{C\cup\{\chi(T)\}: \exists u\left[u \preceq_T v\wedge C\in A_{\chi}[u,i]\wedge\chi(T)\not\in C\right]\right\}\\
         = & 
        \negspace\bigcup_{T\in\mathbb{T}: v\in V(T)}\negspace
        \left\{C\cup\{\chi(T)\}:C\in
        {\cal C}_i(T,v) \wedge\chi(T)\not\in C\right\},
    \end{split}
    \end{align*}
where ${\cal C}_i(T,v) = \bigcup_{u \preceq_T v}A_{\chi}[u,i]$. We can compute $A_{\chi}[v,1]$ for all $v\in V$ by simply scanning all trips in $\tnet$: this can be done in $O(m)$ time. Moreover, we can execute the above update rule for all $v\in V$ by scanning each trip $T\in \tnet$ once as follows. We first set  $A_{\chi}[v,i+1] := \emptyset$ for all $v\in V$. Then, for each trip $T$, we iterate over the edges of $T$ in their respective order in $T$: we first initialize a collection ${\cal C}:=\emptyset$, and for each edge $(u,v)$ of $T$, we update ${\cal C}:={\cal C}\cup A_\chi[u,i]$ and $A_{\chi}[v,i+1] := A_{\chi}[v,i+1] \cup \left\{C\cup\{\chi(T)\}:C\in{\cal C} \wedge\chi(T)\not\in C\right\}$. % (if $(u,v)$ is the first edge of $T$ we simply use ${\cal C}_i(T,v)=A_\chi[u,i]$ as $u$ is then the only node preceding $v$ in $T$).
Note that when we scan the last edge $(u,v)$ of $T$ having $v$ as head, it corresponds to the last occurrence of $v$ in $T$, and the update of $\cal C$ leads to ${\cal C}={\cal C}_i(T,v)$. This implies that the update of $A_{\chi}[v,i+1]$ then coincides with the above update rule.
Note also that both updates of ${\cal C}$ and $A_{\chi}[v,i+1]$
take $O(k2^k)$ time since, for any node $u\in V$, we have %and for $j\in[k]$,
$|A_{\chi}[u,i]|\leq 2^{k}$ and $|{\cal C}|\le 2^k$. Hence computing the values $A_\chi[v,i+1]$ for all $v$ from values $A_\chi[v,i]$ takes $O(k2^k m)$ time and the whole computation of matrix $A_\chi$ requires $2^{O(k)} m$ time.

Observe now that if there exists a temporalisation of $(D,\tnet{})$ such that $t$ is $(k,\tau)$-reachable from $s$, then there must exist a color assignment $\chi$ such that $M$ includes an $(i,\chi)$-walk from $s$ to $t$ for some $i\in[k]$. In order to find such a walk, we can use an appropriate set of perfect hash functions from $[|\tnet|]$ to $[k]$. Indeed, it is possible to design $2^{O(k)}\log |\tnet|$ hash functions such that any subset of $k$ trips has image $\{1,\ldots,k\}$ for at least one function~\cite{NaorSS95}. This implies that any subset of $i$ trips has $i$ pairwise distinct colors as image for at least one function as such a set can be completed in a set of $k$ trips. %Each hash function can be coded with $O(k+\log\log |\tnet|)$ bits, and can be generated in $O(k^3\log |\tnet|)$ time~\cite{FredmanKS84}. The number of such functions is $O(2^k\log|\tnet|)$ and their computation takes $2^{O(k)}\log^2|\tnet|$ time.
More precisely, a set of $e^kk^{O(\log k)}\log |\tnet|=2^{O(k)}\log |\tnet|$ such functions can be explicitly computed in $e^kk^{O(\log k)}|\tnet| \log |\tnet|=2^{O(k)}|\tnet| \log |\tnet|$ time~\cite{NaorSS95} (see also Theorem~5.18 in~\cite{CyganFKLMPPS15}). Testing the coloring obtained through each function thus yields a $2^{O(k)}m\log |\tnet|$-time algorithm for solving the $k$-\otomrtt{} problem (we use $|\tnet|\le m$). The theorem is thus proved.
    \qed
\end{proof}

% ------------------- \input{body/stronglytemporalisable}
%

\section{Strongly temporalisable trip networks}
\label{sec:stronglytemporisable}

We now switch to strongly temporalisable trip networks where one-to-one reachability is assumed for all pairs of nodes. This clearly implies that the \otomrtt{} problem is trivially solvable when restricted to strongly temporalisable trip networks, since the answer is always yes (actually, one-to-one reachability is always satisfied under strong temporalisability). On the other hand, we will prove that both the \mrtt{} and the \ssmrtt{} problem cannot be approximated within a factor less than $\frac{\sqrt{n}}{12}$ (unless $\ptime=\nptime$). To this aim, we first show that the strong temporalisability by itself is not enough to ensure the existence of a temporalisation $\tau$ with a $\tau$-reachability which is a constant fraction of all pairs of nodes.

\begin{theorem}\label{thm:pairscheduleconnected}
For any $r>3$, there exists a strongly temporalisable trip network $(D_{r},\tnet_{r})$ with $n=r^2+2r$ nodes, such that any temporalisation $\tau$ of $(D_{r},\tnet_{r})$ has $\tau$-reachability $O(n^{\frac{3}{2}})$.
\end{theorem}

\begin{proof}
We first define the trip network $(D_{r},\tnet_{r})$ through the gadgets that compose it (see Figure~\ref{fig:pairscheduleconnected}). We then prove that the trip network is strongly temporalisable and, finally, we prove that, for any temporalisation $\tau$, the $\tau$-reachability is $O(n^{\frac{3}{2}})$.

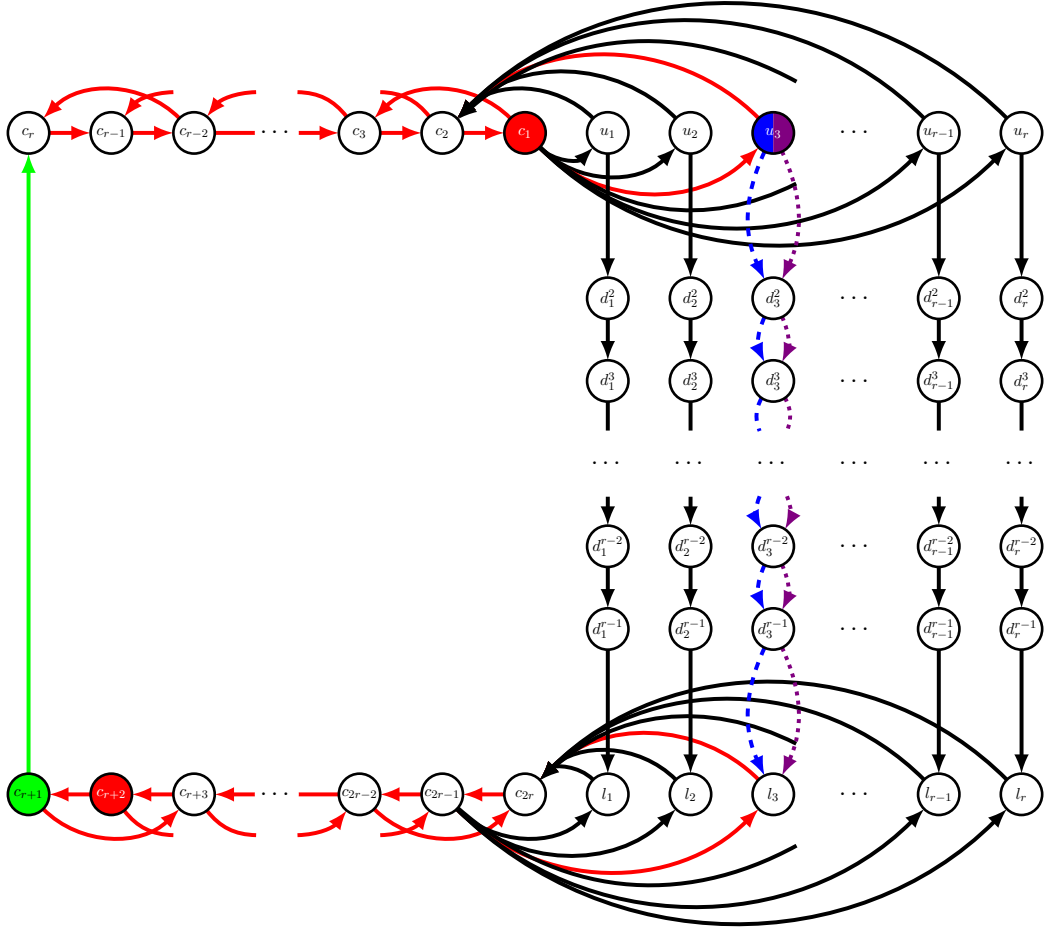
\begin{figure}[t]
    \centerfloat
    \SetVertexStyle[FillColor=white,TextFont=\normalsize]
\SetEdgeStyle[Color=black]
\begin{tikzpicture}[x=2cm,y=2cm,scale=0.55, every node/.style={scale=0.55}]
  \Vertex[x=8,y=0,size=1,label=$c_1$,color=red,fontcolor=white]{c1}
  \Vertex[x=6,y=0,size=1,label=$c_2$]{c2}
  \Vertex[x=4,y=0,size=1,label=$c_3$]{c3}
  \Vertex[x=2,y=0,size=1,label=$\cdots$,fontsize=\LARGE,fontsize=\LARGE,style={color=white},opacity=0]{c4}
  \Vertex[x=0,y=0,size=1,label=$c_{r-2}$]{c5}
  \Vertex[x=-2,y=0,size=1,label=$c_{r-1}$]{c6}
  \Vertex[x=-4,y=0,size=1,label=$c_{r}$]{c6a}
  \Vertex[x=10,y=0,size=1,label=$u_1$]{t1}
  \Vertex[x=12,y=0,size=1,label=$u_2$]{t2}
  \draw[fill=blue] (7,0.25) arc (90:270:0.5cm);
  \draw[fill=violet] (7,0.25) arc (90:-90:0.5cm);
  \Vertex[x=14,y=0,size=1,label=$u_3$,opacity=0,fontcolor=white]{t3}
  \Vertex[x=16,y=0,size=1,label=$\cdots$,fontsize=\LARGE,style={color=white},opacity=0]{t4}
  \Vertex[x=18,y=0,size=1,label=$u_{r-1}$]{t5}
  \Vertex[x=20,y=0,size=1,label=$u_{r}$]{t6}
  \Vertex[x=-4,y=-16,size=1,label=$c_{r+1}$,color=green,fontcolor=black]{c7a}
  \Vertex[x=-2,y=-16,size=1,label=$c_{r+2}$,color=red,fontcolor=white]{c7}
  \Vertex[x=0,y=-16,size=1,label=$c_{r+3}$]{c8}
  \Vertex[x=2,y=-16,size=1,label=$\cdots$,fontsize=\LARGE,style={color=white},opacity=0]{c9}
  \Vertex[x=4,y=-16,size=1,label=$c_{2r-2}$]{c10}
  \Vertex[x=6,y=-16,size=1,label=$c_{2r-1}$]{c11}
  \Vertex[x=8,y=-16,size=1,label=$c_{2r}$]{c12}
  \Vertex[x=10,y=-16,size=1,label=$l_1$]{b1}
  \Vertex[x=12,y=-16,size=1,label=$l_2$]{b2}
  \Vertex[x=14,y=-16,size=1,label=$l_3$]{b3}
  \Vertex[x=16,y=-16,size=1,label=$\cdots$,fontsize=\LARGE,style={color=white},opacity=0]{b4}
  \Vertex[x=18,y=-16,size=1,label=$l_{r-1}$]{b5}
  \Vertex[x=20,y=-16,size=1,label=$l_{r}$]{b6}
  \Vertex[x=10,y=-4,size=1,label=$d_1^2$]{d1-f}
  \Vertex[x=12,y=-4,size=1,label=$d_2^2$]{d2-f}
  \Vertex[x=14,y=-4,size=1,label=$d_3^2$]{d3-f}
  \Vertex[x=16,y=-4,size=1,label=$\cdots$,fontsize=\LARGE,style={color=white},opacity=0]{d4-f}
  \Vertex[x=18,y=-4,size=1,label=$d_{r-1}^2$]{d5-f}
  \Vertex[x=20,y=-4,size=1,label=$d_{r}^2$]{d6-f}
  \Vertex[x=10,y=-6,size=1,label=$d_1^3$]{d1-s}
  \Vertex[x=12,y=-6,size=1,label=$d_2^3$]{d2-s}
  \Vertex[x=14,y=-6,size=1,label=$d_3^3$]{d3-s}
  \Vertex[x=16,y=-6,size=1,label=$\cdots$,fontsize=\LARGE,style={color=white},opacity=0]{d4-s}
  \Vertex[x=18,y=-6,size=1,label=$d_{r-1}^3$]{d5-s}
  \Vertex[x=20,y=-6,size=1,label=$d_{r}^3$]{d6-s}
  \Vertex[x=10,y=-8,size=1,label=$\cdots$,fontsize=\LARGE,style={color=white},opacity=0]{cd1}
  \Vertex[x=12,y=-8,size=1,label=$\cdots$,fontsize=\LARGE,style={color=white},opacity=0]{cd1}
  \Vertex[x=14,y=-8,size=1,label=$\cdots$,fontsize=\LARGE,style={color=white},opacity=0]{cd1}
  \Vertex[x=16,y=-8,size=1,label=$\cdots$,fontsize=\LARGE,style={color=white},opacity=0]{cd1}
  \Vertex[x=18,y=-8,size=1,label=$\cdots$,fontsize=\LARGE,style={color=white},opacity=0]{cd1}
  \Vertex[x=20,y=-8,size=1,label=$\cdots$,fontsize=\LARGE,style={color=white},opacity=0]{cd1}
  \Vertex[x=10,y=-10,size=1,label=$d_1^{r-2}$]{d1-bl}
  \Vertex[x=12,y=-10,size=1,label=$d_2^{r-2}$]{d2-bl}
  \Vertex[x=14,y=-10,size=1,label=$d_3^{r-2}$]{d3-bl}
  \Vertex[x=16,y=-10,size=1,label=$\cdots$,fontsize=\LARGE,style={color=white},opacity=0]{d4-bl}
  \Vertex[x=18,y=-10,size=1,label=$d_{r-1}^{r-2}$]{d5-bl}
  \Vertex[x=20,y=-10,size=1,label=$d_{r}^{r-2}$]{d6-bl}
  \Vertex[x=10,y=-12,size=1,label=$d_1^{r-1}$]{d1-l}
  \Vertex[x=12,y=-12,size=1,label=$d_2^{r-1}$]{d2-l}
  \Vertex[x=14,y=-12,size=1,label=$d_3^{r-1}$]{d3-l}
  \Vertex[x=16,y=-12,size=1,label=$\cdots$,fontsize=\LARGE,style={color=white},opacity=0]{d4-l}
  \Vertex[x=18,y=-12,size=1,label=$d_{r-1}^{r-1}$]{d5-l}
  \Vertex[x=20,y=-12,size=1,label=$d_{r}^{r-1}$]{d6-l}
  \Edge[Direct,color=red](c6a)(c6)
  \Edge[Direct,color=red](c6)(c5)
  \Edge[color=red](c5)(c4)
  \Edge[Direct,color=red](c4)(c3)
  \Edge[Direct,color=red](c3)(c2)
  \Edge[Direct,color=red](c2)(c1)
  \Edge[Direct,bend=-45,color=red](c1)(c3)
  \Vertex[x=4,y=1,size=1,Pseudo]{etd3}
  \Edge[bend=-25,color=red](c2)(etd3)
  \Vertex[x=2,y=1,size=1,Pseudo]{etd4}
  \Edge[bend=-25,color=red](c3)(etd4)
  \Edge[Direct,bend=-25,color=red](etd4)(c5)
  \Vertex[x=0,y=1,size=1,Pseudo]{etd5}
  \Edge[Direct,bend=-25,color=red](etd5)(c6)
  \Edge[Direct,bend=-45,color=red](c5)(c6a)
  \Edge[Direct,bend=-45](c1)(t1)
  \Edge[Direct,bend=-45](c1)(t2)
  \Edge[Direct,bend=-45,color=red](c1)(t3)
  \Vertex[x=15,y=-1,size=1,Pseudo]{etd2}
  \Edge[bend=-35](c1)(etd2)
  \Edge[Direct,bend=-45](c1)(t5)
  \Edge[Direct,bend=-45](c1)(t6)
  \Edge[Direct,bend=-45](t1)(c2)
  \Edge[Direct,bend=-45](t2)(c2)
  \Edge[Direct,bend=-45,color=red](t3)(c2)
  \Vertex[x=15,y=1,size=1,Pseudo]{etd1}
  \Edge[Direct,bend=-35](etd1)(c2)
  \Edge[Direct,bend=-45](t5)(c2)
  \Edge[Direct,bend=-45](t6)(c2)
  \Edge[Direct,color=red](c7)(c7a)
  \Edge[Direct,color=red](c8)(c7)
  \Edge[Direct,color=red](c9)(c8)
  \Edge[color=red](c10)(c9)
  \Edge[Direct,color=red](c11)(c10)
  \Edge[Direct,color=red](c12)(c11)
  \Edge[Direct,bend=-45,color=red](c7a)(c8)
  \Vertex[x=0,y=-17,size=1,Pseudo]{ebd3}
  \Edge[bend=-25,color=red](c7)(ebd3)
  \Vertex[x=2,y=-17,size=1,Pseudo]{ebd4}
  \Edge[bend=-25,color=red](c8)(ebd4)
  \Edge[Direct,bend=-25,color=red](ebd4)(c10)
  \Vertex[x=4,y=-17,size=1,Pseudo]{ebd5}
  \Edge[Direct,bend=-25,color=red](ebd5)(c11)
  \Edge[Direct,bend=-45,color=red](c10)(c12)
  \Edge[Direct,bend=-45](c11)(b1)
  \Edge[Direct,bend=-45](c11)(b2)
  \Edge[Direct,bend=-45,color=red](c11)(b3)
  \Vertex[x=15,y=-17,size=1,Pseudo]{ebd2}
  \Edge[bend=-35](c11)(ebd2)
  \Edge[Direct,bend=-45](c11)(b5)
  \Edge[Direct,bend=-45](c11)(b6)
  \Edge[Direct,bend=-45](b1)(c12)
  \Edge[Direct,bend=-45](b2)(c12)
  \Edge[Direct,bend=-45,color=red](b3)(c12)
  \Vertex[x=15,y=-15,size=1,Pseudo]{ebd1}
  \Edge[Direct,bend=-35](ebd1)(c12)
  \Edge[Direct,bend=-45](b5)(c12)
  \Edge[Direct,bend=-45](b6)(c12)
  \Edge[Direct,color=green](c7a)(c6a)
  \Edge[Direct](t1)(d1-f)
  \Edge[Direct](d1-f)(d1-s)
  \Vertex[x=10,y=-7.5,Pseudo]{td1-1}
  \Edge(d1-s)(td1-1)
  \Edge[Direct](t2)(d2-f)
  \Edge[Direct](d2-f)(d2-s)
  \Vertex[x=12,y=-7.5,Pseudo]{td2-1}
  \Edge(d2-s)(td2-1)
  \Edge[Direct,color=blue,bend=-25,style={dashed}](t3)(d3-f)
  \Edge[Direct,color=violet,bend=25,style={dotted}](t3)(d3-f)
  \Edge[Direct,bend=-25,color=blue,style={dashed}](d3-f)(d3-s)
  \Edge[Direct,bend=25,color=violet,style={dotted}](d3-f)(d3-s)
  \Vertex[x=13.75,y=-7.5,Pseudo]{td3-1}
  \Edge[bend=-25,color=blue,style={dashed}](d3-s)(td3-1)
  \Vertex[x=14.25,y=-7.5,Pseudo]{td3-2}
  \Edge[bend=25,color=violet,style={dotted}](d3-s)(td3-2)
  \Edge[Direct](t5)(d5-f)
  \Edge[Direct](d5-f)(d5-s)
  \Vertex[x=18,y=-7.5,Pseudo]{td5-1}
  \Edge(d5-s)(td5-1)
  \Edge[Direct](t6)(d6-f)
  \Edge[Direct](d6-f)(d6-s)
  \Vertex[x=20,y=-7.5,Pseudo]{td6-1}
  \Edge(d6-s)(td6-1)
  \Edge[Direct](d1-l)(b1)
  \Edge[Direct](d1-bl)(d1-l)
  \Vertex[x=10,y=-8.5,Pseudo]{bd1-2}
  \Edge[Direct](bd1-2)(d1-bl)
  \Edge[Direct](d2-l)(b2)
  \Edge[Direct](d2-bl)(d2-l)
  \Vertex[x=12.25,y=-8.5,Pseudo]{bd2-1}
  \Vertex[x=12,y=-8.5,Pseudo]{bd2-2}
  \Edge[Direct](bd2-2)(d2-bl)
  \Edge[Direct,bend=-25,color=blue,style={dashed}](d3-l)(b3)
  \Edge[Direct,bend=25,color=violet,style={dotted}](d3-l)(b3)
  \Edge[Direct,bend=-25,color=blue,style={dashed}](d3-bl)(d3-l)
  \Edge[Direct,bend=25,color=violet,style={dotted}](d3-bl)(d3-l)
  \Vertex[x=14.25,y=-8.5,Pseudo]{bd3-1}
  \Edge[Direct,bend=25,color=violet,style={dotted}](bd3-1)(d3-bl)
  \Vertex[x=13.75,y=-8.5,Pseudo]{bd3-2}
  \Edge[Direct,bend=-25,color=blue,style={dashed}](bd3-2)(d3-bl)
  \Edge[Direct](d5-l)(b5)
  \Edge[Direct](d5-bl)(d5-l)
  \Vertex[x=18,y=-8.5,Pseudo]{bd5-2}
  \Edge[Direct](bd5-2)(d5-bl)
  \Edge[Direct](d6-l)(b6)
  \Edge[Direct](d6-bl)(d6-l)
  \Vertex[x=20,y=-8.5,Pseudo]{bd6-2}
  \Edge[Direct](bd6-2)(d6-bl)
\end{tikzpicture}
    \caption{An example of a strongly temporalisable trip network, such that any temporalisation cannot connect a constant fraction of the total pairs of nodes, obtained via the construction described in the proof of Theorem~\ref{thm:pairscheduleconnected}. The two red solid trips, starting from the two red nodes, correspond to the trip $T_3^U$ and $T_3^L$, respectively, in the construction (the first time a trip passes through a node with more than one red outgoing edge, it continues towards the node with the smaller index). The blue dashed (respectively, violet dotted) trip, starting from the half blue (respectively, violet) node, corresponds to the trip $T_3^{\downarrow\mathrm{l}}$ (respectively, $T_3^{\downarrow\mathrm{r}}$) in the construction. Finally, the green solid trip, starting from the green node, corresponds to the trip $T^\uparrow$ in the construction.}
    \label{fig:pairscheduleconnected}
\end{figure}

\begin{description}
\item[\textbf{Upper gadget} \boldmath$V^U,E^U,\tnet^U$.] The set $V^{U}$ contains the nodes $c_1, \dots, c_{r}$, and the nodes $u_1, \dots, u_{r}$. These nodes are connected through the following set of directed edges:
\[E^{U} = \{(c_{i+1},c_{i}) : i\in[r-1] \} \cup \{(c_i,c_{i+2}) : i\in[r-2] \}\cup \{(c_1,u_i),(u_i,c_2) : i\in[r]\}.\]
On this gadget, we have the following collection of trips: $\tnet^{U} = \{ T^{U}_i : i\in[r] \}$, where
\[T_i^{U} = \langle c_1,u_i,c_2, c_1, c_3, c_2, \dots, c_{r - 1},c_{r - 2}, c_{r}, c_{r-1}\rangle\]
(each edge $(c_{i+1},c_i)$ is followed by $(c_i,c_{i+2})$, see also Figure~\ref{fig:pairscheduleconnected}, where the upper red solid trip is $T_{3}^{U}$).

\item[\textbf{Lower gadget} \boldmath$V^L,E^L,\tnet^L$.] The set $V^{L}$ contains the nodes $c_{r + 1}, \dots, c_{2r}$, and the nodes $l_1, \dots, l_{r}$. These nodes are connected through the following set of directed edges:
\begin{eqnarray*}
E^{L} & = & \{(c_{r+i+1},c_{r+i}) : i\in[r-1] \} \cup \{(c_{r+i},c_{r+i+2}) : i\in[r-2] \}\\
& & \cup \{(c_{2r -1},l_i),(l_i, c_{2r}) : i\in[r] \}.
\end{eqnarray*}
On this gadget, we have the following collection of trips: $\tnet^{L} = \{ T^{L}_i : i\in[r] \}$, where
\[T_i^{L} = \langle c_{r+2},c_{r+1}, c_{r+3}, c_{r+2},\dots,c_{2r-2},c_{2r}, c_{2r -1}, l_i, c_{2r}\rangle\]
(each edge $(c_{r+i+1},c_{r+i})$ is followed by $(c_{r+i},c_{r+i+2})$, see also Figure~\ref{fig:pairscheduleconnected}, where the lower red solid trip is $T_{3}^{L}$).

\item[\textbf{Descending gadgets} \boldmath$V^{\downarrow}_i,E^{\downarrow}_i,{\tnet^{\downarrow}}$.]  For any $i$ with $i\in[r]$, we refer to node $u_i$ and $l_i$ as $d_i^1$ and $d_i^{r}$, respectively. The set $V^{\downarrow}_i$ contains the nodes $d^2_i, \dots, d^{r - 1}_i$. These nodes are connected among them and to the previous gadgets through the following set of directed edges:
\[E^{\downarrow}_i = \{(d_i^j,d_i^{j+1}) : j\in[r - 1] \}.\]
On this gadget, we have the following collection of trips: $\tnet^{\downarrow} = \{ T^{\downarrow\mathrm{l}}_i, T^{\downarrow\mathrm{r}}_i : i\in[r] \}$, where
\[T^{\downarrow\mathrm{l}}_i = T^{\downarrow\mathrm{r}}_i = \langle u_i=d_i^1, d_i^2, \dots, d_i^{r-1}, d_i^{r}=l_i  \rangle\]
(see Figure~\ref{fig:pairscheduleconnected}, where the blue dashed trip is $T_{3}^{\downarrow\mathrm{l}}$ and the violet dotted trip is $T_{3}^{\downarrow\mathrm{r}}$). 

\item[\textbf{Ascending gadget} \boldmath$e^{\uparrow\fixup},T^{\uparrow\fixup}$.] This gadget contains the edge $e^{\uparrow}=(c_{r+1}, c_{r})$, which connects the lower gadget to the upper gadget, and is also a one-edge trip $T^{\uparrow}$. Note that this gadget does not introduce any new nodes.
\end{description}

To conclude the definition of the network, we set $D_r=(V,E)$ where
\[V = V^{U} \cup V^{L} \cup\bigcup_{i=1}^{r} V^{\downarrow}_i\]
(note that $|V|=2r+2r+r(r-2) = r^2+2r$) and
\[E = E^{U} \cup E^{L} \cup\bigcup_{i=1}^{r} E^{\downarrow}_i\cup\{e^{\uparrow}\},\]
and we set
\[\tnet_r = \tnet^{U} \cup \tnet^{L} \cup \tnet^{\downarrow}\cup \{T^{\uparrow}\}.\]
Note that the descending gadgets contain the majority of the nodes in the trip network, and that it is possible to visit them by entering from the upper nodes and by travelling all the way down (the second descending trip $T_i^{\downarrow\mathrm{r}}$ is necessary in order to ensure that the trip network is strongly temporalisable).

\medskip
\noindent\textbf{$(D_{r},\tnet_{r})$ is strongly temporalisable.} In order to ease the reading of the proof let us first define the following $r$ (partial) schedules of the trip networks $(D_{r},\tnet^U)$ and $(D_{r},\tnet^L)$, respectively. 

\begin{itemize}
    \item Schedule $S_i^{U}$ for $i\in[r]$. This schedule has the purpose of making $u_i$ $S_{i}^{U}$-reachable from $c_{r}$. The schedule $S_i^{U}$ consists in having the trips $\tnet^{U}$ scheduled in any order that has $T_i^{U}$ as the last one. Starting from $c_{r}$, it is possible to go to $c_{r-1}$ through the edge $(c_{r}, c_{r-1})$ of the first trip scheduled, to $c_{r-2}$ through the edge $(c_{r-1}, c_{r-2})$ of the second trip scheduled, and so on. In this way, it is possible to reach $c_1$ using the first $r-1$ trips scheduled. Finally, from $c_1$ it is possible to go to $u_{i}$ through the edge $(c_1,u_i)$ of the trip $T_i^{U}$. Note that this also shows that, for any $h,k\in[r]$ with $h > k$, $c_k$ is $S_{i}^{U}$-reachable from $c_h$.
    
    \item Schedule $S_i^{L}$ for $i\in[r]$. This schedule is similar to the previous one, but applied to the lower gadget. It allows node $c_{r+1}$ to be $S_i^{L}$-reachable from $l_i$. The schedule $S_i^{L}$ consists in having the trips $\tnet^{L}$ scheduled in any order that has $T_i^{L}$ as the first one. Starting from $l_i$, it is possible to go to $c_{2r}$ through the edge $(l_{i},c_{2r})$ of the trip $T_i^{L}$. At this point, it is possible to reach $c_{r+1}$ from $c_{2r}$ by using one edge of each of the remaining $r - 1$ trips. Note that this also shows that, for any $h,k\in[r]$ with $h > k$, $c_{k+r}$ is $S_{i}^{L}$-reachable from $c_{h+r}$.
\end{itemize}

\begin{table}[t]
\scalebox{.82}{\begin{small}
  \small
\begin{tabular}{||p{1.4cm}||p{2.2cm}|p{1.6cm}|p{2.2cm}|p{3.3cm}|p{6.4cm}||}
\cline{2-6}
\multicolumn{1}{c||}{} & \multicolumn{5}{c||}{\textbf{Destination}}\\
\hline
\hfil \textbf{Source} & \hfil $c_k\in V^U$ & \hfil $u_k\in V^U$ & \hfil $c_k\in V^L$ & \hfil $l_k\in V^L$  & \hfil $d_k^{l_2}\in V^{\downarrow\mathrm{l}}_k$\\
\hline\hline
$c_h\in V^U$ & $\begin{array}{ll}S_1^U & \mbox{if $k<h$}\\ T_1^U & \mbox{if $k>h$}\end{array}$  & \centering{$S_k^U$} & \centering{$S_1^U,T_{1}^{\downarrow\mathrm{l}},S_1^L$}  & \multicolumn{2}{p{8.4cm}||}{\centering{$S_k^U,T_k^{\downarrow\mathrm{l}}$}}\\
\hline
$u_h\in V^U$ & \centering{$T_h^U$} & \centering{$T_h^U,T_k^U$} &  \centering{$T_h^{\downarrow\mathrm{l}},S_h^L$}  & \multicolumn{2}{p{8.4cm}||}{\centering{$\begin{array}{ll}T_k^{\downarrow\mathrm{l}} & \mbox{if $k=h$}\\ T_h^U,T_k^U,T_k^{\downarrow\mathrm{l}} & \mbox{if $k\neq h$}\end{array}$}}\\
\hline
$c_h\in V^L$ & \centering{$S_1^L,T^{\uparrow},S_1^U$} & \centering{$S_1^L,T^{\uparrow},S_k^U$} & $\begin{array}{ll}S_1^U & \mbox{if $k<h$}\\ T_1^L & \mbox{if $k>h$}\end{array}$ & \centering{$T_k^L$}  & \multirow{2}{3cm}{\centering{$S_h^L,T^{\uparrow},S_k^U,T_k^{\downarrow\mathrm{l}}$}}\\
\cline{1-5}
$l_h\in V^L$ & \multicolumn{2}{p{3.4cm}|}{\centering{$S_h^U,T^{\uparrow},S_k^U$}} & \centering{$S_h^L$} & \centering{$T_h^L,T_k^L$} &\\
\hline
$d_h^{l_1}\in V^{\downarrow\mathrm{l}}_h$ & \multicolumn{2}{p{3.4cm}|}{\centering{$T_h^{\downarrow\mathrm{l}},S_h^L,T^{\uparrow},S_k^U$}} &  \centering{$T_h^{\downarrow\mathrm{l}},S_h^L$}  &  $\begin{array}{ll}T_h^{\downarrow\mathrm{l}} & \mbox{if $k=h$}\\ T_h^{\downarrow\mathrm{l}},T_h^L,T_k^L & \mbox{if $k\neq h$}\end{array}$ & $\begin{array}{ll}T_h^{\downarrow\mathrm{l}},S_h^L,T^{\uparrow},S_k^U,T_k^{\downarrow\mathrm{l}} & \mbox{if $k\neq h$}\\ T_h^{\downarrow\mathrm{l}} & \mbox{if $k=h$ and $l_{2}>l_{1}$}\\ T_h^{\downarrow\mathrm{l}},S_h^L,T^{\uparrow},S_k^U,T_k^{\downarrow\mathrm{r}} & \mbox{if $k=h$ and $l_{2}<l_{1}$}\end{array}$\\
\hline\hline
\end{tabular}
\end{small}
}
\caption{The different cases in the proof that the trip network $(D_{r},\tnet_{r})$, defined in the proof of Theorem~\ref{thm:pairscheduleconnected} and illustrated in Figure~\ref{fig:pairscheduleconnected}, is strongly temporalisable. For each node $u$ labeling the row and for each node $v$ labeling the column, the corresponding cell specifies which (partial) schedule $S$ has to be used in order to guarantee that $v$ is $S$-reachable from $u$ (the trips that do not appear can be scheduled in any order).}
\label{tbl:pair-scheduleconnectivity}
\end{table}

By using the (partial) schedules above, we can now easily show that, for any two nodes $u$ and $v$ in $V$, there exists a schedule $S$ such that $v$ is $S$-reachable from $u$. These schedules are specified in Table~\ref{tbl:pair-scheduleconnectivity} for each possible pair of nodes. For example, in order to reach $d_h^{l_2}\in V^{\downarrow}$ from $d_h^{l_1}\in V^{\downarrow}$ with $l_{2} < l_{1}$, we first schedule the trip $T_h^{\downarrow\mathrm{l}}$ in order to reach $l_h\in V^L$, we then use the (partial) schedule $S_h^L$ in order to reach $c_{r+1}$, we then schedule the trip $T^{\uparrow}$ in order to reach $c_{r}$, we then use the (partial) schedule $S_h^U$ in order to reach $u_h$, and we finally schedule the trip $T_h^{\downarrow\mathrm{r}}$ to reach $d_h^{l_2}$ (note how, \textit{only} in this case, we need the second descending trip). 

\medskip
\noindent\textbf{Any temporalisation $\tau$ has $O(n \sqrt{n})$ $\tau$-reachability.} Given any temporalisation $\tau$ of the trip network $(D_r,\tnet_r)$, let $T_{i_{\min}}^L$ be one of the trips with minimum starting time according to $\tau$ among all the trips in the lower gadget, and let $T_{i_{\max}}^U$ be one of the trips with maximum starting time according to $\tau$ among all the trips in the upper gadget. We will prove the following claim.
\begin{claim}
For any pair of nodes $(d^{l_1}_{h_1},d^{l_2}_{h_2})$ with $1 < l_1,l_2 < r$, $h_1,h_2\in[r]$, $h_1\neq h_2$, and $h_1\neq i_{\min} \vee h_2 \neq i_{\max}$, $d^{l_2}_{h_2}$ is not $\tau$-reachable from $d^{l_1}_{h_1}$.
\end{claim}

Note that there are $r-2$ possible choices for $l_1$ and $l_2$, $r-1$ choices for $h_1$, and at least $r-2$ choices for $h_2$. Hence, the number of pairs of nodes satisfying the conditions in the claim is at least $(r-1)(r-2)^3 > (r-1)(r^3-6r^2)=r^4 - 7r^3+6r^2>r^4 - 7r^3$, thus implying that, since $n=r^2+2r$, the $\tau$-reachability is at most $(r^2+2r)^2-(r^4 - 7r^3)=r^4+4r^3+4r^2-r^4+7r^3=11r^3+4r^2<15r^3$. Since $n=r^2+2r$, we have that $r<\sqrt{n}$, and that the $\tau$-reachability is at most $15n\sqrt{n}$. The theorem thus follows.

It thus remains to prove the claim. To this aim, let $(d^{l_1}_{h_1},d^{l_2}_{h_2})$ be such that $1 < l_1,l_2 < r$, $h_1,h_2\in[r]$, $h_1\neq h_2$, and $h_1\neq i_{\min} \vee h_2 \neq i_{\max}$. Note that, since $h_1\neq h_2$, each walk in $D_{r}$ from $d^{l_1}_{h_1}$ to $d^{l_2}_{h_2}$ has to pass through the nodes $l_{h_1}$, $c_{r+1}$, $c_{r}$, and $u_{h_2}$ in this order. Note also that any walk from $l_{h_1}$ to $c_{r+1}$ contains at least $r$ edges. Since travelling on more than one edge of the same trip in the lower gadget results in going back towards $l_{h_1}$, all the $r$ trips in $\tnet^{L}$ have to be used in order to go from $l_{h_1}$ to $c_{r+1}$. As in any temporal path from $l_{h_1}$ to $c_{r+1}$ in $G[D_{r},\tnet_{r},\tau]$, the appearance times of the temporal edges must increase, this implies that all starting times of trips in $\tnet^{L}$ must be pairwise distinct and that the trip with the earliest starting time is $T_{h_1}^{L}$ as it is the only one containing the edge $(l_{h_1},c_{2r})$. If $h_1\neq i_{min}$, then $\tau$ fails to make $c_{r+1}$ $\tau$-reachable from $l_{h_1}$. If $h_2\neq i_{max}$, a similar reasoning allows us to show that $\tau$ fails to make $u_{h_2}$ reachable from $c_r$ by considering a temporal path from $c_{r}$ to $u_{h_2}$ in $G[D_{r},\tnet_{r},\tau]$, and the trips in $\tnet^{U}$. Hence, $h_1\neq i_{\min} \vee h_2 \neq i_{\max}$ and $h_1 \neq h_2$ implies that $\tau$ fails to make $d^{l_2}_{h_2}$ $\tau$-reachable from $d^{l_1}_{h_1}$, and the claim is proved.\qed
\end{proof}

\begin{figure}[t]
    \centering{\SetVertexStyle[FillColor=white,TextFont=\normalsize]
\SetEdgeStyle[Color=black]
\begin{tikzpicture}[x=2cm,y=2cm,scale=0.6, every node/.style={scale=0.6}]
\Vertex[x=8,y=0,size=1,label=$c_1$]{c1}
\Vertex[x=6,y=0,size=1,label=$c_2$]{c2}
\Vertex[x=4,y=0,size=1,label=$c_3$]{c3}
\Vertex[x=2,y=0,size=1,label=$\cdots$,fontsize=\LARGE,fontsize=\LARGE,style={color=white},opacity=0]{c4}
\Vertex[x=0,y=0,size=1,label=$c_{r-2}$]{c5}
\Vertex[x=-2,y=0,size=1,label=$c_{r-1}$]{c6}
\Vertex[x=-4,y=0,size=1,label=$c_{r}$]{c6a}
\Vertex[x=10,y=0,size=1,label=$u_1$]{t1}
\Vertex[x=12,y=0,size=1,label=$u_2$]{t2}
\Vertex[x=14,y=0,size=1,label=$\cdots$,fontsize=\LARGE,style={color=white},opacity=0]{t4}
\Vertex[x=16,y=0,size=1,label=$u_{r-1}$]{t5}
\Vertex[x=20,y=0,size=1,label=$t$]{t6}
\Edge[Direct](t2)(t1)
\Edge[Direct](t4)(t2)
\Edge(t5)(t4)
\Edge[Direct](t6)(t5)
\Vertex[x=-4,y=-16,size=1,label=$c_{r+1}$]{c7a}
\Vertex[x=-2,y=-16,size=1,label=$c_{r+2}$]{c7}
\Vertex[x=0,y=-16,size=1,label=$c_{r+3}$]{c8}
\Vertex[x=2,y=-16,size=1,label=$\cdots$,fontsize=\LARGE,style={color=white},opacity=0]{c9}
\Vertex[x=4,y=-16,size=1,label=$c_{2r-2}$]{c10}
\Vertex[x=6,y=-16,size=1,label=$c_{2r-1}$]{c11}
\Vertex[x=8,y=-16,size=1,label=$c_{2r}$]{c12}
\Vertex[x=10,y=-16,size=1,label=$l_1$]{b1}
\Vertex[x=12,y=-16,size=1,label=$l_2$]{b2}
\Vertex[x=14,y=-16,size=1,label=$\cdots$,fontsize=\LARGE,style={color=white},opacity=0]{b4}
\Vertex[x=16,y=-16,size=1,label=$l_{r-1}$]{b5}
\Vertex[x=20,y=-16,size=1,label=$l_{r}$]{b6}
\Edge[Direct](b1)(b2)
\Edge(b2)(b4)
\Edge[Direct](b4)(b5)
\Edge[Direct](b5)(b6)
\Vertex[x=10,y=-4,size=1,label=$d_1^2$]{d1-f}
\Vertex[x=12,y=-4,size=1,label=$d_2^2$]{d2-f}
\Vertex[x=14,y=-4,size=1,label=$\cdots$,fontsize=\LARGE,style={color=white},opacity=0]{d4-f}
\Vertex[x=16,y=-4,size=1,label=$d_{r-1}^2$]{d5-f}
\Vertex[x=20,y=-4,size=1,label=$x_{2}$]{d6-f}
\Vertex[x=10,y=-6,size=1,label=$d_1^3$]{d1-s}
\Vertex[x=12,y=-6,size=1,label=$d_2^3$]{d2-s}
\Vertex[x=14,y=-6,size=1,label=$\cdots$,fontsize=\LARGE,style={color=white},opacity=0]{d4-s}
\Vertex[x=16,y=-6,size=1,label=$d_{r-1}^3$]{d5-s}
\Vertex[x=20,y=-6,size=1,label=$x_{3}$]{d6-s}
\Vertex[x=10,y=-8,size=1,label=$\cdots$,fontsize=\LARGE,style={color=white},opacity=0]{cd1}
\Vertex[x=12,y=-8,size=1,label=$\cdots$,fontsize=\LARGE,style={color=white},opacity=0]{cd1}
\Vertex[x=14,y=-8,size=1,label=$\cdots$,fontsize=\LARGE,style={color=white},opacity=0]{cd1}
\Vertex[x=16,y=-8,size=1,label=$\cdots$,fontsize=\LARGE,style={color=white},opacity=0]{cd1}
\Vertex[x=20,y=-8,size=1,label=$\cdots$,fontsize=\LARGE,style={color=white},opacity=0]{cd1}
\Vertex[x=10,y=-10,size=1,label=$d_1^{r-2}$]{d1-bl}
\Vertex[x=12,y=-10,size=1,label=$d_2^{r-2}$]{d2-bl}
\Vertex[x=14,y=-10,size=1,label=$\cdots$,fontsize=\LARGE,style={color=white},opacity=0]{d4-bl}
\Vertex[x=16,y=-10,size=1,label=$d_{r-1}^{r-2}$]{d5-bl}
\Vertex[x=20,y=-10,size=1,label=$x_{p-1}$]{d6-bl}
\Vertex[x=10,y=-12,size=1,label=$d_1^{r-1}$]{d1-l}
\Vertex[x=12,y=-12,size=1,label=$d_2^{r-1}$]{d2-l}
\Vertex[x=14,y=-12,size=1,label=$\cdots$,fontsize=\LARGE,style={color=white},opacity=0]{d4-l}
\Vertex[x=16,y=-12,size=1,label=$d_{r-1}^{r-1}$]{d5-l}
\Vertex[x=20,y=-12,size=1,label=$s$]{d6-l}
\Edge[Direct](c6a)(c6)
\Edge[Direct](c6)(c5)
\Edge(c5)(c4)
\Edge[Direct](c4)(c3)
\Edge[Direct](c3)(c2)
\Edge[Direct](c2)(c1)
\Edge[Direct,bend=-45](c1)(c3)
\Vertex[x=4,y=1,size=1,Pseudo]{etd3}
\Edge[bend=-25](c2)(etd3)
\Vertex[x=2,y=1,size=1,Pseudo]{etd4}
\Edge[bend=-25](c3)(etd4)
\Edge[Direct,bend=-25](etd4)(c5)
\Vertex[x=0,y=1,size=1,Pseudo]{etd5}
\Edge[Direct,bend=-25](etd5)(c6)
\Edge[Direct,bend=-45](c5)(c6a)
\Edge(c1)(t1)
\Edge[Direct,bend=-45](c1)(t2)
\Vertex[x=13,y=-1,size=1,Pseudo]{etd2}
\Edge[bend=-35](c1)(etd2)
\Edge[Direct,bend=-45](c1)(t5)
\Edge[Direct,bend=-45](c1)(t6)
\Edge[Direct,bend=-45](t1)(c2)
\Edge[Direct,bend=-45](t2)(c2)
\Vertex[x=13,y=1,size=1,Pseudo]{etd1}
\Edge[Direct,bend=-35](etd1)(c2)
\Edge[Direct,bend=-45](t5)(c2)
\Edge[Direct,bend=-45](t6)(c2)
\Edge[Direct](c7)(c7a)
\Edge[Direct](c8)(c7)
\Edge[Direct](c9)(c8)
\Edge(c10)(c9)
\Edge[Direct](c11)(c10)
\Edge[Direct](c12)(c11)
\Edge[Direct,bend=-45](c7a)(c8)
\Vertex[x=0,y=-17,size=1,Pseudo]{ebd3}
\Edge[bend=-25](c7)(ebd3)
\Vertex[x=2,y=-17,size=1,Pseudo]{ebd4}
\Edge[bend=-25](c8)(ebd4)
\Edge[Direct,bend=-25](ebd4)(c10)
\Vertex[x=4,y=-17,size=1,Pseudo]{ebd5}
\Edge[Direct,bend=-25](ebd5)(c11)
\Edge[Direct,bend=-45](c10)(c12)
\Edge[Direct,bend=-45](c11)(b1)
\Edge[Direct,bend=-45](c11)(b2)
\Vertex[x=13,y=-17,size=1,Pseudo]{ebd2}
\Edge[bend=-35](c11)(ebd2)
\Edge[Direct,bend=-45](c11)(b5)
\Edge[Direct,bend=-45](c11)(b6)
\Edge(b1)(c12)
\Edge[Direct,bend=-45](b2)(c12)
\Vertex[x=13,y=-15,size=1,Pseudo]{ebd1}
\Edge[Direct,bend=-35](ebd1)(c12)
\Edge[Direct,bend=-45](b5)(c12)
\Edge[Direct,bend=-45](b6)(c12)
\Edge[Direct](c7a)(c6a)
\Edge[Direct](t1)(d1-f)
\Edge[Direct](d1-f)(d1-s)
\Vertex[x=10,y=-7.5,Pseudo]{td1-1}
\Edge(d1-s)(td1-1)
\Edge[Direct](t2)(d2-f)
\Edge[Direct](d2-f)(d2-s)
\Vertex[x=12,y=-7.5,Pseudo]{td2-1}
\Edge(d2-s)(td2-1)
\Edge[Direct](t5)(d5-f)
\Edge[Direct](d5-f)(d5-s)
\Vertex[x=16,y=-7.5,Pseudo]{td5-1}
\Edge(d5-s)(td5-1)
\Edge(t6)(d6-f)
\Edge(d6-f)(d6-s)
\Vertex[x=20,y=-7.5,Pseudo]{td6-1}
\Edge(d6-s)(td6-1)
\Edge[Direct](d1-l)(b1)
\Edge[Direct](d1-bl)(d1-l)
\Vertex[x=10,y=-8.5,Pseudo]{bd1-2}
\Edge[Direct](bd1-2)(d1-bl)
\Edge[Direct](d2-l)(b2)
\Edge[Direct](d2-bl)(d2-l)
\Vertex[x=12.25,y=-8.5,Pseudo]{bd2-1}
\Vertex[x=12,y=-8.5,Pseudo]{bd2-2}
\Edge[Direct](bd2-2)(d2-bl)
\Edge[Direct](d5-l)(b5)
\Edge[Direct](d5-bl)(d5-l)
\Vertex[x=16,y=-8.5,Pseudo]{bd5-2}
\Edge[Direct](bd5-2)(d5-bl)
\Edge(d6-l)(b6)
\Edge(d6-bl)(d6-l)
\Vertex[x=20,y=-8.5,Pseudo]{bd6-2}
\Edge(bd6-2)(d6-bl)
\node[thick,cloud, cloud puffs=25, cloud ignores aspect, minimum width=6cm, minimum height=16cm, align=center, draw] (cloud) at (20cm, -6cm) {};
\node at (10,-3.75) {\Large{($D=(V,E),\mbox{$\mathbb{T}$}$)}};
\end{tikzpicture}}
    \caption{The reduction from \otomrtt{} to \mrtt{} used in the proof of Theorem~\ref{th:inapprox-mrtt} (the edges in  $D$ are not shown, unless they coincide with one of the shown edges). All the edges without arrows are present in both directions (for instance, both $(t,x_{2})$ and $(x_{2},t)$ are included in the set of edges, while only $(t,u_{r-1})$ is included in the set of edges).} 
    \label{fig:stronglytemporalisablenonapproximability}
\end{figure}
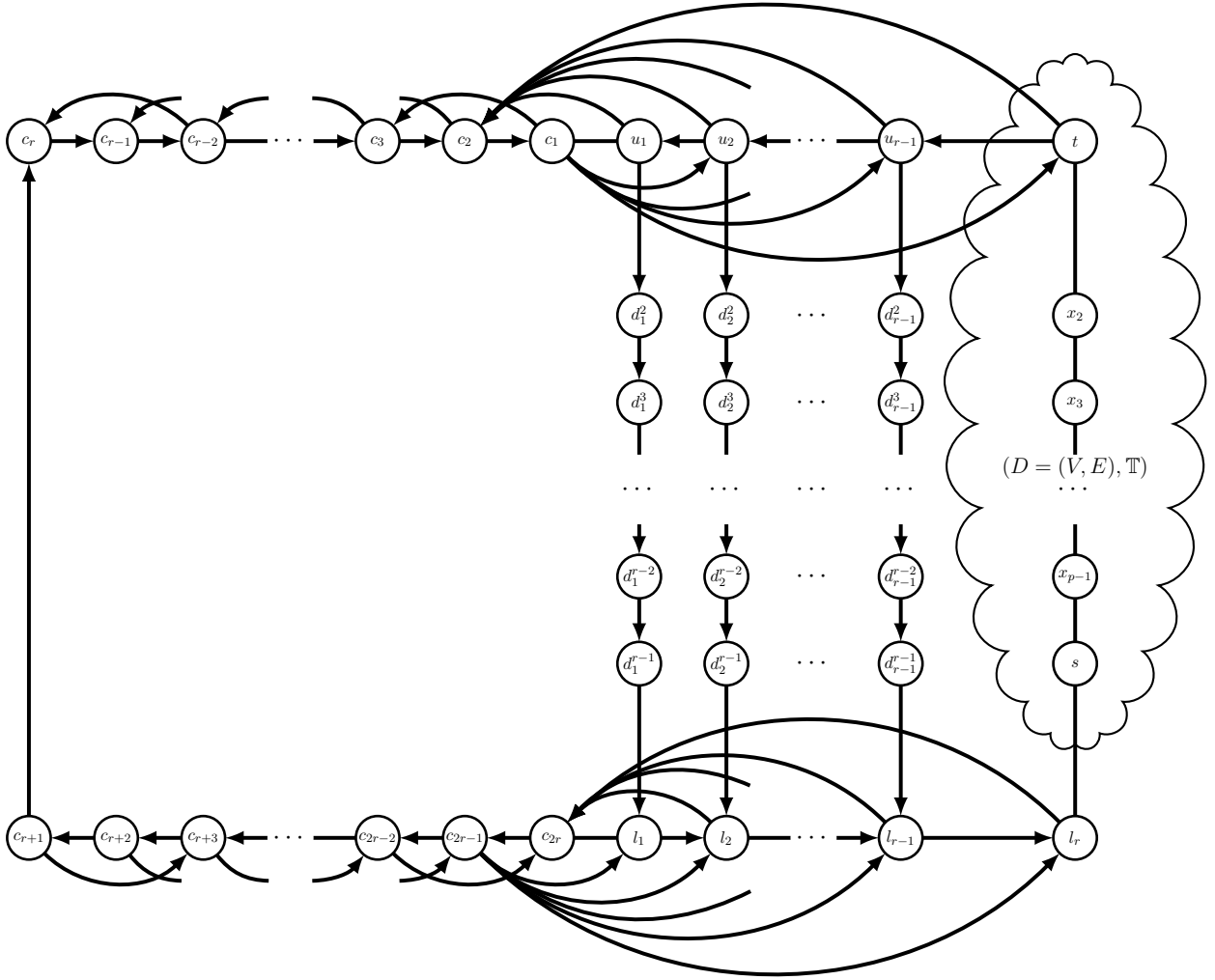

The construction of the trip network $(D_{r},\tnet_{r})$ in the proof of the above theorem can be adapted in order to prove the following inapproximability results for both the \mrtt{} and the \ssmrtt{} problem, in the case of strongly temporalisable trip networks.

\begin{theorem}\label{th:inapprox-mrtt}
  Unless $P=NP$, the \mrtt{} problem cannot be approximated within a factor less than $\frac{\sqrt{n}}{12}$ even if the input trip network is strongly temporalisable.
\end{theorem}

\begin{proof}
  As in the proof of Theorem~\ref{thm:mrtthard}, we use the gap technique by reducing in polynomial time the \otomrtt{} decision problem to the \mrtt{} problem. Consider an instance $\langle(D=(V,E),\tnet),s,t\rangle$ of the \otomrtt{} problem, where $V=\{t=x_1,\ldots,x_p=s\}$ (without loss of generality, we assume that $p>22$). We then define a trip network $(D'=(V',E'),\tnet')$ as follows (see Figure~\ref{fig:stronglytemporalisablenonapproximability}). Let $(D_{r}=(V_{r},E_{r}),\tnet_r)$, with $r=p+1$, be the trip network constructed in the proof of Theorem~\ref{thm:pairscheduleconnected} (note that $r>23$): in the following, in order to more easily define $E'$ we identify each node $x_{i}\in V$ with the node $d_{r}^{i}$ of the last descending gadget of $D_{r}$. Note that $l_r=d_{r}^{r}$ is not a node in $V$.
  We then set $V'=V_r$ and
  \begin{eqnarray*}
  E' & = & E_r\cup \{(d_r^i,d_r^j) : (x_i,x_j)\in E\} \cup \{(d_r^{i+1},d_r^{i}) : i\in [r-1]\}\\
  & & \cup \{(u_{i+1},u_i) : i\in [r-1]\}\cup \{(l_i,l_{i+1}) : i\in [r-1]\}\cup \{(u_{1},c_{1}), (c_{2r},l_{1})\}
  \end{eqnarray*}
  Note that, according to the definition of $V'$ and $E'$, each trip in $\tnet$ can be considered as a walk in $D'$. We then set \[\tnet'=\tnet\cup\tnet_r\setminus\{T_r^U,T_r^L,T_r^{\downarrow\mathrm{l}},T_r^{\downarrow\mathrm{r}}\}\cup\{T_U,T_L,T_{\uparrow\downarrow}\},\]
  where $T_U,T_L,T_{\uparrow\downarrow}$ are the following three trips.
  
  \begin{itemize}
  \item $T_U=\langle t,u_{r-1},\dots,u_1,c_1,t,c_2, c_1, c_3, c_2, \dots, c_{r - 1},c_{r - 2}, c_{r}$, $c_{r-1}\rangle$ (intuitively, $T_{U}$ replaces $T^U_r$: it first visits $t=u_{r},\ldots,u_1$, it then goes to $c_{1}$, and it finally continues exactly as $T^U_r$).
  
  \item $T_L=\langle c_{r+2},c_{r+1}, c_{r+3}, c_{r+2},\dots,c_{2r-2},c_{2r}$, $c_{2r -1}, l_r, c_{2r},l_1,l_2,\ldots,l_r\rangle$ (intuitively, $T_{L}$ replaces $T^L_r$: it first starts exactly as $T^L_r$, and it then visits $l_1,\ldots,l_r$).
  
  \item $T_{\uparrow\downarrow}=\langle s,d_r^{r-2},\ldots,t,d_r^2,\ldots,d_r^{r-2},s,l_r,s\rangle$ (intuitively, $T_{\uparrow\downarrow}$ replaces both $T_r^{\downarrow\mathrm{l}}$ and $T_r^{\downarrow\mathrm{r}}$).
  \end{itemize}

\medskip
\noindent\textbf{$(D',\tnet')$ is strongly temporalisable} (even if $(D,\tnet)$ is not). The proof is similar to the one proving that $(D_r,\tnet_r)$ is strongly temporalisable (see the proof of Theorem~\ref{thm:pairscheduleconnected}). Indeed, whenever $h=r$ or $k=r$ in Table~\ref{tbl:pair-scheduleconnectivity}, we can replace $T^U_h$ or $T^U_k$ by $T_U$ (respectively, $T^L_h$ or $T^L_k$ by $T_L$ and $T_h^{\downarrow\mathrm{l}}$ or $T_k^{\downarrow\mathrm{l}}$ by $T_{\uparrow\downarrow}$) in the schedules included in the table. Since $T^U_r$ (respectively, $T^L_r$ and $T_r^{\downarrow\mathrm{l}}$) is included in $T_U$ (respectively, $T_L$ and $T_{\uparrow\downarrow}$), by doing so we have that all the reachability properties of the table are still satisfied apart from the last case of the last cell of the table itself, with $h=k=r$. However, in this case, for any temporalisation $\tau$, we have that $d_r^{l_2}\in V^{\downarrow}$ is $\tau$-reachable from $d_r^{l_1}\in V^{\downarrow}$ with $l_{2} < l_{1}$, since we can just use the trip $T_{\uparrow\downarrow}$ (which allows to go up from $d_r^{l_1}$ to $d_r^{l_2}$).

\medskip
\textbf{If $(D,\tnet)$ is $(s,t)$-temporalisable, then there exists a temporalisation $\tau'$ of $(D',\tnet')$ with $\tau'$-reachability at least $((r-1)r)^2$}. To this aim, first note that, for any temporalisation $\tau'$ of $(D',\tnet')$, if $\tau'(T_{\uparrow\downarrow}) = t'$, then the trip $T_{\uparrow\downarrow}$ arrives at $l_{r}$ at time $t'+2r-3$, and terminates in $s$ at time $t'+2r-2$. If $(D,\tnet)$ is $(s,t)$-temporalisable, then there exists a temporalisation $\tau$ of $(D,\tnet)$ such that $t$ is $\tau$-reachable from $s$. Let $P$ be any temporal path in $G[D,\tnet,\tau]$ from $s$ to $t$, and let $t_{s}$ (respectively, $t_{a}$) be the starting (respectively, arrival) time of $P$ from $s$ (respectively, in $t$). We then define a temporalisation $\tau'$ of $(D',\tnet')$ as follows. For any $T\in\tnet{}$, $\tau'(T)=\tau(T)$. Moreover, for any $h\in[r-1]$, $\tau'(T_{h}^{\downarrow\mathrm{l}})=t_{s}-2r+1$, and $\tau'(T_{L})=t_s-1-|T_L|=t_s-3r$. This allows the trips $T_{h}^{\downarrow\mathrm{l}}$ to ``meet'' the trip $T_{L}$ in $l_{h}$ at time $t_{s}-(r-h)-1$: note that $T_{L}$ arrives in $l_r$ at time $t_{s}-1$. Hence, we set $\tau'(T_{\uparrow\downarrow})=t_{s}-2r+2$ so that $T_{\uparrow\downarrow}$ arrives in $l_{r}$ also at time $t_{s}-1$. By using the (last edge of the) trip $T_{\uparrow\downarrow}$, and then the path $P$, we can arrive in $t$ at time $t_{a}$. By setting $\tau'(T_{U})=t_{a}$ and, for any $h\in[r-1]$, $\tau'(T_{h}^{\downarrow\mathrm{r}})=t_{a}+r-h$, we can arrive at any node $d_{h}^{k}$ at time $t_{a}+r-h+k-1$. In other words, we have shown that all nodes of the first $r-1$ descending gadgets are $\tau'$-reachable one from the other. That is, the $\tau'$-reachability is at least $((r-1)r)^2$.

\medskip
\textbf{If $(D,\tnet)$ is not $(s,t)$-temporalisable, then the $\tau'$-reachability of any temporalisation $\tau'$ of $(D',\tnet')$ is at most $3rn+7r^2(r-1)$}. Let $\tau'$ be a temporalisation of $(D',\tnet')$. First note that $\tau'$ induces a temporalisation $\tau$ of $(D,\tnet)$. Since $(D,\tnet)$ is not $(s,t)$-temporalisable, we have that $t$ is not $\tau$-reachable from $s$. Moreover, since the edge $(l_{r}, s)$ is the last edge in the trip $T_{\uparrow\downarrow}$, it cannot be used before the other edges in this trip. As a consequence, $t$ is not $\tau''$-reachable from $l_r$, where $\tau''$ is the temporalisation induced by $\tau'$ on $(D',\tnet\cup \{T_{\uparrow\downarrow}\})$. The topology of $D'$ thus implies that, for all $i,j\in [r]$, all temporal paths from $l_i$ to $u_j$ in $G[D',\tnet',\tau']$ must pass through the nodes $c_{2r},c_{2r-1},\ldots,c_{r+1}$, the edge $(c_{r+1},c_r)$, and the nodes $c_{r},c_{r-1},\ldots,c_{1}$ .

Let $T_{i_{\min}}^L$ be one of the trips with minimum starting time according to $\tau'$ among all the trips in the lower gadget ($i_{\min}=r$ if $T_L$ is the only trip with minimum starting time), and let $T_{i_{\max}}^U$ be one of the trips with maximum starting time according to $\tau'$ among all the trips in the upper gadget ($i_{\max}=r$ if $T_{U}$ is the only trip with maximum starting time). Similar to the proof of Theorem \ref{thm:pairscheduleconnected}, we can then show that there is no temporal path in $G[D',\tnet',\tau']$  from $l_h$ to $c_{r+1}$ for $h\in [r]\setminus\{i_{\min}\}$ nor from $c_r$ to $u_k$ for $k\in [r]\setminus\{i_{\max}\}$. Note that the additional part $\langle c_{2r},l_1,\ldots,l_r\rangle$ of $T_L$ (respectively, $\langle t,u_{r-1},\dots,u_1,c_1\rangle$ of $T_U$), compared to $T^L_r$ (respectively, $T^{U}_{r}$), does not change the reasoning as it comes at the end (respectively, the beginning) of the trip and, in particular, after edge $(l_r,c_{2r})$ (respectively, before the edge $(c_{1},l_{1})$).

We can thus similarly conclude that $d_{h_{2}}^{l_{2}}$ is not $\tau'$-reachable from $d_{h_{1}}^{l_{1}}$ for all $h_{1},h_{2}\in [r-1]$ with $h_{1}\neq i_{\min}$ or $h_{2}\neq i_{\max}$ and all $l_{1},l_{2}$ with $1<l_{1},l_{2}<r$. Note that the situation is different from the previous construction for $l_{1}=l_{2}=r$ or $l_{1}=l_{2}=1$ as $T_L$ makes $d_{h_2}^{r}=l_{h_{2}}$ $\tau'$-reachable from $d_{h_1}^{r}=l_{h_{1}}$ for $h_{1}<h_{2}$ and that $T_U$ makes $d_{h_{1}}^{1}=u_{h_1}$ $\tau'$-reachable from $d_{h_{2}}^{1}=u_{h_{2}}$ for $h_{1}<h_{2}$. Overall, this means that only nodes in
\[\{d_{h_1}^{1},\ldots d_{h_1}^{r},l_{1},\ldots,l_{r},c_{1},\ldots,c_{2r},u_{1},\ldots,u_{r},d_{i_{\max}}^{1},\ldots,d_{i_{\max}}^{r},d_{r}^{1},\ldots,d_{r}^{r}\}\]
can be $\tau'$-reachable from $d_{h_{1}}^{l_{1}}$ for $h_{1}\in [r-1]$ and $l_{1}\in [r]$. Thus, the nodes in the first $r-1$ descending gadgets have $\tau'$-reachability at most $7r$. The other $3r$ nodes have  $\tau'$-reachability at most $n$.

\medskip
\textbf{The MRTT problem cannot be approximated within a factor less than $\frac{\sqrt{n}}{12}$}. We now prove that any polynomial-time algorithm $\mathcal{A}$ solving \mrtt{} with approximation ratio $\rho < \frac{\sqrt{n}}{12}$ would allow us to decide in polynomial-time whether $(D,\tnet)$ is $(s,t)$-temporalisable. As this latter problem is NP-complete, as stated by Theorem~\ref{thm:otomrtthard}, this will conclude the proof of the theorem. If $(D,\tnet)$ is $(s,t)$-temporalisable, then then there exists a temporalisation of $(D',\tnet')$ whose reachability is at least $((r-1)r)^2$. This implies that $\mathcal{A}$, with input the trip network $(D',\tnet')$, has to provide a temporalisation $\tau'$ with $\tau'$-reachability at least $\frac{((r-1)r)^2}{\rho} > \frac{12r^2(r-1)^2}{r+1} > 11r^2(r-1)$ (note that $\sqrt{n}=\sqrt{r^2+2r}<r+1$ and $\frac{r-1}{r+1}=1-\frac{2}{r+1} > \frac{11}{12}$ for $r>23$). On the other hand, if $(D,\tnet)$ is not $(s,t)$-temporalisable, then the $\tau'$-reachability of any temporalisation $\tau'$ of $(D',\tnet')$ is at most $3rn+7r^2(r-1) = 10r^3-r^2 < 11r^2(r-1)$ (note that $r^3 > 10r^2$ for $r>10$). In summary, $(D,\tnet)$ is $(s,t)$-temporalisable if and only if $\mathcal{A}$, with input the trip network $(D',\tnet')$, returns a temporalisation whose reachability is greater than $11r^2(r-1)$.\qed
\end{proof}

\begin{theorem}\label{th:inapprox-ssmrtt}
  Unless $P=NP$, the \ssmrtt{} problem cannot be approximated within a factor less than $\frac{\sqrt{n}}{12}$ even if the input trip network is strongly temporalisable.
\end{theorem}

\begin{proof}
The reduction is exactly the same as the one used in the proof of the previous theorem. According to that proof, if $(D,\tnet)$ is $(s,t)$-temporalisable, then there exists a temporalisation $\tau'$ of $(D',\tnet')$ such that any source in one of the first $r-1$ descending gadgets has $\tau'$-reachability at least $r(r-1)$. On the other hand, if $(D,\tnet)$ is not $(s,t)$-temporalisable, then, for any temporalisation $\tau'$ of $(D',\tnet')$, any source in one of the first $r-1$ descending gadgets has $\tau'$-reachability at most $7r$. Any polynomial-time algorithm $\mathcal{A}$ solving \ssmrtt{} with approximation ratio $\rho < \frac{\sqrt{n}}{12}$ would then allow us to decide, in polynomial-time, whether $(D,\tnet)$ is $(s,t)$-temporalisable, since $\frac{r(r-1)}{\rho}>\frac{12r(r-1)}{r+1}>11r>7r$ for $r>23$. This concludes the proof of the theorem.\qed
\end{proof}

\subsection{Symmetric and strongly temporalisable trip networks}
\label{sec:symmetricstronglytempooralisable}

Because of the last theorem, we now focus on symmetric \textit{and} strongly temporalisable trip networks.

\begin{fact}\label{obs:one-to-all}
Let $(D,\tnet)$ be a symmetric trip network. For any node $u$, there exists a schedule $S$ of $(D,\tnet)$ such that, for any node $v$ reachable from $u$ in $D$, $v$ is $S$-reachable from $u$.
\end{fact}

\begin{proof}
We first note that, due to the symmetricity, any schedule $S$ of $(D,\tnet)$ is such that each node of a trip $T \in \tnet{}$ is $S$-reachable from any other node of $T$. Let $\tree{u}$ be a breadth-first search tree in the induced multidigraph $M$ rooted at $u$, whose height is $h_{\tree{u}}$. By using $\tree{u}$, we will now define a schedule $S$ of $(D,\tnet)$ such that, for any node $v$ in $\tree{u}$, $v$ is $S$-reachable from $u$. This will prove the fact. 

We will construct nested partial schedules $S_0,\ldots,S_{h_{\tree{u}}}=S$ where a larger and larger subset of $\tnet{}$ is scheduled. Given a partial schedule $S_\ell$, we say that an edge is ``covered'' by $S_\ell$ if it belongs to one of the trips scheduled in $S_\ell$. Similarly, a node is ``covered'' by $S_\ell$ if it is the head or the tail of an edge covered by $S_\ell$. At the beginning, we consider an empty schedule $S_0$. For each level $\ell$ of $\tree{u}$ with $\ell\in[h_{\tree{u}}]$, let $e_{1},\ldots,e_{k_\ell}$ be the edges connecting a node at level $\ell-1$ to a node at level $\ell$ which are not yet covered by $S_{\ell-1}$. Let $T_{\ell,1},\ldots,T_{\ell,k_\ell}$ be $k_\ell$ (not necessarily distinct) trips in \tnet{}, which contain the edges $e_{1},\ldots,e_{k_\ell}$, respectively.  Recall that $\symtrip{$T$}_{\ell,1},$\ldots$,\symtrip{$T$}_{\ell,k_\ell}$ denote their respective reverse trips. Let $\tnet_\ell$ denote the set of trips in $\{T_{\ell,1}, \symtrip{$T$}_{\ell,1},\ldots,T_{\ell,k_\ell},\symtrip{$T$}_{\ell,k_\ell}\}$. As the edges $e_1,\ldots,e_{k_\ell}$ were not covered by $S_{\ell-1}$, trips in $\tnet_\ell$ are not included in $S_{\ell-1}$. We can thus define $S_\ell$ as $S_{\ell-1}$ followed by the trips in $\tnet_\ell$ in an arbitrary order. Note that all edges from level $\ell-1$ to level $\ell$ are now covered by $S_\ell$. We continue similarly for the next levels and define $S=S_{h_{\tree{u}}}$ as the schedule obtained for the last layer.

To prove that any node $v$ in $\tree{u}$ is $S$-reachable, we show that the following invariant is preserved: after processing each level $\ell$, any node covered by $S_\ell$ is $S_\ell$-reachable. Consider an edge $e$ which is covered by $S_\ell$ but not by $S_{\ell-1}$ and let $v'$ be any extremity (head or tail) of $e$ which is not covered by $S_{\ell-1}$.
%its head or tail.
Let $T_{\ell,i},\symtrip{$T$}_{\ell,i}$ denote the trip pair in $\tnet_\ell$ that contains $e$. Consider the edge $e_i$ from level $\ell-1$ to $\ell$ that belongs to $T_{\ell,i}$.  The tail $u'$ of $e_i$ is either $u$ or the head of an edge from level $\ell-2$ to $\ell-1$. As such an edge is covered by $S_{\ell-1}$, $u'$ is thus $S_{\ell-1}$-reachable according to the invariant. As $T_{\ell,i}$ or $\symtrip{$T$}_{\ell,i}$ contains a walk from $u'$ to $v'$, and both $T_{\ell,i}$ and $\symtrip{$T$}_{\ell,i}$ are scheduled after $S_{\ell-1}$ in $S_\ell$, $v'$ is $S_\ell$-reachable. The conclusion follows from the fact that all nodes in $\tree{u}$ are covered by $S=S_{h_{\tree{u}}}$.\qed
\end{proof}

\begin{fact}\label{obs:all-to-one}
Let $(D,\tnet)$ be a symmetric trip network. For any node $u$, there exists a schedule $S$ of $(D,\tnet)$ such that, for any node $v$ such that $u$ is reachable from $v$ in $D$, $u$ is $S$-reachable from $v$.
\end{fact}

\begin{proof}
The proof is similar to the proof of Fact~\ref{obs:one-to-all}.\qed
\end{proof}
\begin{corollary}\label{cor:connected}
Let $(D,\tnet)$ be a symmetric trip network. Then, $(D,\tnet)$ is strongly temporalisable if and only if $D$ is strongly connected.
\end{corollary}

\begin{proof}
If $(D,\tnet)$ is strongly temporalisable, then, for any two nodes $u$ and $v$ in $D$, there exists a temporalisation $\tau_{u,v}$ of $(D,\tnet)$, such that $v$ is $\tau_{u,v}$-reachable from $u$. In other words, there exists a temporal path $P_{u\rightarrow v}$ from $u$ to $v$ in $G[D,\tnet, \tau_{u,v}]$. Hence, $v$ is reachable from $u$ in $D$. The converse implication is a direct consequence of Fact~\ref{obs:one-to-all}.\qed
\end{proof}

We now prove that the \mrtt{} problem remains NP-hard, even when we assume that the underlying multidigraph is strongly connected and that the trip network is symmetric (from the previous corollary, this implies that the trip networks is strongly temporalisable).

\begin{theorem}\label{th:sym-np}
The \mrtt{} problem problem is NP-hard, even if $(D,\tnet)$ is a symmetric trip network and $D$ is strongly connected.
\end{theorem}

\begin{proof}
We reduce \tsat{} to \mrtt{} as follows. Let us consider a \tsat{} formula $\Phi$, with $n$ variables $x_1, \dots, x_n$ and $m$ clauses $c_1, \dots, c_m$. Without loss of generality, we will assume that each variable appears positive in at least one clause and negative in at least one clause, that no literal appears in all clauses, and that there are at least two clauses. We set $l=\lceil(7n+m(m+3))^2/(m+2)\rceil+1$ and $L=(7n+m(m+3))^2+1$, and we define the digraph $D=(V,E)$ as the union of the following gadgets (in the following, for each edge included in $E$, its reverse edge is implicitly also included in $E$).

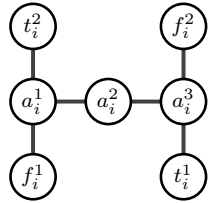
\begin{figure}[ht]
\centering{\SetVertexStyle[FillColor=white]
\begin{tikzpicture}
  \Vertex[x=2,y=0,label=$t_{i}^{1}$]{ti1}
  \Vertex[x=2,y=1,label=$a_{i}^{3}$]{ai3}
  \Vertex[x=2,y=2,label=$f_{i}^{2}$]{fi2}
  \Vertex[x=1,y=1,label=$a_{i}^{2}$]{ai2}
  \Vertex[x=0,y=1,label=$a_{i}^{1}$]{ai1}
  \Vertex[x=0,y=0,label=$f_{i}^{1}$]{fi1}
  \Vertex[x=0,y=2,label=$t_{i}^{2}$]{ti2}
  \Edge (ti1)(ai3)
  \Edge (ai3)(fi2)
  \Edge (ai3)(ai2)
  \Edge (ai2)(ai1)
  \Edge (fi1)(ai1)
  \Edge (ai1)(ti2)
\end{tikzpicture}}
\caption{The variable gadget in the reduction of \tsat{} to symmetric \mrtt{} (see the proof of Theorem~\ref{th:sym-np}).}
\label{fig:symvariablegadget}
\end{figure}

\begin{figure}[ht]
\centering{\SetVertexStyle[FillColor=white,TextFont=\scriptsize,MinSize=0.75\DefaultUnit]
\begin{tikzpicture}
  \Vertex[x=2.2,y=-0.2,label=$t_{i_1}^{1}$]{ti11}
  \Vertex[x=2.2,y=1,label=$a_{i_1}^{3}$]{ai13}
  \Vertex[x=2.2,y=2.2,label=$f_{i_1}^{2}$]{fi12}
  \Vertex[x=1,y=1,label=$a_{i_1}^{2}$]{ai12}
  \Vertex[x=-0.2,y=1,label=$a_{i_1}^{1}$]{ai11}
  \Vertex[x=-0.2,y=-0.2,label=$f_{i_1}^{1}$]{fi11}
  \Vertex[x=-0.2,y=2.2,label=$t_{i_1}^{2}$]{ti12}
  \Edge[style={dotted}](ti11)(ai13)
  \Edge[style={dotted}](ai13)(fi12)
  \Edge[style={dotted}](ai13)(ai12)
  \Edge[style={dotted}](ai12)(ai11)
  \Edge[style={dotted}](fi11)(ai11)
  \Edge[style={dotted}](ai11)(ti12)
  \begin{scope}[xshift=4cm]
  \Vertex[x=2.2,y=-0.2,label=$t_{i_2}^{1}$]{ti21}
  \Vertex[x=2.2,y=1,label=$a_{i_2}^{3}$]{ai23}
  \Vertex[x=2.2,y=2.2,label=$f_{i_2}^{2}$]{fi22}
  \Vertex[x=1,y=1,label=$a_{i_2}^{2}$]{ai22}
  \Vertex[x=-0.2,y=1,label=$a_{i_2}^{1}$]{ai21}
  \Vertex[x=-0.2,y=-0.2,label=$f_{i_2}^{1}$]{fi21}
  \Vertex[x=-0.2,y=2.2,label=$t_{i_2}^{2}$]{ti22}
  \Edge[style={dotted}](ti21)(ai23)
  \Edge[style={dotted}](ai23)(fi22)
  \Edge[style={dotted}](ai23)(ai22)
  \Edge[style={dotted}](ai22)(ai21)
  \Edge[style={dotted}](fi21)(ai21)
  \Edge[style={dotted}](ai21)(ti22)
  \begin{scope}[xshift=4cm]
  \Vertex[x=2.2,y=-0.2,label=$t_{i_3}^{1}$]{ti31}
  \Vertex[x=2.2,y=1,label=$a_{i_3}^{3}$]{ai33}
  \Vertex[x=2.2,y=2.2,label=$f_{i_3}^{2}$]{fi32}
  \Vertex[x=1,y=1,label=$a_{i_3}^{2}$]{ai32}
  \Vertex[x=-0.2,y=1,label=$a_{i_3}^{1}$]{ai31}
  \Vertex[x=-0.2,y=-0.2,label=$f_{i_3}^{1}$]{fi31}
  \Vertex[x=-0.2,y=2.2,label=$t_{i_3}^{2}$]{ti32}
  \Edge[style={dotted}](ti31)(ai33)
  \Edge[style={dotted}](ai33)(fi32)
  \Edge[style={dotted}](ai33)(ai32)
  \Edge[style={dotted}](ai32)(ai31)
  \Edge[style={dotted}](fi31)(ai31)
  \Edge[style={dotted}](ai31)(ti32)
  \end{scope}
  \end{scope}
  \begin{scope}[xshift=5cm]
  \Vertex[x=0,y=-1.5,label=$c_{j}^{1}$]{cj1}
  \Vertex[x=0,y=3.5,label=$c_{j}^{2}$]{cj2}
  \Edge (cj1)(ti11)
  \Edge (cj1)(fi21)
  \Edge (cj1)(ti31)
  \Edge (cj2)(ti12)
  \Edge (cj2)(fi22)
  \Edge (cj2)(ti32)
  \end{scope}
  \begin{scope}[xshift=2cm,yshift=3.5cm]
  \Vertex[x=0,y=-0.25,label=$e_{j}^{i_1}$]{ej1}
  \Vertex[x=0,y=0.6,label=$e_{j}^{i_2}$]{ej2}
  \Vertex[x=0,y=1.45,label=$e_{j}^{i_3}$]{ej3}
  \Edge (cj2)(ej1)
  \Edge (cj2)(ej2)
  \Edge (cj2)(ej3)
  \end{scope}
  \begin{scope}[xshift=5cm,yshift=5cm]
  \Vertex[x=0,y=0,label=$g_{j}^{1}$]{gj1}
  \Vertex[x=0.75,y=0,label=$\cdots$,style={color=white},size=0.4]{gji1}
  \Vertex[x=1.5,y=0,label=$g_{j}^{j-1}$]{gjj-1}
  \Vertex[x=2.5,y=0,label=$g_{j}^{j+1}$]{gjj+1}
  \Vertex[x=3.25,y=0,label=$\cdots$,style={color=white},size=0.4]{gji2}
  \Vertex[x=4,y=0,label=$g_{j}^{m}$]{gjm}
  \Edge (cj2)(gj1)
  \Edge (cj2)(gjj-1)
  \Edge (cj2)(gjj+1)
  \Edge (cj2)(gjm)
  \end{scope}
  \begin{scope}[xshift=-0.5cm,yshift=-1.5cm]
  \Vertex[x=1,y=0,label=$d_{j}^{1}$]{dj1}
  \Vertex[x=2,y=0,label=$d_{j}^{2}$]{dj2}
  \Vertex[x=3,y=0,label=$\cdots$,style={color=white},size=0.4]{dji}
  \Vertex[x=4,y=0,label=$d_{j}^{l}$]{djk}
  \Edge (cj1)(djk)
  \Edge (djk)(dji)
  \Edge (dji)(dj2)
  \Edge (dj2)(dj1)
  \end{scope}
\end{tikzpicture}}
\caption{The clause gadget in the reduction of \tsat{} to symmetric \mrtt{} (see the proof of Theorem~\ref{th:sym-np}), corresponding to the clause $c_j=x_{i_1}\vee \neg x_{i_2}\vee x_{i_3}$ (the dotted edges are included in the variable gadgets corresponding to the variables $x_{i_1}$, $x_{i_2}$, and $x_{i_3}$). Note that, for each $h\in[m]$ with $h\neq j$, $E$ includes also the edge $(c_{j}^{1}, c_{h}^{2})$ (and its reverse edge).}
\label{fig:symclausegadget}
\end{figure}
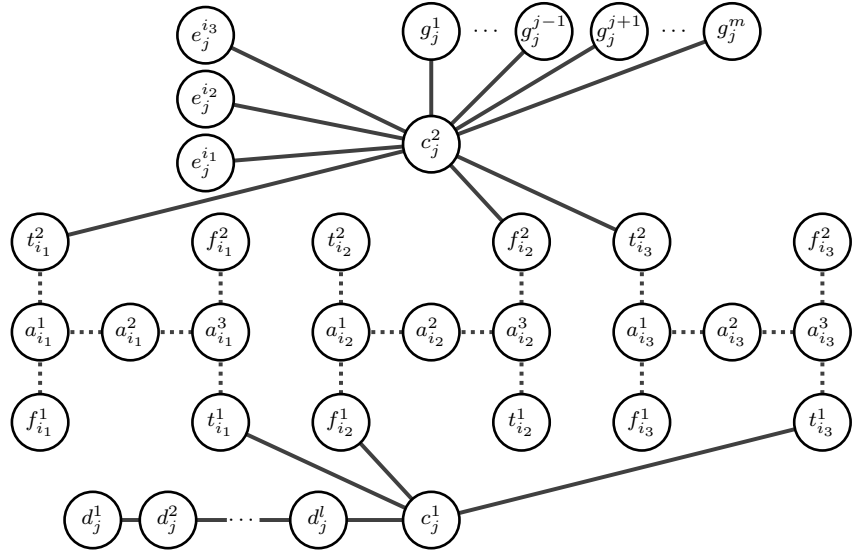

\begin{description}
\item[\bf Variable gadgets] (see Figure~\ref{fig:symvariablegadget}). For each variable $x_i$ of $\Phi$ with $i\in[n]$, $V$ contains the seven \emph{variable nodes} $t_i^1$, $t_i^2$, $f_i^1$, $f_i^2$, $a_i^1$, $a_i^2$, and $a_i^3$, and $E$ contains the six edges $(t_i^1,a_i^3)$, $(a_i^3,f_i^2)$, $(a_i^1,a_i^2)$, $(a_i^2,a_i^3)$, $(f_i^1,a_i^1)$, and $(a_i^1,t_i^2)$. 

\item[\bf Clause gadgets] (see Figure~\ref{fig:symclausegadget}). For each clause $c_j$ of $\Phi$ with $j\in[m]$, $V$ contains the two \emph{clause nodes} $c_j^1$ and $c_j^2$ (we will call $c_1^1,\ldots,c_m^1$ the \emph{bottom} clause nodes and $c_1^2,\ldots,c_m^2$ the \emph{top} clause nodes), and the \textit{middle nodes} $d_j^k$ for $k \in [l]$. For each variable $x_i$ which appears (positive or negative) in $c_j$, $V$ contains the \textit{head} node $e_j^i$. Finally, for each $h\in[m]$ with $h \neq j$, $V$ contains the \textit{head} node $g_j^h$. Concerning the edges, for each variable $x_i$ which appears positive in $c_j$, $E$ contains the edges $(c_j^1,t_i^1)$ and $(t_i^2,c_j^2)$, while, for each variable $x_i$ which appears negative in $c_j$, $E$ contains the edges $(c_j^1,f_i^1)$ and $(f_i^2,c_j^2)$. In both cases, $E$ contains the edge $(c_j^2,e_j^i)$. For each $h\in[m]$ with $h \neq j$, $E$ contains the edges $(c_j^1,c_h^2)$ and $(c_{j}^{2},g_{j}^{h})$. Finally, $E$ contains the edges $(d_j^k,d_j^{k+1})$, for $k \in [l-1]$, and the edge $(d_j^{l},c_j^1)$.

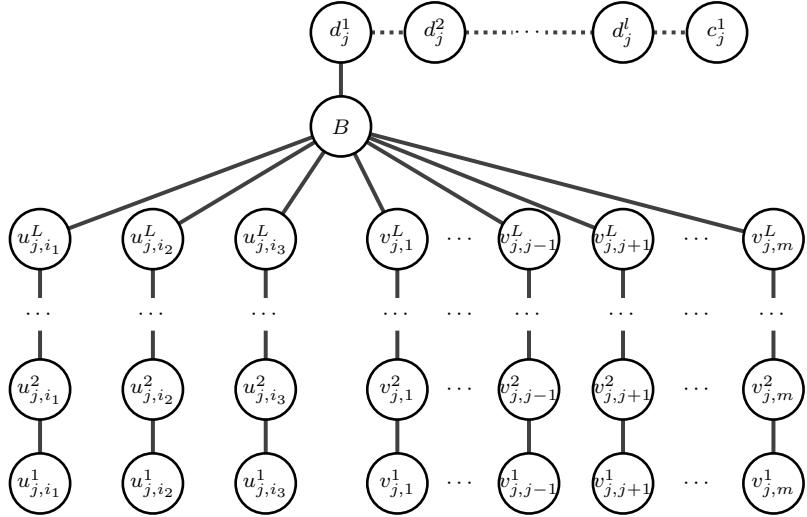
\begin{figure}[ht]
\centering{\SetVertexStyle[FillColor=white,TextFont=\scriptsize,MinSize=0.8\DefaultUnit]
\begin{tikzpicture}
  \Vertex[x=0,y=1.25,label=$d_{j}^1$]{dj1}
  \Vertex[x=1.25,y=1.25,label=$d_{j}^2$]{dj2}
  \Vertex[x=2.5,y=1.25,label=$\cdots$,style={color=white},size=0.4]{dji}
  \Vertex[x=3.75,y=1.25,label=$d_{j}^{l}$]{djk}
  \Vertex[x=5,y=1.25,label=$c_{j}^{1}$]{cj1}
  \Vertex[x=0,y=0,label=$B$]{b}
  \Edge (b)(dj1)
  \Edge[style={dotted}](dj1)(dj2)
  \Edge[style={dotted}](dj2)(dji)
  \Edge[style={dotted}](dji)(djk)
  \Edge[style={dotted}](djk)(cj1)
  \begin{scope}[xshift=-4cm,yshift=-1.5cm]
  \Vertex[x=0,y=0,label=$u_{j,i_1}^{L}$]{uji1k}
  \Vertex[x=0,y=-1,label=$\cdots$,style={color=white},size=0.4]{uji1i}
  \Vertex[x=0,y=-2,label=$u_{j,i_1}^{2}$]{uji12}
  \Vertex[x=0,y=-3.25,label=$u_{j,i_1}^{1}$]{uji11}
  \Edge (uji11)(uji12)
  \Edge (uji12)(uji1i)
  \Edge (uji1i)(uji1k)
  \Edge (uji1k)(b)
  \end{scope}
  \begin{scope}[xshift=-2.5cm,yshift=-1.5cm]
  \Vertex[x=0,y=0,label=$u_{j,i_2}^{L}$]{uji2k}
  \Vertex[x=0,y=-1,label=$\cdots$,style={color=white},size=0.4]{uji2i}
  \Vertex[x=0,y=-2,label=$u_{j,i_2}^{2}$]{uji22}
  \Vertex[x=0,y=-3.25,label=$u_{j,i_2}^{1}$]{uji21}
  \Edge (uji21)(uji22)
  \Edge (uji22)(uji2i)
  \Edge (uji2i)(uji2k)
  \Edge (uji2k)(b)
  \end{scope}
  \begin{scope}[xshift=-1cm,yshift=-1.5cm]
  \Vertex[x=0,y=0,label=$u_{j,i_3}^{L}$]{uji3k}
  \Vertex[x=0,y=-1,label=$\cdots$,style={color=white},size=0.4]{uji3i}
  \Vertex[x=0,y=-2,label=$u_{j,i_3}^{2}$]{uji32}
  \Vertex[x=0,y=-3.25,label=$u_{j,i_3}^{1}$]{uji31}
  \Edge (uji31)(uji32)
  \Edge (uji32)(uji3i)
  \Edge (uji3i)(uji3k)
  \Edge (uji3k)(b)
  \end{scope}
  \begin{scope}[xshift=0.75cm,yshift=-1.5cm]
  \Vertex[x=0,y=0,label=$v_{j,1}^{L}$]{vj1k}
  \Vertex[x=0,y=-1,label=$\cdots$,style={color=white},size=0.4]{vj1i}
  \Vertex[x=0,y=-2,label=$v_{j,1}^{2}$]{vj12}
  \Vertex[x=0,y=-3.25,label=$v_{j,1}^{1}$]{vj11}
  \Edge (vj11)(vj12)
  \Edge (vj12)(vj1i)
  \Edge (vj1i)(vj1k)
  \Edge (vj1k)(b)
  \end{scope}
  \begin{scope}[xshift=1.6cm,yshift=-1.5cm]
  \Vertex[x=0,y=0,label=$\cdots$,style={color=white},size=0.1]{vjii}
  \Vertex[x=0,y=-1,label=$\cdots$,style={color=white},size=0.1]{vjii}
  \Vertex[x=0,y=-2,label=$\cdots$,style={color=white},size=0.1]{vjii}
  \Vertex[x=0,y=-3.25,label=$\cdots$,style={color=white},size=0.1]{vjii}
  \end{scope}
  \begin{scope}[xshift=2.5cm,yshift=-1.5cm]
  \Vertex[x=0,y=0,label=$v_{j,j-1}^{L}$]{vjj-1k}
  \Vertex[x=0,y=-1,label=$\cdots$,style={color=white},size=0.4]{vjj-1i}
  \Vertex[x=0,y=-2,label=$v_{j,j-1}^{2}$]{vjj-12}
  \Vertex[x=0,y=-3.25,label=$v_{j,j-1}^{1}$]{vjj-11}
  \Edge (vjj-11)(vjj-12)
  \Edge (vjj-12)(vjj-1i)
  \Edge (vjj-1i)(vjj-1k)
  \Edge (vjj-1k)(b)
  \end{scope}
  \begin{scope}[xshift=3.75cm,yshift=-1.5cm]
  \Vertex[x=0,y=0,label=$v_{j,j+1}^{L}$]{vjjk}
  \Vertex[x=0,y=-1,label=$\cdots$,style={color=white},size=0.4]{vjji}
  \Vertex[x=0,y=-2,label=$v_{j,j+1}^{2}$]{vjj2}
  \Vertex[x=0,y=-3.25,label=$v_{j,j+1}^{1}$]{vjj1}
  \Edge (vjj1)(vjj2)
  \Edge (vjj2)(vjji)
  \Edge (vjji)(vjjk)
  \Edge (vjjk)(b)
  \end{scope}
  \begin{scope}[xshift=4.75cm,yshift=-1.5cm]
  \Vertex[x=0,y=0,label=$\cdots$,style={color=white},size=0.1]{vjii}
  \Vertex[x=0,y=-1,label=$\cdots$,style={color=white},size=0.1]{vjii}
  \Vertex[x=0,y=-2,label=$\cdots$,style={color=white},size=0.1]{vjii}
  \Vertex[x=0,y=-3.25,label=$\cdots$,style={color=white},size=0.1]{vjii}
  \end{scope}
  \begin{scope}[xshift=5.75cm,yshift=-1.5cm]
  \Vertex[x=0,y=0,label=$v_{j,m}^{L}$]{vjmk}
  \Vertex[x=0,y=-1,label=$\cdots$,style={color=white},size=0.4]{vjmi}
  \Vertex[x=0,y=-2,label=$v_{j,m}^{2}$]{vjm2}
  \Vertex[x=0,y=-3.25,label=$v_{j,m}^{1}$]{vjm1}
  \Edge (vjm1)(vjm2)
  \Edge (vjm2)(vjmi)
  \Edge (vjmi)(vjmk)
  \Edge (vjmk)(b)
  \end{scope}
\end{tikzpicture}}
\caption{The part of the bottom hub gadget in the reduction of \tsat{} to symmetric \mrtt{} (see the proof of Theorem~\ref{th:sym-np}), corresponding to the clause $c_j$ which contains the variables $x_{i_1}$, $x_{i_2}$, and $x_{i_3}$ (the dotted edges are included in the clause gadget corresponding to the clause $c_j$).}
\label{fig:symbottomhubgadget}
\end{figure}

\item[\bf Bottom hub gadget] (see Figure~\ref{fig:symbottomhubgadget}). $V$ contains the \textit{bottom hub node} $B$. For each clause $c_j$ of $\Phi$ with $j\in[m]$, and for each variable $x_i$ that appears (positive or negative) in $c_j$, $V$ contains the set  $A_{j,i}^\mathrm{u} = \{u_{j,i}^k : k \in [L]\}$. Moreover, for each $h\in[m]$, such that $h \neq j$, $V$ contains the set $A_{j,h}^\mathrm{v} = \{v_{j,h}^k : k \in [L]\}$. We will refer to the nodes in $A_{j,i}^\mathrm{u}$ and in $A_{j,h}^\mathrm{v}$ as the \emph{bottom tail nodes}. Concerning the edges, $E$ contains the edge $(B, d_j^1)$, and, for each variable $x_i$ that appears (positive or negative) in $c_j$, the set $\{(u_{j,i}^{k},u_{j,i}^{k+1}) : k \in [L-1]\}$ and the edge $(u_{j,i}^L,B)$. Finally, for each $h\in[m]$, such that $h \neq j$, $E$ contains the set $\{(v_{j,h}^{k},v_{j,h}^{k+1}) : k \in [L-1]\}$ and the edge $(v_{j,h}^{L},B)$.

\item[\bf Top hub gadget] (see Figure~\ref{fig:symtophubgadget}) $V$ contains the \textit{top hub node} $U$. For each clause $c_j$ of $\Phi$ with $j\in[m]$, and for each variable $x_i$ that appears (positive or negative) in $c_j$, $V$ contains the set $A_{j,i}^\mathrm{w} = \{w_{j,i}^{k} : k \in [L]\}$. We will refer to the nodes in $A_{j,i}^\mathrm{w}$ as the \emph{top tail nodes}. Concerning the edges, for each variable $x_i$ that appears (positive or negative) in $c_j$, $E$ contains the set $\{(w_{j,i}^{k},w_{j,i}^{k+1}) : k \in [L-1]\}$ and the edge $(w_{j,i}^L,U)$. Finally, if $x_i$ appears positive in $c_j$, $E$ contains the edge $(U,t_{i}^{2})$, while if $x_i$ appears negative in $c_j$, $E$ contains the edge $(U,f_i^2)$. 
\end{description}

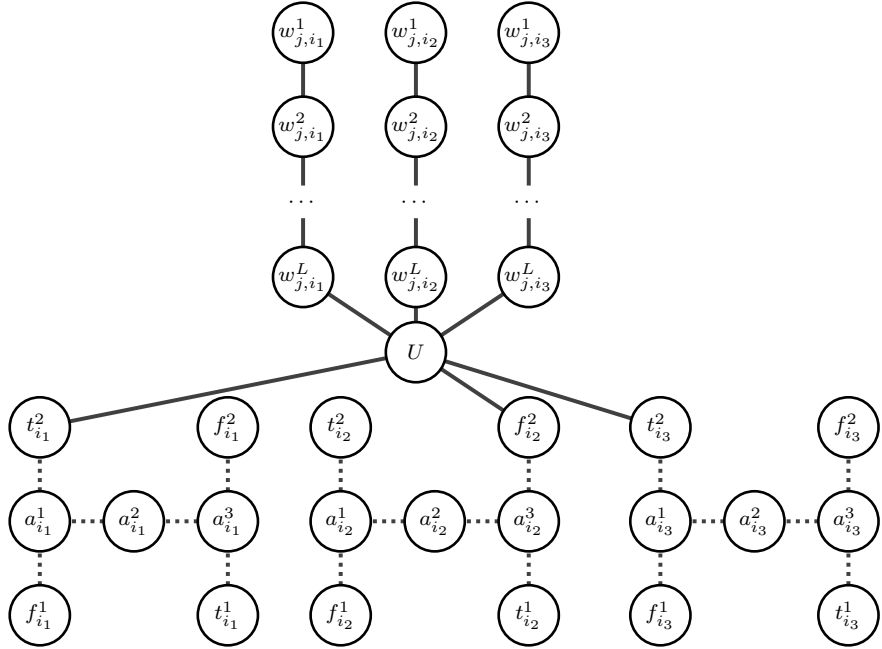
\begin{figure}[ht]
\centering{\SetVertexStyle[FillColor=white,TextFont=\scriptsize,MinSize=0.8\DefaultUnit]
\begin{tikzpicture}
  \Vertex[x=2.5,y=-0.5,label=$t_{i_1}^{1}$]{ti11}
  \Vertex[x=2.5,y=0.75,label=$a_{i_1}^{3}$]{ai13}
  \Vertex[x=2.5,y=2,label=$f_{i_1}^{2}$]{fi12}
  \Vertex[x=1.25,y=0.75,label=$a_{i_1}^{2}$]{ai12}
  \Vertex[x=0,y=0.75,label=$a_{i_1}^{1}$]{ai11}
  \Vertex[x=0,y=-0.5,label=$f_{i_1}^{1}$]{fi11}
  \Vertex[x=0,y=2,label=$t_{i_1}^{2}$]{ti12}
  \Edge[style={dotted}](ti11)(ai13)
  \Edge[style={dotted}](ai13)(fi12)
  \Edge[style={dotted}](ai13)(ai12)
  \Edge[style={dotted}](ai12)(ai11)
  \Edge[style={dotted}](fi11)(ai11)
  \Edge[style={dotted}](ai11)(ti12)
  \begin{scope}[xshift=4cm]
  \Vertex[x=2.5,y=-0.5,label=$t_{i_2}^{1}$]{ti21}
  \Vertex[x=2.5,y=0.75,label=$a_{i_2}^{3}$]{ai23}
  \Vertex[x=2.5,y=2,label=$f_{i_2}^{2}$]{fi22}
  \Vertex[x=1.25,y=0.75,label=$a_{i_2}^{2}$]{ai22}
  \Vertex[x=0,y=0.75,label=$a_{i_2}^{1}$]{ai21}
  \Vertex[x=0,y=-0.5,label=$f_{i_2}^{1}$]{fi21}
  \Vertex[x=0,y=2,label=$t_{i_2}^{2}$]{ti22}
  \Edge[style={dotted}](ti21)(ai23)
  \Edge[style={dotted}](ai23)(fi22)
  \Edge[style={dotted}](ai23)(ai22)
  \Edge[style={dotted}](ai22)(ai21)
  \Edge[style={dotted}](fi21)(ai21)
  \Edge[style={dotted}](ai21)(ti22)
  \begin{scope}[xshift=4.25cm]
  \Vertex[x=2.5,y=-0.5,label=$t_{i_3}^{1}$]{ti31}
  \Vertex[x=2.5,y=0.75,label=$a_{i_3}^{3}$]{ai33}
  \Vertex[x=2.5,y=2,label=$f_{i_3}^{2}$]{fi32}
  \Vertex[x=1.25,y=0.75,label=$a_{i_3}^{2}$]{ai32}
  \Vertex[x=0,y=0.75,label=$a_{i_3}^{1}$]{ai31}
  \Vertex[x=0,y=-0.5,label=$f_{i_3}^{1}$]{fi31}
  \Vertex[x=0,y=2,label=$t_{i_3}^{2}$]{ti32}
  \Edge[style={dotted}](ti31)(ai33)
  \Edge[style={dotted}](ai33)(fi32)
  \Edge[style={dotted}](ai33)(ai32)
  \Edge[style={dotted}](ai32)(ai31)
  \Edge[style={dotted}](fi31)(ai31)
  \Edge[style={dotted}](ai31)(ti32)
  \end{scope}
  \end{scope}
  \begin{scope}[xshift=5cm,yshift=3cm]
  \Vertex[x=0,y=0,label=$U$]{u}
  \Edge (u)(ti12)
  \Edge (u)(fi22)
  \Edge (u)(ti32)
  \end{scope}
  \begin{scope}[xshift=3.5cm,yshift=4cm]
  \Vertex[x=0,y=0,label=$w_{j,i_1}^{L}$]{wji1k}
  \Vertex[x=0,y=1,label=$\cdots$,style={color=white},size=0.4]{wji1i}
  \Vertex[x=0,y=2,label=$w_{j,i_1}^{2}$]{wji12}
  \Vertex[x=0,y=3.25,label=$w_{j,i_1}^{1}$]{wji11}
  \Edge (wji11)(wji12)
  \Edge (wji12)(wji1i)
  \Edge (wji1i)(wji1k)
  \Edge (wji1k)(u)
  \end{scope}
  \begin{scope}[xshift=5cm,yshift=4cm]
  \Vertex[x=0,y=0,label=$w_{j,i_2}^{L}$]{wji1k}
  \Vertex[x=0,y=1,label=$\cdots$,style={color=white},size=0.4]{wji1i}
  \Vertex[x=0,y=2,label=$w_{j,i_2}^{2}$]{wji12}
  \Vertex[x=0,y=3.25,label=$w_{j,i_2}^{1}$]{wji11}
  \Edge (wji11)(wji12)
  \Edge (wji12)(wji1i)
  \Edge (wji1i)(wji1k)
  \Edge (wji1k)(u)
  \end{scope}
  \begin{scope}[xshift=6.5cm,yshift=4cm]
  \Vertex[x=0,y=0,label=$w_{j,i_3}^{L}$]{wji1k}
  \Vertex[x=0,y=1,label=$\cdots$,style={color=white},size=0.4]{wji1i}
  \Vertex[x=0,y=2,label=$w_{j,i_3}^{2}$]{wji12}
  \Vertex[x=0,y=3.25,label=$w_{j,i_3}^{1}$]{wji11}
  \Edge (wji11)(wji12)
  \Edge (wji12)(wji1i)
  \Edge (wji1i)(wji1k)
  \Edge (wji1k)(u)
  \end{scope}
\end{tikzpicture}}
\caption{The part of the top hub gadget in the reduction of \tsat{} to symmetric \mrtt{} (see the proof of Theorem~\ref{th:sym-np}), corresponding to the clause $c_j=x_{i_1}\vee \neg x_{i_2}\vee x_{i_3}$ (the dotted edges are included in the variable gadgets corresponding to the variables $x_{i_1}$, $x_{i_2}$, and $x_{i_3}$).}
\label{fig:symtophubgadget}
\end{figure}

\medskip
\noindent\textbf{Number of nodes.}
Let us first compute the cardinality of $V$. There are $7n$ variable nodes, $2m$ clause nodes, $3m + m(m-1)$ head nodes, $ml$ middle nodes, $2$ hub nodes, $3mL +m(m-1)L $ bottom tail nodes, and $3mL$ top tail nodes. Thus, $|V| = 2 + 7n + m(L+1)(m+4) +ml+ mL$.

\medskip
\noindent\textbf{Trips.}
We now define the trip collection \tnet{} in D (see Figure~\ref{fig:symtrips}). In the following, for each trip $T$ included in \tnet{}, we implicitly assume that the symmetric trip \symtrip{$T$} is also included in \tnet{}.

\begin{description}
\item[\textit{Variable trips}.] For each $i \in [n]$, \tnet{} contains the trips $T_{i}^{\mathrm{t}}=\langle t_i^1, a_i^3, f_i^2\rangle$, $T_{i}^{\mathrm{f}}=\langle f_i^1, a_i^1, t_i^2\rangle$, and $T_{i}^{\mathrm{a}}=\langle a_i^1, a_i^2, a_i^3\rangle$ (see Figure~\ref{fig:symtrips}(a)).

\item[\textit{Bottom-variable trips}.] For each clause $c_j$ of $\Phi$ with $j\in[m]$, if $c_j$ contains the literal $x_i$, \tnet{} contains the trip $T_{j,i}^{\mathrm{u}} = \langle u_{j,i}^1, \dots, u_{j,i}^L$, $B, d_j^1, \dots, d_j^l, c_j^1, t_i^1\rangle$ (see Figure~\ref{fig:symtrips}(b)), while if $c_j$ contains the literal $\neg x_i$, \tnet{} contains the trip $T_{j,i}^{\mathrm{u}} = \langle u_{j,i}^1, \dots, u_{j,i}^L, B, d_j^1, \dots, d_j^l, c_j^1, f_i^1 \rangle$ (see Figure~\ref{fig:symtrips}(c)). 

\item[\textit{Bottom-clause trips}.] For each $(j,h) \in [m]^2$ such that $j \neq h$, \tnet{} contains the trip $T_{j,h}^{\mathrm{v}} = \langle v_{j,h}^1, \dots, v_{j,h}^L, B, d_j^1, \dots, d_j^l, c_j^1, c_h^2, g_{h}^{j}\rangle$ (see Figure~\ref{fig:symtrips}(d)).

\item[\textit{Top trips}.] For each clause $c_j$ of $\Phi$ with $j\in[m]$, if $c_j$ contains the literal $x_i$, \tnet{} contains the trip $T_{j,i}^{\mathrm{w}} = \langle w_{j,i}^1, \dots, w_{j,i}^L, U, t_i^2, c_j^2, e_j^i \rangle$ (see Figure~\ref{fig:symtrips}(e)), while if $c_j$ contains the literal $\neg x_i$, \tnet{} contains the trip $T_{j,i}^{\mathrm{w}} = \langle w_{j,i}^1, \dots, w_{j,i}^L, U, f_i^2, c_j^2, e_j^i \rangle$ (see Figure~\ref{fig:symtrips}(f)).
\end{description}

\begin{figure}[ht]
\centering{\SetVertexStyle[FillColor=white,TextFont=\scriptsize]
\begin{tikzpicture}
  \Vertex[x=-4,y=0,label={(a)},style={color=white},size=0.4]{a}
  \Vertex[x=-4,y=-1,label={(b)},style={color=white},size=0.4]{a}
  \Vertex[x=-4,y=-2,label={(c)},style={color=white},size=0.4]{a}
  \Vertex[x=-4,y=-3,label={(d)},style={color=white},size=0.4]{a}
  \Vertex[x=-4,y=-4,label={(e)},style={color=white},size=0.4]{a}
  \Vertex[x=-4,y=-5,label={(f)},style={color=white},size=0.4]{a}
  \begin{scope}[xshift=-3cm,yshift=0cm]
  \Vertex[x=0,y=0,label=$t_{i}^{1}$]{dj1}
  \Vertex[x=1,y=0,label=$a_{i}^{3}$]{cj1}
  \Vertex[x=2,y=0,label=$f_{i}^{2}$]{cj2}
  \Vertex[x=2.65,y=0,label={$T_{i}^{\mathrm{t}}$},style={color=white},size=0.4]{tit}
  \Edge (dj1)(cj1)
  \Edge (cj1)(cj2)
  \end{scope}
  \begin{scope}[xshift=0.5cm,yshift=0cm]
  \Vertex[x=0,y=0,label=$f_{i}^{1}$]{dj1}
  \Vertex[x=1,y=0,label=$a_{i}^{1}$]{cj1}
  \Vertex[x=2,y=0,label=$t_{i}^{2}$]{cj2}
  \Vertex[x=2.65,y=0,label={$T_{i}^{\mathrm{f}}$},style={color=white},size=0.4]{tif}
  \Edge (dj1)(cj1)
  \Edge (cj1)(cj2)
  \end{scope}
  \begin{scope}[xshift=4cm,yshift=0cm]
  \Vertex[x=0,y=0,label=$a_{i}^{1}$]{dj1}
  \Vertex[x=1,y=0,label=$a_{i}^{2}$]{cj1}
  \Vertex[x=2,y=0,label=$a_{i}^{3}$]{cj2}
  \Vertex[x=2.65,y=0,label={$T_{i}^{\mathrm{a}}$},style={color=white},size=0.4]{tia}
  \Edge (dj1)(cj1)
  \Edge (cj1)(cj2)
  \end{scope}
  \begin{scope}[yshift=-1cm]
  \draw[decorate, decoration={calligraphic brace, raise = 2pt, amplitude = 4pt, mirror}] (5.5,-1.5) --  (5.5,0.5);
  \Vertex[x=6,y=-0.5,label={$T_{j,i}^{\mathrm{u}}$},style={color=white},size=0.4]{tjiu}
  \begin{scope}[xshift=0cm,yshift=-2cm]
  \Vertex[x=0,y=0,label=$B$]{b}
  \Vertex[x=1,y=0,label=$d_{j}^{1}$]{dj1}
  \Vertex[x=2,y=0,label=$\cdots$,style={color=white},size=0.4]{dji}
  \Vertex[x=3,y=0,label=$d_{j}^{l}$]{djl}
  \Vertex[x=4,y=0,label=$c_{j}^{1}$]{cj1}
  \Vertex[x=5,y=0,label=$c_{h}^{2}$]{cj2}
  \Vertex[x=6,y=0,label=$g_{h}^{j}$]{ghj}
  \Vertex[x=6.65,y=0,label={$T_{j,h}^{\mathrm{v}}$},style={color=white},size=0.4]{tjhv}
  \Edge (b)(dj1)
  \Edge (dj1)(dji)
  \Edge (dji)(djl)
  \Edge (djl)(cj1)
  \Edge (cj1)(cj2)
  \Edge (cj2)(ghj)
  \begin{scope}[xshift=-3cm,yshift=0cm]
  \Vertex[x=2,y=0,label=$v_{j,h}^{L}$]{vjhk}
  \Vertex[x=1,y=0,label=$\cdots$,style={color=white},size=0.4]{vjhi}
  \Vertex[x=0,y=0,label=$v_{j,h}^{1}$]{vjh1}
  \Edge (vjh1)(vjhi)
  \Edge (vjhi)(vjhk)
  \Edge (vjhk)(b)
  \end{scope}
  \end{scope}
  \begin{scope}[xshift=0cm,yshift=0cm]
  \Vertex[x=0,y=0,label=$B$]{b}
  \Vertex[x=1,y=0,label=$d_{j}^{1}$]{dj1}
  \Vertex[x=2,y=0,label=$\cdots$,style={color=white},size=0.4]{dji}
  \Vertex[x=3,y=0,label=$d_{j}^{l}$]{djl}
  \Vertex[x=4,y=0,label=$c_{j}^{1}$]{cj1}
  \Vertex[x=5,y=0,label=$t_{i}^{1}$]{cj2}
  \Edge (b)(dj1)
  \Edge (dj1)(dji)
  \Edge (dji)(djl)
  \Edge (djl)(cj1)
  \Edge (cj1)(cj2)
  \begin{scope}[xshift=-3cm,yshift=0cm]
  \Vertex[x=2,y=0,label=$u_{j,i}^{L}$]{vjhk}
  \Vertex[x=1,y=0,label=$\cdots$,style={color=white},size=0.4]{vjhi}
  \Vertex[x=0,y=0,label=$u_{j,i}^{1}$]{vjh1}
  \Edge (vjh1)(vjhi)
  \Edge (vjhi)(vjhk)
  \Edge (vjhk)(b)
  \end{scope}
  \end{scope}
  \begin{scope}[xshift=0cm,yshift=-1cm]
  \Vertex[x=0,y=0,label=$B$]{b}
  \Vertex[x=1,y=0,label=$d_{j}^{1}$]{dj1}
  \Vertex[x=2,y=0,label=$\cdots$,style={color=white},size=0.4]{dji}
  \Vertex[x=3,y=0,label=$d_{j}^{l}$]{djl}
  \Vertex[x=4,y=0,label=$c_{j}^{1}$]{cj1}
  \Vertex[x=5,y=0,label=$f_{i}^{1}$]{cj2}
  \Edge (b)(dj1)
  \Edge (dj1)(dji)
  \Edge (dji)(djl)
  \Edge (djl)(cj1)
  \Edge (cj1)(cj2)
  \begin{scope}[xshift=-3cm,yshift=0cm]
  \Vertex[x=2,y=0,label=$u_{j,i}^{L}$]{vjhk}
  \Vertex[x=1,y=0,label=$\cdots$,style={color=white},size=0.4]{vjhi}
  \Vertex[x=0,y=0,label=$u_{j,i}^{1}$]{vjh1}
  \Edge (vjh1)(vjhi)
  \Edge (vjhi)(vjhk)
  \Edge (vjhk)(b)
  \end{scope}
  \end{scope}
  \begin{scope}[xshift=0cm,yshift=-3cm]
  \draw[decorate, decoration={calligraphic brace, raise = 2pt, amplitude = 4pt, mirror}] (3.5,-1.5) --  (3.5,0.5);
  \Vertex[x=4,y=-0.5,label={$T_{j,i}^{\mathrm{w}}$},style={color=white},size=0.4]{tjiw}
  \Vertex[x=0,y=0,label=$U$]{b}
  \Vertex[x=1,y=0,label=$t_{i}^2$]{dj1}
  \Vertex[x=2,y=0,label=$c_{j}^{2}$]{cj1}
  \Vertex[x=3,y=0,label=$e_{j}^{i}$]{cj2}
  \Edge (b)(dj1)
  \Edge (dj1)(cj1)
  \Edge (cj1)(cj2)
  \begin{scope}[xshift=-3cm,yshift=0cm]
  \Vertex[x=2,y=0,label=$w_{j,i}^{L}$]{vjhk}
  \Vertex[x=1,y=0,label=$\cdots$,style={color=white},size=0.4]{vjhi}
  \Vertex[x=0,y=0,label=$w_{j,i}^{1}$]{vjh1}
  \Edge (vjh1)(vjhi)
  \Edge (vjhi)(vjhk)
  \Edge (vjhk)(b)
  \end{scope}
  \end{scope}
  \begin{scope}[xshift=0cm,yshift=-4cm]
  \Vertex[x=0,y=0,label=$U$]{b}
  \Vertex[x=1,y=0,label=$f_{i}^2$]{dj1}
  \Vertex[x=2,y=0,label=$c_{j}^{2}$]{cj1}
  \Vertex[x=3,y=0,label=$e_{j}^{i}$]{cj2}
  \Edge (b)(dj1)
  \Edge (dj1)(cj1)
  \Edge (cj1)(cj2)
  \begin{scope}[xshift=-3cm,yshift=0cm]
  \Vertex[x=2,y=0,label=$w_{j,i}^{L}$]{vjhk}
  \Vertex[x=1,y=0,label=$\cdots$,style={color=white},size=0.4]{vjhi}
  \Vertex[x=0,y=0,label=$w_{j,i}^{1}$]{vjh1}
  \Edge (vjh1)(vjhi)
  \Edge (vjhi)(vjhk)
  \Edge (vjhk)(b)
  \end{scope}
  \end{scope}
  \end{scope}
\end{tikzpicture}}
\caption{The trips in the reduction of \tsat{} to symmetric \mrtt{} (see the proof of Theorem~\ref{th:sym-np}): (a) variable trips, (b) and (c) bottom-variable trips, (d) bottom-clause trips, and (e) and (f) top trips. Each kind of trip is included in both directions (that is, the trip collection is symmetric).}
\label{fig:symtrips}
\end{figure}
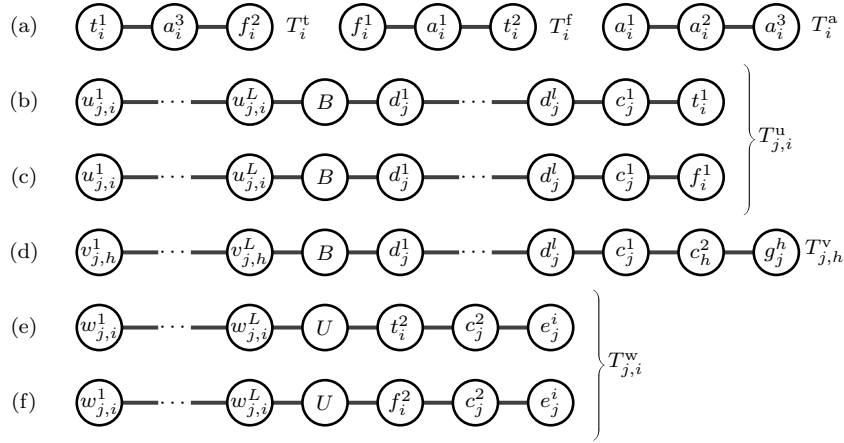

See Figure~\ref{fig:sym-reduction} for a global view of the trip network $(D,\tnet)$.

\begin{figure}[ht]
\centerline{\includegraphics[width=1.2\textwidth]{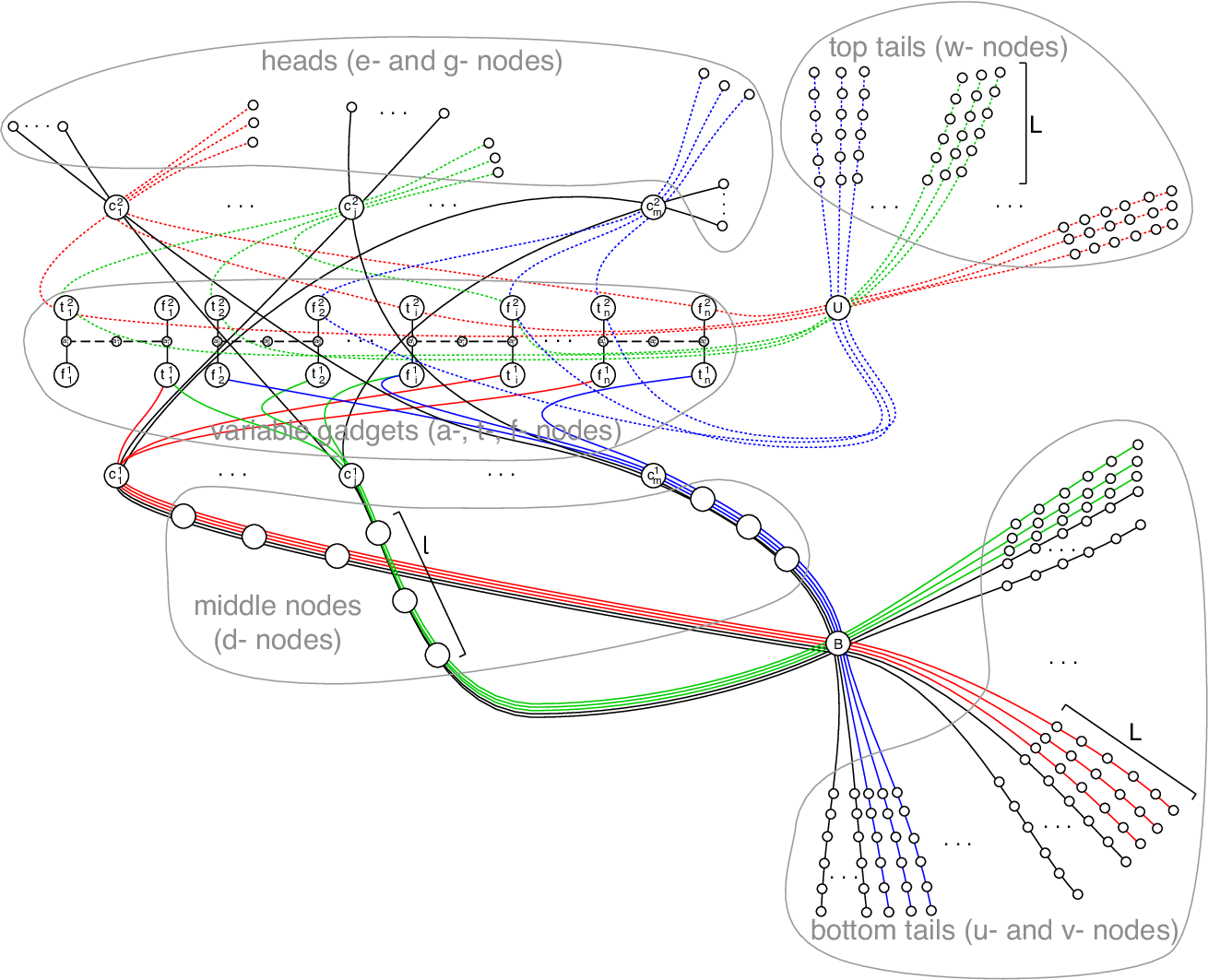}}
\caption{A global view of the reduction of \tsat{} to symmetric \mrtt{}. Here the clause $c_j=x_1\vee  x_2\vee \neg x_i$ is associated with two nodes $c_j^1$ and $c_j^2$ that are connected to variable gadgets for $x_1,x_2$ and $x_i$ through green lines (plain lines for bottom-variable trips and dotted lines for top trips). For a more detailed view of each gadget, see Figures~\ref{fig:symvariablegadget}, \ref{fig:symclausegadget}, \ref{fig:symbottomhubgadget}, \ref{fig:symtophubgadget}, and \ref{fig:symtrips}.}
\label{fig:sym-reduction}
\end{figure}

\medskip
\noindent\textbf{Basic idea of the reduction.} The temporal connections that are made possible by a temporalisation of the variable trips $T_{i}^{\mathrm{t}}$, $T_{i}^{\mathrm{f}}$, and $T_{i}^{\mathrm{a}}$ correspond to choosing whether the variable $x_i$ is set to true or false. Enabling temporal connections from the middle nodes, which are between the bottom hub $B$ and a bottom clause node $c_j^1$, to the head nodes connected to the top clause node $c_j^2$, corresponds to a ``reward'' in terms of reachability for satisfying the clause. The large size of the tails forces some constraints on the temporalisations with high reachability. This ensures that the ``reward'' is obtained only if the clause is satisfied. In particular, we will show that if there exists a satisfying assignment for $\Phi$, then it is possible to produce a temporalisation $\tau$ such that the $\tau$-reachability is at least $Q$, where $Q=|V|^2-(7n+m(m+3))^2$. Otherwise, if $\Phi$ is not satisfiable, then any temporalisation has reachability less than $Q$. More precisely, since $Q>|V|^2-L$ (recall that $L=(7n+m(m+3))^{2}+1$), we will show that any temporalisation that misses a connection from a bottom/top tail node to any node or from a node to any bottom/top tail node has reachability less than $Q$. Similarly, since $Q>|V^2|-(m+2)l$ (recall that $l=\lceil(7n+m(m+3))^2/(m+2)\rceil+1$), we will show that missing the connections from the middle nodes to the head nodes of the corresponding top clause node, also leads to a reachability less than $Q$.

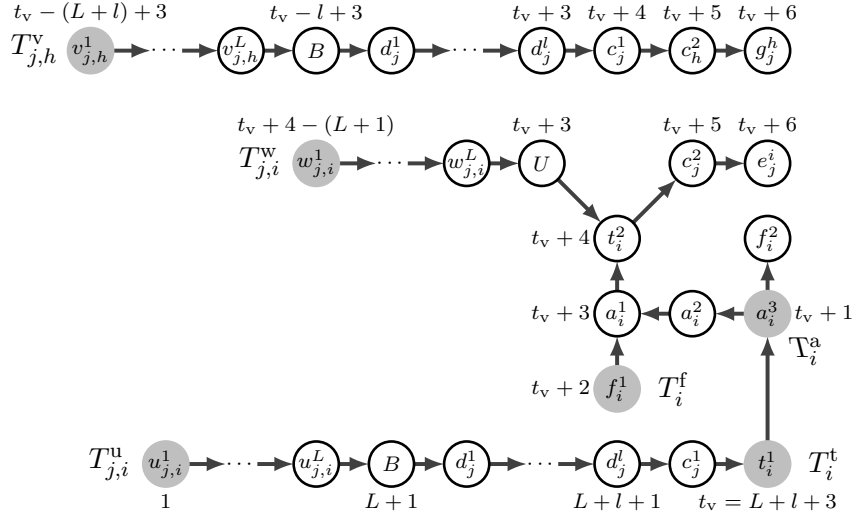
\begin{figure}[ht]
\centering{\SetVertexStyle[FillColor=white,TextFont=\scriptsize]
\begin{tikzpicture}
  \begin{scope}[yshift=-1cm]
  \begin{scope}[xshift=0cm,yshift=-2cm]
  \node at (-3.75,0) {$T_{j,i}^{\mathrm{u}}$};
  \Vertex[x=-3,y=-0.5,label={$1$},style={color=white},size=0.4]{}
  \Vertex[x=0,y=-0.5,label={$L+1$},style={color=white},size=0.4]{}
  \Vertex[x=3,y=-0.5,label={$L+l+1$},style={color=white},size=0.4]{}
  \Vertex[x=5,y=-0.5,label={$t_{\mathrm{v}}=L+l+3$},style={color=white},size=0.4]{}
  \node at (5.75,0) {$T_{i}^{\mathrm{t}}$};
  \Vertex[x=5,y=0,label=$t_{i}^{1}$,style={color=violet},fontcolor=white]{ti1}
  \Vertex[x=0,y=0,label=$B$]{b}
  \Vertex[x=1,y=0,label=$d_{j}^{1}$]{dj1}
  \Vertex[x=2,y=0,label=$\cdots$,style={color=white},size=0.4]{djdots}
  \Vertex[x=3,y=0,label=$d_{j}^{l}$]{djl}
  \Vertex[x=4,y=0,label=$c_{j}^{1}$]{cj1}
  \SetEdgeStyle[Color=red]
  \Edge[Direct](b)(dj1)
  \Edge[Direct](dj1)(djdots)
  \Edge[Direct](djdots)(djl)
  \Edge[Direct](djl)(cj1)
  \Edge[Direct](cj1)(ti1)
  \begin{scope}[xshift=-3cm,yshift=0cm]
  \Vertex[x=2,y=0,label=$u_{j,i}^{L}$]{ujiL}
  \Vertex[x=1,y=0,label=$\cdots$,style={color=white},size=0.4]{ujidots}
  \Vertex[x=0,y=0,label=$u_{j,i}^{1}$,style={color=red},fontcolor=white]{uji1}
  \Edge[Direct](uji1)(ujidots)
  \Edge[Direct](ujidots)(ujiL)
  \Edge[Direct](ujiL)(b)
  \end{scope}
  \begin{scope}[xshift=5cm,yshift=1cm,rotate=90]
  \Vertex[x=1,y=0,label=$a_{i}^{3}$,style={color=yellow},fontcolor=black]{ai3}
  \Vertex[x=2,y=0,label=$f_{i}^{2}$]{fi2}
  \SetEdgeStyle[Color=violet]
  \Edge[Direct](ti1)(ai3)
  \Edge[Direct](ai3)(fi2)
  \end{scope}
  \begin{scope}[xshift=3cm,yshift=2cm]
  \node at (2.5,-0.5) {\reflectbox{$T$}$_{i}^{\mathrm{a}}$};
  \Vertex[x=2.75,y=0,label={$t_{\mathrm{v}}+1$},style={color=white},size=0.4]{}
  \Vertex[x=0,y=0,label=$a_{i}^{1}$]{ai1}
  \Vertex[x=1,y=0,label=$a_{i}^{2}$]{ai2}
  \SetEdgeStyle[Color=yellow]
  \Edge[Direct](ai3)(ai2)
  \Edge[Direct](ai2)(ai1)
  \end{scope}
  \begin{scope}[xshift=3cm,yshift=1cm,rotate=90]
  \node at (0,-0.75) {$T_{i}^{\mathrm{f}}$};
  \Vertex[x=0,y=0.75,label={$t_{\mathrm{v}}+2$},style={color=white},size=0.4]{}
  \Vertex[x=1,y=0.75,label={$t_{\mathrm{v}}+3$},style={color=white},size=0.4]{}
  \Vertex[x=2,y=0.75,label={$t_{\mathrm{v}}+4$},style={color=white},size=0.4]{}
  \Vertex[x=0,y=0,label=$f_{i}^{1}$,style={color=orange}]{fi1}
  \Vertex[x=2,y=0,label=$t_{i}^{2}$]{ti2}
  \SetEdgeStyle[Color=orange]
  \Edge[Direct](fi1)(ai1)
  \Edge[Direct](ai1)(ti2)
  \end{scope}
  \end{scope}
  \begin{scope}[xshift=2cm,yshift=2cm]
  \node at (-3.75,0) {$T_{j,i}^{\mathrm{w}}$};
  \Vertex[x=-3,y=0.5,label={$t_{\mathrm{v}}+4-(L+1)$},style={color=white},size=0.4]{}
  \Vertex[x=0,y=0.5,label={$t_{\mathrm{v}}+3$},style={color=white},size=0.4]{}
  \Vertex[x=2,y=0.5,label={$t_{\mathrm{v}}+5$},style={color=white},size=0.4]{}
  \Vertex[x=3,y=0.5,label={$t_{\mathrm{v}}+6$},style={color=white},size=0.4]{}
  \Vertex[x=0,y=0,label=$U$]{b}
  \Vertex[x=2,y=0,label=$c_{j}^{2}$]{cj2}
  \Vertex[x=3,y=0,label=$e_{j}^{i}$]{ej1}
  \SetEdgeStyle[Color=red]
  \Edge[Direct,style={dashed}](b)(ti2)
  \Edge[Direct,style={dashed}](ti2)(cj2)
  \Edge[Direct,style={dashed}](cj2)(ej1)
  \begin{scope}[xshift=-3cm,yshift=0cm]
  \Vertex[x=2,y=0,label=$w_{j,i}^{L}$]{wjiL}
  \Vertex[x=1,y=0,label=$\cdots$,style={color=white},size=0.4]{wjidots}
  \Vertex[x=0,y=0,label=$w_{j,i}^{1}$,style={color=red},fontcolor=white]{wji1}
  \Edge[Direct,style={dashed}](wji1)(wjidots)
  \Edge[Direct,style={dashed}](wjidots)(wjiL)
  \Edge[Direct,style={dashed}](wjiL)(b)
  \end{scope}
  \end{scope}
  \begin{scope}[xshift=-1cm,yshift=3.5cm]
  \node at (-3.75,0) {$T_{j,h}^{\mathrm{v}}$};
  \Vertex[x=-3,y=0.5,label={$t_{\mathrm{v}}-(L+l)+3$},style={color=white},size=0.4]{}
  \Vertex[x=0,y=0.5,label={$t_{\mathrm{v}}-l+3$},style={color=white},size=0.4]{}
  \Vertex[x=3,y=0.5,label={$t_{\mathrm{v}}+3$},style={color=white},size=0.4]{}
  \Vertex[x=4,y=0.5,label={$t_{\mathrm{v}}+4$},style={color=white},size=0.4]{}
  \Vertex[x=5,y=0.5,label={$t_{\mathrm{v}}+5$},style={color=white},size=0.4]{}
  \Vertex[x=6,y=0.5,label={$t_{\mathrm{v}}+6$},style={color=white},size=0.4]{}
  \Vertex[x=0,y=0,label=$B$]{b}
  \Vertex[x=1,y=0,label=$d_{j}^{1}$]{dj1}
  \Vertex[x=2,y=0,label=$\cdots$,style={color=white},size=0.4]{djdots}
  \Vertex[x=3,y=0,label=$d_{j}^{l}$]{djl}
  \Vertex[x=4,y=0,label=$c_{j}^{1}$]{cj1}
  \Vertex[x=5,y=0,label=$c_{h}^{2}$]{ch2}
  \Vertex[x=6,y=0,label=$g_{h}^{j}$]{ghj}
  \SetEdgeStyle[Color=black]
  \Edge[Direct](b)(dj1)
  \Edge[Direct](dj1)(djdots)
  \Edge[Direct](djdots)(djl)
  \Edge[Direct](djl)(cj1)
  \Edge[Direct](cj1)(ch2)
  \Edge[Direct](ch2)(ghj)
  \begin{scope}[xshift=-3cm,yshift=0cm]
  \Vertex[x=2,y=0,label=$v_{j,h}^{L}$]{vjhL}
  \Vertex[x=1,y=0,label=$\cdots$,style={color=white},size=0.4]{vjhdots}
  \Vertex[x=0,y=0,label=$v_{j,h}^{1}$,style={color=black},fontcolor=white]{vjh1}
  \Edge[Direct](vjh1)(vjhdots)
  \Edge[Direct](vjhdots)(vjhL)
  \Edge[Direct](vjhL)(b)
  \end{scope}
  \end{scope}
  \end{scope}
\end{tikzpicture}}
\caption{The temporalisation of the ``forward'' trips obtained from a truth assignment $\alpha$ satisfying the Boolean formula $\Phi$ (here, we assume that the variable $x_i$ appears positive in the clause $c_j$, and that $\alpha(x_i)=\true$). The gray nodes are the starting nodes of the trips. Note that $T^u_{j,i}$ arrives in $B$ at time $L+1=t_v-l-2$.}
\label{fig:forwardtemporalisation}
\end{figure}

\medskip
\noindent\textbf{Activation of pairs of variable nodes.}
Consider a variable $x_i$ and the associated variable gadget (see Figure~\ref{fig:symvariablegadget}). For a given temporalisation $\tau$, let $t_1=\tau(T_{i}^{\mathrm{f}})$, $t_2=\tau(T_{i}^{\mathrm{a}})$, and $t_3=\tau(T_{i}^{\mathrm{t}})$. Note that $f_i^2$ is $\tau$-reachable from $f_i^1$ by using only the variable trips corresponding to the variable $x_i$ if and only if $t_1+1\leq t_2$ and $t_2+2\leq t_3+1$, as these conditions enable transfers at $a_{i}^{1}$ and $a_{i}^{3}$. When this is the case, we say that $\tau$ \emph{activates} $(f_i^1,f_i^2)$. Note that if $\tau$ \emph{activates} $(f_i^1,f_i^2)$, then $t_1\leq t_3-2$. Similarly, $t_i^2$ is $\tau$-reachable from $t_i^1$ by using only the variable trips corresponding to the variable $x_i$ if and only if $t_3+1\leq t_{4}$ and $t_{4}+2\leq t_1+1$, where $t_4=\tau(\symtrip{$T$}_{i}^{\mathrm{a}})$. When this is the case, we say that $\tau$ activates pair $(t_i^1,t_i^2)$. Note that if $\tau$ \emph{activates} $(t_i^1,t_i^2)$, then $t_3\leq t_1-2$. The key observation for the following is that \textit{no temporalisation can activate both $(f_i^1,f_i^2)$ and $(t_i^1,t_i^2)$}. We similarly say that $\tau$ activates pair $(f_i^2,f_i^1)$ (respectively, $(t_i^2,t_i^1)$) when $f_i^1$ (respectively, $t_i^1$) is $\tau$-reachable from $f_i^2$ (respectively, $t_i^2$), by using only the variable trips corresponding to the variable $x_i$. Once again, no temporalisation can activate both $(f_i^2,f_i^1)$ and $(t_i^2,t_i^1)$. However, one can easily see that it is possible to activate either both $(f_i^1,f_i^2)$ and $(f_i^2,f_i^1)$ or both $(t_i^1,t_i^2)$ and $(t_i^2,t_i^1)$.

\medskip
\noindent\textbf{Constructing a temporalisation from a satisfying assignment.} 
We first show how to construct a temporalisation $\tau$, when $\Phi$ is satisfiable, with reachability at least $Q$. Let $\alpha$ be a truth assignment to $x_1,\ldots,x_n$ that satisfies $\Phi$ (see Figures~\ref{fig:forwardtemporalisation} and~\ref{fig:backwardtemporalisation}, where we assume that $x_{i}$ appears positive in $c_j$ and that $\alpha(x_{i})=\true$).

For any clause $c_j$ and for any variable $x_i$ appearing in $c_j$, we set $\tau(T_{j,i}^{\mathrm{u}})= 1$. Note that $T_{j,i}^{\mathrm{u}}$ arrives in $B$ at time $L+1$, in $c_{j}^{1}$ at time $L+l+2$, and in the variable node connected to $c_{j}^{1}$ and included in the variable gadget corresponding to $x_{i}$ at time $t_{\mathrm{v}} = L+l+3$. For each $i \in [n]$, if $\alpha(x_i)=\true$, then we set $\tau(T_{i}^{\mathrm{t}}) = t_{\mathrm{v}}$, $\tau(\symtrip{$T$}_{i}^{\mathrm{a}}) = t_{\mathrm{v}}+1$, and $\tau(T_{i}^{\mathrm{f}}) = t_{\mathrm{v}}+2$, so that $\tau$ activates $(t_i^1,t_i^2)$ ($t_i^2$ being reachable at time $t_{\mathrm{v}}+4$). Otherwise (that is, $\alpha(x_i)=\false$), we set $\tau(T_{i}^{\mathrm{f}}) = t_{\mathrm{v}}$, $\tau(T_{i}^{\mathrm{a}}) = t_{\mathrm{v}}+1$, and $\tau(T_{i}^{\mathrm{t}}) = t_{\mathrm{v}}+2$ so that $\tau$ activates $(f_i^1,f_i^2)$ ($f_i^2$ being reachable at time $t_{\mathrm{v}}+4$). For any clause $c_j$ and for any variable $x_i$ appearing in $c_j$, we set $\tau(T_{j,i}^{\mathrm{w}}) = t_{\mathrm{v}}+4-(L+1)$, and, for any two clause $c_j$ and $c_h$ with $j\neq h$, we set $\tau(T_{j,h}^{\mathrm{v}}) = t_{\mathrm{v}}-(L+l)+3$: this implies that all these trips reach their top clause node at the same time, that is, $t_{\mathrm{v}}+5$. 

\begin{figure}[ht]
\centering{\SetVertexStyle[FillColor=white,TextFont=\scriptsize]
\begin{tikzpicture}
  \begin{scope}[yshift=-1cm]
  \begin{scope}[xshift=0cm,yshift=-2cm]
  \Vertex[x=-3,y=-0.5,label={$t_{\mathrm{v}}+14+l+L$},style={color=white},size=0.4]{}
  \Vertex[x=0,y=-0.5,label={$t_{\mathrm{v}}+14+l$},style={color=white},size=0.4]{}
  \Vertex[x=3,y=-0.5,label={$t_{\mathrm{v}}+14$},style={color=white},size=0.4]{}
  \Vertex[x=5,y=-0.5,label={$t_{\mathrm{v}}+12$},style={color=white},size=0.4]{}
  \node at (5.75,0) {\reflectbox{$T$}$_{j,i}^{\mathrm{u}}$};
  \Vertex[x=5,y=0,label=$t_{i}^{1}$,style={color=red},fontcolor=white]{ti1}
  \Vertex[x=0,y=0,label=$B$]{b}
  \Vertex[x=1,y=0,label=$d_{j}^{1}$]{dj1}
  \Vertex[x=2,y=0,label=$\cdots$,style={color=white},size=0.4]{djdots}
  \Vertex[x=3,y=0,label=$d_{j}^{l}$]{djl}
  \Vertex[x=4,y=0,label=$c_{j}^{1}$]{cj1}
  \SetEdgeStyle[Color=red]
  \Edge[Direct](dj1)(b)
  \Edge[Direct](djdots)(dj1)
  \Edge[Direct](djl)(djdots)
  \Edge[Direct](cj1)(djl)
  \Edge[Direct](ti1)(cj1)
  \begin{scope}[xshift=-3cm,yshift=0cm]
  \Vertex[x=2,y=0,label=$u_{j,i}^{L}$]{ujiL}
  \Vertex[x=1,y=0,label=$\cdots$,style={color=white},size=0.4]{ujidots}
  \Vertex[x=0,y=0,label=$u_{j,i}^{1}$]{uji1}
  \Edge[Direct](ujidots)(uji1)
  \Edge[Direct](ujiL)(ujidots)
  \Edge[Direct](b)(ujiL)
  \end{scope}
  \begin{scope}[xshift=5cm,yshift=1cm,rotate=90]
  \Vertex[x=2,y=-0.75,label={$t_{\mathrm{v}}+10$},style={color=white},size=0.4]{}
  \node at (2,0.65) {\reflectbox{$T$}$_{i}^{\mathrm{t}}$};
  \Vertex[x=1,y=0,label=$a_{i}^{3}$]{ai3}
  \Vertex[x=2,y=0,label=$f_{i}^{2}$,style={color=violet},fontcolor=white]{fi2}
  \SetEdgeStyle[Color=violet]
  \Edge[Direct](ai3)(ti1)
  \Edge[Direct](fi2)(ai3)
  \end{scope}
  \begin{scope}[xshift=3cm,yshift=2cm]
  \Vertex[x=2.75,y=0,label={$t_{\mathrm{v}}+11$},style={color=white},size=0.4]{}
  \node at (-0.5,0.5) {$T_{i}^{\mathrm{a}}$};
  \Vertex[x=0,y=0,label=$a_{i}^{1}$,style={color=yellow}]{ai1}
  \Vertex[x=1,y=0,label=$a_{i}^{2}$]{ai2}
  \SetEdgeStyle[Color=yellow]
  \Edge[Direct](ai2)(ai3)
  \Edge[Direct](ai1)(ai2)
  \end{scope}
  \begin{scope}[xshift=3cm,yshift=1cm,rotate=90]
  \Vertex[x=0,y=0.75,label={$t_{\mathrm{v}}+10$},style={color=white},size=0.4]{}
  \Vertex[x=1,y=0.75,label={$t_{\mathrm{v}}+9$},style={color=white},size=0.4]{}
  \Vertex[x=2,y=0.75,label={$t_{\mathrm{v}}+8$},style={color=white},size=0.4]{}
  \node at (2,-0.65) {\reflectbox{$T$}$_{i}^{\mathrm{f}}$};
  \Vertex[x=0,y=0,label=$f_{i}^{1}$]{fi1}
  \Vertex[x=2,y=0,label=$t_{i}^{2}$,style={color=orange}]{ti2}
  \SetEdgeStyle[Color=orange]
  \Edge[Direct](ai1)(fi1)
  \Edge[Direct](ti2)(ai1)
  \end{scope}
  \end{scope}
  \begin{scope}[xshift=2cm,yshift=2cm]
  \Vertex[x=-3,y=0.5,label={$t_{\mathrm{v}}+9+L$},style={color=white},size=0.4]{}
  \Vertex[x=0,y=0.5,label={$t_{\mathrm{v}}+9$},style={color=white},size=0.4]{}
  \Vertex[x=2,y=0.5,label={$t_{\mathrm{v}}+7$},style={color=white},size=0.4]{}
  \Vertex[x=3,y=0.5,label={$t_{\mathrm{v}}+6$},style={color=white},size=0.4]{}
  \node at (3.75,0) {\reflectbox{$T$}$_{j,i}^{\mathrm{w}}$};
  \Vertex[x=0,y=0,label=$U$]{b}
  \Vertex[x=2,y=0,label=$c_{j}^{2}$]{cj2}
  \Vertex[x=3,y=0,label=$e_{j}^{i}$,style={color=red},fontcolor=white]{eji}
  \SetEdgeStyle[Color=red]
  \Edge[Direct,style=dashed](ti2)(b)
  \Edge[Direct,style=dashed](cj2)(ti2)
  \Edge[Direct,style=dashed](eji)(cj2)
  \begin{scope}[xshift=-3cm,yshift=0cm]
  \Vertex[x=2,y=0,label=$w_{j,i}^{L}$]{wjiL}
  \Vertex[x=1,y=0,label=$\cdots$,style={color=white},size=0.4]{wjidots}
  \Vertex[x=0,y=0,label=$w_{j,i}^{1}$]{wji1}
  \Edge[Direct,style=dashed](wjidots)(wji1)
  \Edge[Direct,style=dashed](wjiL)(wjidots)
  \Edge[Direct,style=dashed](b)(wjiL)
  \end{scope}
  \end{scope}
  \begin{scope}[xshift=-1cm,yshift=3.5cm]
  \Vertex[x=-3,y=0.5,label={$t_{\mathrm{v}}+9+l+L$},style={color=white},size=0.4]{}
  \Vertex[x=0,y=0.5,label={$t_{\mathrm{v}}+9+l$},style={color=white},size=0.4]{}
  \Vertex[x=3,y=0.5,label={$t_{\mathrm{v}}+9$},style={color=white},size=0.4]{}
  \Vertex[x=4,y=0.5,label={$t_{\mathrm{v}}+8$},style={color=white},size=0.4]{}
  \Vertex[x=5,y=0.5,label={$t_{\mathrm{v}}+7$},style={color=white},size=0.4]{}
  \Vertex[x=6,y=0.5,label={$t_{\mathrm{v}}+6$},style={color=white},size=0.4]{}
  \node at (6.75,0) {\reflectbox{$T$}$_{j,h}^{\mathrm{v}}$};
  \Vertex[x=0,y=0,label=$B$]{b}
  \Vertex[x=1,y=0,label=$d_{j}^{1}$]{dj1}
  \Vertex[x=2,y=0,label=$\cdots$,style={color=white},size=0.4]{djdots}
  \Vertex[x=3,y=0,label=$d_{j}^{l}$]{djl}
  \Vertex[x=4,y=0,label=$c_{j}^{1}$]{cj1}
  \Vertex[x=5,y=0,label=$c_{h}^{2}$]{ch2}
  \Vertex[x=6,y=0,label=$g_{h}^{j}$,style={color=black},fontcolor=white]{ghj}
  \SetEdgeStyle[Color=black]
  \Edge[Direct](dj1)(b)
  \Edge[Direct](djdots)(dj1)
  \Edge[Direct](djl)(djdots)
  \Edge[Direct](cj1)(djl)
  \Edge[Direct](ch2)(cj1)
  \Edge[Direct](ghj)(ch2)
  \begin{scope}[xshift=-3cm,yshift=0cm]
  \Vertex[x=2,y=0,label=$v_{j,h}^{L}$]{vjhL}
  \Vertex[x=1,y=0,label=$\cdots$,style={color=white},size=0.4]{vjhdots}
  \Vertex[x=0,y=0,label=$v_{j,h}^{1}$]{vjh1}
  \Edge[Direct](vjhdots)(vjh1)
  \Edge[Direct](vjhL)(vjhdots)
  \Edge[Direct](b)(vjhL)
  \end{scope}
  \end{scope}
  \end{scope}
\end{tikzpicture}}
\caption{The temporalisation of the ``backward'' trips obtained from a truth assignment $\alpha$ satisfying the Boolean formula $\Phi$ (here, we assume that the variable $x_i$ appears positive in the clause $c_j$, and that $\alpha(x_i)=\true$). The gray nodes are the starting nodes of the trips and $t_{\mathrm{v}} = L+l+3$.}
\label{fig:backwardtemporalisation}
\end{figure}

For any clause $c_j$ and for any variable $x_i$ appearing in $c_j$, we set $\tau(\symtrip{$T$}_{j,i}^{\mathrm{w}}) = t_{\mathrm{v}}+6$, and, for any two clause $c_j$ and $c_h$ with $j\neq h$, we set $\tau(\symtrip{$T$}_{j,h}^{\mathrm{v}}) = t_{\mathrm{v}}+6$: this implies that all these trips reach their top clause node at the same time, that is, $t_{\mathrm{v}}+7$. For each $i \in [n]$, if $\alpha(x_i)=\true$, then we set $\tau(\symtrip{$T$}_{i}^{\mathrm{f}}) = t_{\mathrm{v}}+8$, $\tau(T_{i}^{\mathrm{a}}) = t_{\mathrm{v}}+9$, and $\tau(\symtrip{$T$}_{i}^{\mathrm{t}}) = t_{\mathrm{v}}+10$, so that $\tau$ activates $(t_i^2,t_i^1)$ ($t_i^1$ being reachable at time $t_{\mathrm{v}}+12$). Otherwise (that is, $\alpha(x_i)=\false$), we set $\tau(\symtrip{$T$}_{i}^{\mathrm{t}}) = t_{\mathrm{v}}+8$, $\tau(\symtrip{$T$}_{i}^{\mathrm{a}}) = t_{\mathrm{v}}+9$, and $\tau(\symtrip{$T$}_{i}^{\mathrm{f}}) = t_{\mathrm{v}}+10$ so that $\tau$ activates $(f_i^2,f_i^1)$ ($f_i^1$ being reachable at time $t_{\mathrm{v}}+12$).
For any clause $c_j$ and for any variable $x_i$ appearing in $c_j$, we set $\tau(\symtrip{$T$}_{j,i}^{\mathrm{u}})= t_{\mathrm{v}}+12$. Note that $\symtrip{$T$}_{j,i}^{\mathrm{u}}$arrives in $B$ at time $t_{\mathrm{v}}+14+l$.

We now show that the $\tau$-reachability is at least $Q$.  To this aim, let $G=G[D,\tnet{},\tau]$ be the temporal graph induced by $\tau$, and let $X$ be the set of the following nodes: the top and bottom tail nodes, the middle nodes, the two hub nodes, and the nodes $c_{j}^{1}$ for $j\in[m]$ (note that $V\setminus X$ contains all the head nodes, all the variable nodes, and the nodes $c_{j}^{2}$ for $j\in[m]$, and that, hence, $|V|-|X|=7n+m(m+3)$).

\begin{claim}\label{claim:XtoV}
For any node $x\in X$ and for any node $v\in V$, $v\in\reach{G}{x}$ and $x\in\reach{G}{v}$.
\end{claim}

\begin{proof}
Let us first show that, for any node $x\in X$ and for any node $v\in V\setminus X$, $x\in\reach{G}{v}$. First note that each top clause node $c_j^2$ can be reached at time at most $t_{\mathrm{v}}+7$ by the following set of nodes: $c_j^2$ itself, each head $e_{j}^{i}$ through the trip $\symtrip{$T$}_{j,i}^{\mathrm{w}}$ and each head $g_{j}^{h}$ with $j\neq h$ through the trip $\symtrip{$T$}_{h,j}^{\mathrm{v}}$ (see Figure~\ref{fig:backwardtemporalisation}). Notice that, since each variable appears positive in at least one clause and negative in at least one clause, each variable node of any gadget associated to a variable $x_i$ can reach \textit{some} $c_j^2$ at time $t_{\mathrm{v}}+5$ through the trips $T_{i}^{\mathrm{t}}$, $\symtrip{$T$}_{i}^{\mathrm{a}}$, and $T_{i}^{\mathrm{f}}$ or through the trips $T_{i}^{\mathrm{f}}$, $T_{i}^{\mathrm{a}}$, and $T_{i}^{\mathrm{t}}$ (see Figure~\ref{fig:forwardtemporalisation}). 
Hence, the set of nodes that can reach some $c_j^2$ at time at most $t_{\mathrm{v}}+7$ is equal to $V \setminus X$. Now we show that, by starting from $c_j^2$ at time $t_{\mathrm{v}}+7$, it is possible to reach each node in $X$, thus implying that, for any node $x\in X$ and for any node $v\in V\setminus X$, $x\in\reach{G}{v}$.
\begin{itemize}
    \item The top hub $U$ can be reached at time $t_{\mathrm{v}}+9$ through any trip $\symtrip{$T$}^w_{j,i}$ such that variable $x_i$  appears in $c_j$ (see Figure~\ref{fig:backwardtemporalisation}).

    \item All top tail nodes $w_{p,q}^r$ are reachable through the trips $\symtrip{$T$}_{p,q}^{\mathrm{w}}$, which all ``meet'' in $U$ at time $t_{\mathrm{v}}+9$ (see Figure~\ref{fig:backwardtemporalisation}).
    
    \item For any $h\in[m]$ with $h\neq j$, the bottom clause nodes $c_h^1$, the middle nodes $d_{h}^{r}$ for $r\in[l]$, the bottom hub $B$, and the bottom tail nodes $v_{h,j}^{s}$ for $s\in[L]$ are reachable through the trips $\symtrip{$T$}_{h,j}^{\mathrm{v}}$ (see Figure~\ref{fig:backwardtemporalisation}). Note that all these trips arrive in $B$ at time $t_{\mathrm{v}}+9+l$.
    
    \item Consider a variable $x_i$ appearing in $c_j$ and such that the associated literal has value \true{} according to the assignment $\alpha$ satisfying $\Phi$. The bottom clause node $c_j^1$, the middle nodes $d_{j}^{r}$ for $r\in[l]$, and all bottom tail nodes $u_{j,i}^r$ for $r\in[L]$ are reachable through the trips $\symtrip{$T$}_{i}^{\mathrm{f}}$, $T_{i}^{\mathrm{a}}$, $\symtrip{$T$}_{i}^{\mathrm{t}}$, and $\symtrip{$T$}_{j,i}^{\mathrm{u}}$ or through the trips $\symtrip{$T$}_{i}^{\mathrm{t}}$, $\symtrip{$T$}_{i}^{\mathrm{a}}$,  $\symtrip{$T$}_{i}^{\mathrm{f}}$, and $\symtrip{$T$}_{j,i}^{\mathrm{u}}$ (see Figure~\ref{fig:backwardtemporalisation}).
    
    \item All bottom tail nodes $u_{h,i}^r$ with $h\neq j$ and $r\in[L]$ are reachable through the trips $\symtrip{$T$}_{h,i}^{\mathrm{u}}$, which all ``meet'' in the bottom hub $B$ at time $t_{\mathrm{v}}+14+l$, that is, later than the bottom-clause trips  $\symtrip{$T$}_{h,j}^{\mathrm{v}}$ (see Figure~\ref{fig:backwardtemporalisation}).
\end{itemize}

Let us now prove that, for any node $x\in X$ and for any node $v\in V$, $v\in\reach{G}{x}$. First we prove that, for each $j \in [m]$, the top clause node $c_j^2$ is reachable at time $t_{\mathrm{v}}+5$ from each node in $X$.
\begin{itemize}
    \item The bottom hub $B$ and all the bottom tail nodes $u_{p,q}^r$ and $v_{p,q}^r$ can reach $B$ at time at most $t_{\mathrm{v}}-l+3$ through the trips $T_{p,q}^{\mathrm{u}}$ and $T_{p,q}^{\mathrm{v}}$, and, hence, can reach the top clause node $c_j^2$ at time $t_{\mathrm{v}}+5$ through the trips $T_{p,q}^{\mathrm{v}}$ (see Figure~\ref{fig:forwardtemporalisation}).
    
    \item The top hub $U$ and all the top tail nodes $w_{p,q}^r$ can reach $c_j^2$ at time $t_{\mathrm{v}}+5$ through the trips $T_{p,q}^{\mathrm{w}}$, which all ``meet'' in $U$ at time $t_{\mathrm{v}}+3$ (see Figure~\ref{fig:forwardtemporalisation}).
    
    \item The bottom clause nodes $c_{h}^{1}$, with $h \neq j$, and the middle nodes $d_{h}^{r}$ for $r\in[l]$ can reach the top clause node $c_{j}^{2}$ at time $t_{\mathrm{v}}+5$ through the trips $T_{h,j}^{\mathrm{v}}$ (see Figure~\ref{fig:forwardtemporalisation}).
    
    \item The bottom clause node $c_{j}^{1}$ and the middle nodes $d_{j}^{r}$ for $r\in[l]$ can reach $c_j^2$ at time $t_{\mathrm{v}}+5$ through the trips $T_{i}^{\mathrm{t}}$, $\symtrip{$T$}_{i}^{\mathrm{a}}$, and $T_{i}^{\mathrm{f}}$ or through the trips $T_{i}^{\mathrm{f}}$, $T_{i}^{\mathrm{a}}$, and $T_{i}^{\mathrm{t}}$ (see Figure~\ref{fig:forwardtemporalisation}), where $x_i$ is a variable whose truth assignment satisfies the clause $c_{j}$ (note that the satisfiability of the formula $\Phi$ is also required here).
\end{itemize}
Now we show that, by starting from $c_j^2$ at time $t_{\mathrm{v}}+5$, it is possible to reach the following set of nodes $S_j$: all nodes in $X$ (since we already proved that these nodes are reachable from $c_j^2$, starting at time $t_{\mathrm{v}}+7 > t_{\mathrm{v}}+5$), the head nodes $e_{j}^{r}$ and $g_{j}^{r}$ through the trips $T_{j,r}^{\mathrm{w}}$ and $T_{r,j}^{\mathrm{v}}$, and all variable nodes of the gadget associated to a variable $x_i$ appearing in $c_j$ through the trips $\symtrip{$T$}_{i}^{\mathrm{f}}$, $T_{i}^{\mathrm{a}}$, and $\symtrip{$T$}_{i}^{\mathrm{t}}$ or through the trips $\symtrip{$T$}_{i}^{\mathrm{t}}$, $\symtrip{$T$}_{i}^{\mathrm{a}}$, and $\symtrip{$T$}_{i}^{\mathrm{f}}$ (see Figure~\ref{fig:backwardtemporalisation}). Notice that $\bigcup_{j \in [m]}S_j = V$, and this concludes the proof of the claim.\qed
\end{proof}

We now determine a lower bound on the $\tau$-reachability by counting the number of nodes temporally reachable from different sources.
\begin{itemize}
    \item From the nodes in $X$, it is possible to reach each node in $V$. This adds $|X| \cdot |V|$ to the $\tau$-reachability.
    
    \item From the nodes in $V\setminus X$, it is possible to reach the nodes in $X$. This adds $(|V|-|X|)\cdot |X|$ to the $\tau$-reachability.
\end{itemize}
Hence, the $\tau$-reachability is at least equal to $2|X|\cdot|V|-|X|^2=|V|^2 - (|V|-|X|)^2=|V|^2-(7n+m(m+3))^2=Q$ (recall that $|V|-|X|=7n+m(m+3)$).

\medskip
\noindent\textbf{Bounding reachability when $\Phi$ is not satisfiable.}
Let $\tau$ be any trip temporalisation of the trip network $(D,\tnet{})$ and let $G=G[D,\tnet{},\tau]$ be the temporal graph induced by $\tau$.

\begin{claim}\label{claim:AtoV}
If the $\tau$-reachability is at least equal to $Q$, then, for any $v\in V$ and for any bottom/top tail node $x$, {we have} $v\in\reach{G}{x}$ and $x\in\reach{G}{v}$. 
\end{claim}

\begin{proof}
Without loss of generality, we prove the claim in the case in which $x\in A_{j,i}^{\mathrm{u}}$, for some clause $c_j$ of $\Phi$ with $j\in[m]$ and for some variable $x_i$ that appears (positive or negative) in $c_j$ (the proofs of the other cases are similar). First of all observe that all nodes in $A_{j,i}^{\mathrm{u}}$ have the same reachability set and belong to the same reachability sets. Formally, for any two nodes $u_{j,i}^{r}$ and $u_{j,i}^{s}$ in $A_{j,i}^{\mathrm{u}}$ with $r,s\in[L]$, we have that $\reach{G}{u_{j,i}^{r}}=\reach{G}{u_{j,i}^{s}}$, and that, for any node $v\in V$, $u_{j,i}^{r}\in\reach{G}{v}$ if and only if $u_{j,i}^{s}\in\reach{G}{v}$. This is due to the fact that the bottom-variable trips $T_{j,i}^{\mathrm{u}}$ and $\symtrip{$T$}_{j,i}^{\mathrm{u}}$ are the only trips passing through the nodes in $A_{j,i}^{\mathrm{u}}$. This observation implies that if there exists $v\in V$ such that either $v\not\in\reach{G}{u_{j,i}^{r}}$ or $u_{j,i}^{r}\not\in\reach{G}{v}$ for some bottom tail node $u_{j,i}^{r}$, then the $\tau$-reachability is at most $|V|^2 - L < Q$. Hence, the claim follows.\qed
\end{proof}

Let us now consider the following time constraints that, as a consequence of the above claim, need to be satisfied by $\tau$, if the $\tau$-reachability is at least equal to $Q$. For the sake of brevity, we will give a detailed proof of the first constraint only, since the proofs of the other ones are similar: intuitively, these proofs are based on the fact that the connections provided by some trips between two nodes use the minimum number of edges (that is, they are shortest paths with respect to the number of hops). 
\begin{description}
    \item[C1] For all clauses $c_j$ of $\Phi$ with $j\in[m]$ and for all variables $x_i$ that appear (positive or negative) in $c_j$, all trips $T_{j,i}^{\mathrm{w}}$ are assigned the same starting time $t^{\mathrm{w}}$, that is, $\tau(T_{j,i}^{\mathrm{w}}) = t^{\mathrm{w}}$ (this implies that all these trips reach node $U$ at time $t^{\mathrm{U}} = t^{\mathrm{w}}+L$ and the node $c_{j}^{2}$ at time $t^{\mathrm{U}}+2$). This constraint is needed in order to have any top tail node able to reach any head node at the end of a top trip. Indeed, if there exists two trips $T_{j_{1},i_{1}}^{\mathrm{w}}$ and $T_{j_{2},i_{2}}^{\mathrm{w}}$, with $j_{1},j_{2}\in[m]$, $i_{1},i_{2}\in[n]$, and the variable $x_{i_{1}}$ (respectively, $x_{i_{2}}$) appearing (positive or negative) in the clause $c_{j_{1}}$ (respectively, $c_{j_{2}}$), such that $\tau(T_{j_{1},i_{1}}^{\mathrm{w}})>\tau(T_{j_{2},i_{2}}^{\mathrm{w}})$, since there is no trip connecting $U$ to $c_{j_2}^2$ by using less than two edges, any top tail node in $A_{j_1,i_1}^{\mathrm{w}}$ reaches $c_{j_2}^{2}$ at time $\tau(T_{j_{1},i_{1}}^{\mathrm{w}})+L+2>\tau(T_{j_{2},i_{2}}^{\mathrm{w}})+L+2$, thus implying that it cannot reach the head node $e_{j_2}^{i_2}$ (since the edge $(c_{j_2}^{2},e_{j_2}^{i_2})$ appears at time $\tau(T_{j_{2},i_{2}}^{\mathrm{w}})+L+2$). Because of the previous claim, this contradicts the assumption that the $\tau$-reachability is at least equal to $Q$.
    
    \item[C2] For all clauses $c_j$ of $\Phi$ with $j\in[m]$ and for all variables $x_i$ that appear (positive or negative) in $c_j$, all trips $\symtrip{$T$}_{j,i}^{\mathrm{w}}$ are assigned the same starting time $t^{\mathrm{w,s}}$, that is, $\tau(\symtrip{$T$}_{j,i}^{\mathrm{w}}) = t^{\mathrm{w,s}}$ (this implies that all these trips reach node $U$ at time $t^{\mathrm{U,s}} = t^{\mathrm{w,s}}+3$). This constraint is needed in order to have any head node at the end of a top trip able to reach any top tail node.
    
    \item[C3] $t^{\mathrm{U}}\leq t^{\mathrm{U,s}}$. This constraint is needed in order to have any top tail node in $A_{j_1,i_1}^{\mathrm{w}}$ able to reach any top tail node in $A_{j_2,i_2}^{\mathrm{w}}$, by first using the trip $T_{j_{1},i_{1}}^{\mathrm{w}}$ (in order to reach $U$ at time $t^{\mathrm{U}}$, as stated in~C1) and then using the trip $\symtrip{$T$}_{j_{2},i_{2}}^{\mathrm{w}}$ (which passes through $U$ at time $t^{\mathrm{U,s}}$, as stated in~C2).
    
    \item[C4] For each clause $c_j$ of $\Phi$ with $j\in[m]$ and for any $h\in[m]$ with $h\neq j$, $\tau(T_{j,h}^{\mathrm{v}}) = t^{\mathrm{U}} - L - l$. This constraint is needed in order to have any top tail node able to reach any head node at the end of a bottom-clause trip and any bottom tail node able to reach any head node at the end of a top trip.
    
    \item[C5] For each clause $c_j$ of $\Phi$ with $j\in[m]$ and for any $h\in[m]$ with $h\neq j$, $\tau(\symtrip{$T$}_{j,h}^{\mathrm{v}}) = t^{\mathrm{U,s}} - 3$. This constraint is needed in order to have any head node able to reach both any top tail node and any bottom tail node.
    
    \item[C6] For each clause $c_j$ of $\Phi$ with $j\in[m]$ and for each variable $x_i$ that appears (positive or negative) in $c_j$, $\tau(T_{j,i}^{\mathrm{u}}) \leq t^{\mathrm{U}} - L - l$. This constraint is needed in order to have any bottom tail node able to reach any head node at the end of a bottom-clause trip.
    
    \item[C7] For each clause $c_j$ of $\Phi$ with $j\in[m]$ and for each variable $x_i$ that appears (positive or negative) in $c_j$, $\tau(\symtrip{$T$}_{j,i}^{\mathrm{u}})\geq t^{\mathrm{U,s}} - 2$. This constraint is needed in order to have any head node at the end of a bottom-clause trip able to reach any bottom tail node.
    
    \item[C8] For each clause $c_j$ of $\Phi$ with $j\in[m]$, for each variable $x_i$ that appears (positive or negative) in $c_j$, and for any $h\in[m]$ with $h\neq j$, $\tau(T_{j,h}^{\mathrm{v}}),\tau(T_{j,i}^{\mathrm{u}})\leq t^{\mathrm{U,s}} - L - l - 4$. These constraints are needed in order to have any bottom tail node able to reach any top tail node.
    
    \item[C9] For each clause $c_j$ of $\Phi$ with $j\in[m]$ and for any $h\in[m]$ with $h\neq j$, $t^U +1 \leq \tau(\symtrip{$T$}_{j,h}^{\mathrm{v}})$. This constraint is needed in order to have any top tail node able to reach any bottom-clause tail node.
    
    \item[C10] $t^{\mathrm{U}} +4\leq t^{\mathrm{U,s}}$. This is a consequence of constraints~C5 and~C9. Note that this constraint subsumes constraint~C3.
\end{description}

The above constraints have been derived by using the fact that $Q>|V|^{2}-L$. We now take advantage of the fact that $Q>|V|^{2}-(m+2)l$. Note that, since $\Phi$ is not satisfiable, for any truth-assignment to the variables of $\Phi$, there must exist a clause which is not satisfied by the assignment. Let us then consider the following truth-assignment $\alpha$, which is derived from $\tau$. For each variable gadget corresponding to a variable $x_i$, $\alpha(x_i)=\true$ if $\tau(T_{i}^{\mathrm{t}})\leq \tau(T_{i}^{\mathrm{f}})$, otherwise $\alpha(x_i)=\false$. Note that from the paragraph about the activation of pairs of variable nodes, it follows that if $\alpha(x_i)=\true$ (respectively, $\alpha(x_i)=\false$), then $\tau$ does not activate $(f_i^1,f_i^2)$ (respectively, $(t_i^1,t_i^2)$). Let $c_{j_\alpha}$ be a clause which is not satisfied by $\alpha$. We now show that the middle nodes $d_{j_\alpha}^k$, for $k\in[l]$, cannot reach the head nodes connected to $c_{j_\alpha}^{2}$. Intuitively, the main reason is that $c_{j_\alpha}^{2}$ cannot be reached from the middle nodes through any of the variable gadgets associated to the variables appearing in $c_{j_\alpha}$, as $\tau$ does not activate, in these gadgets, the pair connected to $c_{j_\alpha}^{1}$ and $c_{j_\alpha}^{2}$, and no other temporal path is possible. More formally, let us analyse all possible temporal paths from a middle node $x$ (with $x=d_{j_\alpha}^k$, for some $k\in[l]$) to a head node $h$ connected to $c_{j_\alpha}^{2}$.
\begin{itemize}
    \item Going through a node $c_k^2$ with $k\neq j_\alpha$ is not allowed by the above time constraints (in particular, by the synchronization of forward bottom-clause trips and forward top trips at $c_{j_\alpha}^{2}$ and $c_{k}^{2}$). Indeed, each temporal path from $x$ to $c_{k}^{2}$ reaches $c_{k}^{2}$ using either a top forward trip or a bottom-clause forward trip, which both arrive in $c_{k}^{2}$ at time $t^{\mathrm{U}}+2$ according to time constraints C1 and C4. Although $c_{j_\alpha}^{2}$ might be $\tau$-reachable from $c_k^2$, the arrival time will be greater than $t^{\mathrm{U}}+2$ while the trip to $h$ departs from $c_{j_\alpha}^{2}$ at time $t^{\mathrm{U}}+2$. 

    \item Using any trip to go to $B$ and then reach $h$ through a node $c_k^1$ with $k\neq {j_\alpha}$ is also not possible. Let us suppose we use a backward bottom-variable or bottom-clause trip $T_1$ to go from $x$ to $B$, arrive in $B$ at time $t$, and then use a forward bottom-variable or bottom-clause trip $T_2$ to reach $c_k^1$, and let $t'$ be the time $T_2$ goes through $B$. Clearly, $t' \geq t$. Because of the temporal constraints~C5 and~C7, we have that $t \geq t^{\mathrm{U,s}}+l$, which implies $t \geq t^{\mathrm{U}}+l$ because of constraint~C3. However, because of temporal constraints~C4 and~C6 we have that $t' \leq t^{\mathrm{U}}-l$ contradicting $t\le t'$. 

    \item Going through a variable node $t_i^1$, if we assume that $x_i$ appears positive in $c_{j_\alpha}$, is not possible either. As mentioned before the choice of $c_{j_\alpha}$ implies that $\tau$ does not activate pair $(t_i^1,t_i^2)$ and the path of length four through the variable gadget for $x_i$ from $t_i^1$ to $t_i^2$ is not $\tau$-compatible. Reaching node $U$ and finally node $c_{j_\alpha}^{2}$ is not possible. Indeed, to do this we would need to reach $U$ using a backward top trip, and reach $c_{j_\alpha}^{2}$ from $U$ using a forward top trip, which means $t^{U,s}\leq t^U$. However, this contradicts constraint~C10. We could also consider going from $t_i^1$ to another clause node $c_k^1$ such that $x_i$ also appears positive in $c_k$. Let $\symtrip{$T$}_k$ denote the backward bottom-variable trip which allows us to go from $t_i^1$ to $c_k^1$ and let $t$ denote the time when it arrives at $c_k^1$. The time constraint~C7 then implies that $t=\tau(\symtrip{$T$}_k)+1\ge t^{\mathrm{U,s}}-1$. Let $T$ be a bottom-variable or bottom-clause trip that we use to leave  $c_k^1$ and later reach $h$. $T$ has to arrive at $c_k^1$ at $t' \geq t$. Since $t' = \tau(T) + L + l +1$, from the time constraint~C8 it follows that $t' \leq t^{\mathrm{U,s}}-3$ contradicting $t'\geq t$.
    
    \item {Going through a variable node $f_i^1$, if we assume that $x_i$ appears negative in $c_{j_\alpha}$, is not possible either for similar reasons.}
\end{itemize}
We have thus proved that no head connected to $c_{j_\alpha}^2$ is $\tau$-reachable from any middle node between $B$ and $c_{j_\alpha}^1$: this implies that the reachability of $\tau$ is at most $|V|^2-l(m+2)<Q$. This concludes the proof of the theorem.\qed
\end{proof}

Our last result shows that, in symmetric strongly temporalisable trip networks, there is always a schedule whose reachability is quadratic with respect to the number of nodes. The general idea to prove this result is to find a somewhat central trip, and then schedule trips so that a constant fraction of nodes can reach the central trip and a constant fraction of nodes are reached from the central trip relying on Facts~\ref{obs:one-to-all} and~\ref{obs:all-to-one}. We will find this central trip as a centroid in a weighted tree.

Let us, then, first recall the definition of centroid. Given a node-weighted tree $R$, the weight of $R$ is defined as the sum of the weights of its nodes. We then define a \emph{weighted centroid} of a tree $R$ of weight $K$ as a node $c$ such that the removal of $c$ disconnects $R$ into subtrees of weight $2K/3$ at most. Such a centroid can be found efficiently as stated below.

\begin{lemma}[Folklore]\label{lem:centroid}
Given a node-weighted tree $R$, a centroid node $c$ can be found in linear time. Moreover, if the weight of $R$ is $K$ and the centroid $c$ has weight $2K/3$ at most, then there exists a partition $P_1,P_2$ of its pending subtrees such that both $P_1\cup\{c\}$ and $P_2$ have total weight $2K/3$ at most. Such a partition can be computed at the cost of sorting the subtrees by non-decreasing weight.
\end{lemma}

\begin{proof}
Note that the classical algorithm for finding a centroid in an unweighted tree can easily be adapted to the weighted case. Recall that it consists in starting from any node $v$. If it is not a centroid, then move to a neighbor whose subtree has weight greater than $K/2$ and repeat the test until finding a centroid. The partition $P_1,P_2$ is obtained by trying to add subtrees in $P_1$ one after another by non-decreasing weight and stopping as soon as $P_1\cup\{c\}$ has weight $K/3$ or more.\qed
\end{proof}

\begin{theorem}\label{th:symmetric}
Let $(D,\tnet)$ be a symmetric and strongly temporalisable trip network. Then there exists a schedule $S$ such that the $S$-reachability of $(D,\tnet)$ is at least a fraction $2/9$ of all node pairs. Such a schedule can be computed in polynomial time.
\end{theorem}

\begin{proof}
Without loss of generality, we can assume that all trips in \tnet{} are distinct. If this is not the case, we can keep only one of multiple copies of the same trip: this can only reduce the reachability of the modified trip network. We can then consider trips in pairs $(T,\symtrip{$T$})$, where \symtrip{$T$} is the reverse trip of $T$, and we denote by $\mathbb{TP}$ the set of such pairs.  For any trip $T$, we also denote by $V(T)\subseteq V$ the set of nodes which $T$ (and $\symtrip{$T$}$) passes through. Finally, we assign arbitrarily each node $v\in V$ to a single trip pair $(T,\symtrip{$T$})$ such that $v\in V(T)$, and let $n_{T}$ denote the number of nodes assigned to the pair $(T,\symtrip{$T$})$.

We now define the \textit{transfer} undirected graph $\mathbb{P}=(\mathbb{TP},\mathbb{EP})$, where two trip pairs are connected when they share a node, that is, $\{(T_1,\symtrip{$T$}_1),(T_2,\symtrip{$T$}_2)\}\in\mathbb{EP}$ if and only if $V(T_1)\cap V(T_2)\neq\emptyset$. Let $M$ denote the multidigraph induced by $(D,\tnet)$. According to Corollary~\ref{cor:connected}, $M$ is strongly connected and, hence, $\mathbb{P}$ is connected. We then compute in linear time a weighted spanning tree $R$ of $\mathbb{P}$, where each trip pair $(T,\symtrip{$T$})$ is weighted by the number $n_{T}$. Note that the weight of $R$ is the number $n=|V|$ of nodes in $D$. We then find a centroid $(C,\symtrip{$C$})$ of $R$ according to Lemma~\ref{lem:centroid}.

First suppose that $(C,\symtrip{$C$})$ has weight greater than $2n/3$. Let $S$ be any schedule of $(D,\tnet)$ that starts with $C$ followed by $\symtrip{$C$}$. Then, we have that, for any $u$ and $v$ in $V(C)$, $v$ is $S$-reachable from $u$, that is, the $S$-reachability is greater than $4n^2/9$. The theorem follows.

Conversely, let us suppose that $(C,\symtrip{$C$})$ has weight $2n/3$ at most. According to Lemma~\ref{lem:centroid}, we then consider a partition $P_1,P_2$ of $R\setminus\{(C,\symtrip{$C$})\}$ such that both $P_1\cup\{(C,\symtrip{$C$})\}$ and $P_2$ have weight $2n/3$ at most. For $i=1,2$, let $V(P_{i})$ denote the set of nodes assigned to $T$, for some trip $T$ such that $(T,\symtrip{$T$})\in P_i$. Let $B_1,\ldots,B_{|P_1|}$ be the subtrees in $P_1$, sorted in an arbitrary way. For each $i$ with $i\in[|P_{1}|]$, the subtree $B_{i}$ corresponds to a strongly connected set $V_i$ of $D$. Moreover, $V_i$ must contain a node $u_{i}\in V(C)$ as some trip pair of $B_{i}$ is connected to $(C,\symtrip{$C$})$ in $\mathbb{P}$. We can then define a schedule $S_{i}$ of the trip pairs in $B_i$ according to Fact~\ref{obs:all-to-one} so that, for any node $v$ in $V_i$, $u_i$ is $S_i$-reachable from $v$. Hence, the schedule $S=S_{1},\ldots,S_{|P_1|},C,\symtrip{$C$}$ is such that, for any $u\in V(P_{1})\cup V(C)$ and $c\in V(C)$, $c$ is $S$-reachable from $u$. Similarly, by reasoning on the subtrees in $P_2$, we can extend $S$, so that, for any $u\in V(P_{2})$ and $c\in V(C)$, $u$ is $S$-reachable from $c$. In other words, the final schedule $S$ is such that, for any $u\in V(P_1)\cup V(C)$ and $v\in V(P_{2})$, $v$ is $S$-reachable from $u$.

Let $n_2$ denote the weight of $P_2$, that is, $n_2=|V(P_{2})|$. As both $P_1\cup\{(C,\symtrip{$C$})\}$ and $P_2$ have weight $2n/3$ at most, we have that $n/3\le n_2\le 2n/3$. Hence, the number of pairs of nodes $u$ and $v$ such that $u\in V(P_1)\cup V(C)$ and $v\in V(P_{2})$ is at least $(n-n_2)n_2\ge \frac{2}{9}n^{2}$ for $n/3\le n_2\le 2n/3$. The theorem thus follows.\qed
\end{proof}

\vspace{-1em}

\section{Open problems}

From a theoretical point of view, it would be interesting to close the gap between Theorems~\ref{th:inapprox-mrtt} and~\ref{thm:mrtthard} with regard to the inapproximability of \mrtt{} in a strongly temporalisable network.
From a more applicative point of view, it would be worth exploring other restrictions of the problem where constant approximation is possible. For example, we leave as a future work the study of trip networks with single edge trips as they can be seen as a variation in directed graphs of label-connectivity as defined in \cite{Goebel1991}.

Finally, an interesting generalisation consists of allowing variable waiting times in-between two consecutive edges of a trip. In other words, a temporalisation would then assign an appearance time to each edge of a trip so that the trip becomes a valid temporal walk: each edge appears after the arrival time of the previous edge.

\end{document}